\documentclass[preprint,9pt]{elsarticle}
\usepackage{xpatch}
\xpretocmd{\part}{\setcounter{section}{0}}{}{}
\makeatletter
\def\ps@pprintTitle{%
 \let\@oddhead\@empty
 \let\@evenhead\@empty
 \def\@oddfoot{\centerline{\thepage}}%
 \let\@evenfoot\@oddfoot}
\makeatother
\usepackage{geometry}
\usepackage{floatrow}
\usepackage{subfig}
\usepackage[labelfont=bf,labelsep=period]{caption}

\usepackage{float}
\floatstyle{plaintop}
\restylefloat{table}

\usepackage[colorlinks=true,citecolor=blue]{hyperref}
\usepackage[dvips]{epsfig}
\usepackage{color}
\usepackage{bm}
\usepackage{amssymb}
\usepackage{amsmath}
\usepackage{empheq}
\allowdisplaybreaks[4]
\usepackage{amsthm}
\usepackage{etoolbox}
\makeatletter
\patchcmd{\@endtheorem}{\@endpefalse}{}{}{}
\patchcmd{\endproof}{\@endpefalse}{}{}{}
\makeatother
\usepackage{amsfonts}
\usepackage{xspace}
\usepackage{upgreek}
\usepackage{graphicx}
\usepackage{graphicx}
\usepackage{epstopdf}
\usepackage[table]{xcolor}
\usepackage{longtable}
\usepackage[title]{appendix}
\usepackage{float}
\usepackage{dsfont}
\usepackage{cleveref}
\usepackage{array}
\newcolumntype{M}[1]{>{\centering\arraybackslash}m{#1}}
\newcolumntype{P}[1]{>{\centering\arraybackslash}p{#1}}
\usepackage{multirow}

\usepackage{amssymb, mathtools}
\usepackage[table]{xcolor}

\theoremstyle{plain}
\newtheorem{theorem}{Theorem}
\newtheorem{lemma}[theorem]{Lemma}

\theoremstyle{definition}

\theoremstyle{remark}
\newtheorem{remark}[theorem]{Remark}

\biboptions{sort&compress}
\usepackage{siunitx}
\usepackage{setspace}
\usepackage{esvect}
\usepackage{url}

\DeclareSymbolFont{bbold}{U}{bbold}{m}{n}
\DeclareSymbolFontAlphabet{\mathbbold}{bbold}

\numberwithin{equation}{section}

\usepackage{accents}

\usepackage{mathtools}

\usepackage{color}

\usepackage[normalem]{ulem}
\usepackage{xcolor}
\definecolor{forestgreen}{rgb}{0.33,0.61,0.34}

\newlength\myindention
\DeclareCaptionFormat{myformat}%
{\hspace*{\myindention}#1#2#3}
\DeclareCaptionFormat{center}{{\hfil\textsc{#1}#2\hfil}\par#3}
\begin{document}

\begin{frontmatter}



\title{Fixation dynamics on multilayer networks}


\author[label1]{Ruodan Liu}
\author[label1]{Naoki Masuda\corref{cor1}\fnref{label2,label3}}
\ead{naokimas@gmail.com}
\cortext[cor1]{Corresponding author}
 \address[label1]{Department of Mathematics, State University of New York at Buffalo, Buffalo, NY 14260-2900, USA}
 \address[label2]{Institute for Artificial Intelligence and Data Science, State University of New York at Buffalo, Buffalo, NY 14260-5030, USA}
 \address[label3]{Center for Computational Social Science, Kobe University, Kobe, 657-8501, Japan}

\begin{abstract}
Network structure has a large impact on constant-selection evolutionary dynamics, with which multiple types of fitness (i.e., strength) compete on the network. Here we study constant-selection dynamics on two-layer networks in which the fitness of a node in one layer affects that in the other layer, under birth-death processes and uniform initialization, which are commonly assumed. We show mathematically and numerically that two-layer networks are suppressors of selection, which means that they suppress the effects of the different fitness values among the different types on the final outcomes of the evolutionary dynamics (called fixation probability) relative to the constituent one-layer networks. In fact, many two-layer networks are suppressors of selection relative to the most basic baseline, the Moran process. This result is in stark contrast with the results for conventional one-layer networks for which most networks are amplifiers of selection.
\end{abstract}


\end{frontmatter}

\part*{}

\section{Introduction}\label{sec:introduction}

Evolutionary dynamics is a mathematical modeling framework that allows us to investigate how the composition of different traits in a population changes over time under the assumption that fitter individuals tend to reproduce more often. For example, evolutionary game theory focuses on situations in which the fitness is determined by the game interaction between individuals, such as the prisoner's dilemma game \cite{%
%
Nowak2006book,Sigmund2010book,Perc2017PhyRep}. Another example, which we focus on in the present study, is evolutionary graph theory in which one investigates the effects of network structure and possibly its variation over time on the evolution of traits \cite{Lieberman2005nature, Nowak2006book,Szabo2007PhyRep,Shakarian2012Biosys,Perc2013JRSocInterface}.
In particular, studies of evolutionary games on networks have revealed that the conditions under which cooperation occurs in social dilemma games heavily depend on the network structure and that these conditions can be mathematically derived using random walk theory~\cite{Ohtsuki2006Nature,Allen2017nature}.

Let us consider the constant-selection evolutionary dynamics on networks. In this dynamics, different types are assigned different constant fitness values, each node of the given network is occupied by one of these types, and the different types compete for survival. One can view this dynamics as competition between resident and mutant phenotypes in structured populations or as social dynamics of opinions in which people switch between different opinions when influenced by their neighbors in the network.

A core property of constant-selection evolutionary dynamics on networks is the fixation probability. It is the likelihood that the mutant type initially occupying a single node of the network ultimately fixates, i.e., the mutant type eventually occupies all the nodes of the network, under the assumption that there is no mutation (i.e., the type on any node does not spontaneously change during the dynamics except when influenced by their neighbors). The fixation probability depends on the network structure, the fitness of the mutant type (denoted by $r$) relative to the fitness of the resident type (which is normalized to be 1), and on the initial condition~\cite{Lieberman2005nature,Nowak2006book,Shakarian2012Biosys}. 
The mutant type is more likely to fixate if $r$ is large. The extent to which the fixation probability of the mutant type increases with rising $r$ hinges on the network structure. Some networks are known to be amplifiers of selection. By definition, in a network amplifying selection, a single mutant has a larger fixation probability than the case of the well-mixed population with the same number of nodes, which is equivalent to the so-called Moran process, at any $r>1$, and has a lower fixation probability than the case of the Moran process at any $r<1$. In amplifying networks, the effect of the difference between the mutant and resident types in terms of the fitness (i.e., $r$ versus $1$) is magnified by the network. In contrast, other networks are suppressors of selection such that a single mutant has a lower fixation probability than the case of the Moran process at any $r>1$ and vice versa at any $r<1$. Under a standard assumption of the birth-death process with selection on the birth and uniform initialization, it has been shown that most networks are amplifiers of selection~\cite{Hindersin2015PLoSCB, Cuesta2018PlosOne, Allen2021PlosComputBiol}. Suppressors of selection are rare~\cite{Cuesta2017PlosOne,Cuesta2018PlosOne}.

Studies have shown that the amplifiers of selection under the birth-death process are not necessarily common when we introduce additional factors into evolutionary graph dynamics models, such as the nonuniform initialization~\cite{Adlam2015ProcRSocA,Pavlogiannis2018CommBiol}, directed networks~\cite{Masuda2009JTB}, metapopulation models~\cite{Yagoobi2021SciRep,Marrec2021PRL}, temporal (i.e., time-varying) networks~\cite{Gyan2023JMB}, and hypergraphs~\cite{Liu2023PLoSCB}.
%
%
These results encourage us to study evolutionary dynamics on other extensions of conventional networks with the expectation that the dynamics on them may be drastically different from those on conventional networks.

In the present study, we explore constant-selection evolutionary dynamics on multilayer networks.
Multilayer networks express the situation in which the individuals in a population are pairwise connected by different types of edges, such as different types of social relationships; the same pair of individuals may be directly connected by multiple types of edges~\cite{Kivela2014JCN,Boccaletti2014PhyRep,Bianconi2018book,DeDomenico2023NatPhys}. In evolutionary dynamics on multilayer networks, each layer, corresponding to one type of edge, is a network, and evolutionary dynamics in different network layers are coupled in some manner. This setting has been investigated for evolutionary social dilemma games. See~\cite{Wang2015PhyJB} for a review.
Earlier work considered two-layer networks in which the game interaction occurs in one network layer, and imitation of strategies between players occurs in the other network layer. Cooperation is more enhanced in this model if the edges overlap more heavily between the two layers~\cite{Ohtsuki2007PRL,Ohtsuki2007JTB,Wang2014PRE,Su2019ProcBiolSci}
or under other conditions~\cite{Chen2021NewJPhys, Li2021PhysicaA}
%
%
(but see~\cite{Inaba2023SciRep}). 
When players are assumed to be engaged in game interactions, not just imitation of strategies, in the different layers,
multilayer networks promote cooperation under conditions such as positive degree correlation between two layers~\cite{Duh2019NewJPhys} and asynchronous strategy updating~\cite{Allen2017PhysA}.
Cooperation can thrive in this class model even if each network layer in isolation does not support cooperation~\cite{Su2022nhb}.
However, to the best of our knowledge, constant-selection evolutionary dynamics on multilayer networks have not been studied.


We particularly use two-layer networks. We introduce two models of constant-selection dynamics in multilayer networks, which are analogues of an evolutionary game model in multilayer networks \cite{Su2022nhb}, and we semianalytically calculate the fixation probability of mutants for each network layer for two-layer networks with high symmetry. Using martingale techniques, we also theoretically prove that the complete graph layer and the cycle graph layer in a two-layer network are suppressors of selection, and that the star graph layer and the complete bipartite layer in a two-layer network are more suppressing than the corresponding one-layer network. We numerically show that all the two-layer networks that we have numerically investigated are suppressors of selection, except for the coupled star networks. However, the coupled star networks are more suppressing than the one-layer star graphs. In this manner, we conclude that two-layer networks suppress the effects of selection.

\section{Moran process}\label{moran-process}

The Moran process is a model of stochastic constant-selection evolutionary dynamics in a well-mixed finite population with $N$ individuals. The population consists of two types of individuals, i.e., the resident and mutant, with constant fitness values, $1$ and $r$, respectively. At each time step, an individual is selected as the parent for reproduction with probability proportional to its fitness, and an individual dies uniformly at random. Then, the parent's offspring replaces the dead individual. The fixation probability for a single mutant is given by~\cite{Lieberman2005nature, Nowak2006book}
\begin{equation}\label{fp-moran-bd}
\rho = \frac{1-1/r}{1-1/r^N}.
\end{equation}
Extensions of the Moran process to networks depend on specific update rules to be assumed. The network may be directed or weighted. A major variant of the updating rule that we consider in the present paper is the birth-death process with selection on the birth, or the Bd rule \cite{Masuda2009JTB, Shakarian2012Biosys, Pattni2015ProcRSocA}, which operates as follows. At each time step, an individual is selected as the parent, denoted by $u$, for reproduction with probability proportional to its fitness. This step is the same as in the Moran process. Then, $u$'s type replaces the type of a neighbor of $u$, which is selected with probability proportional to the edge weight between $u$ and itself.
We use the Bd rule because a majority of work on constant-selection evolutionary dynamics on networks does so \cite{Lieberman2005nature,Chalub2016JofDynGames,Askari2015PRE,Monk2014PRSocA,Giakkoupis2016arxiv,Galanis2017JAcm,Pavlogiannis2017SciRep,Pavlogiannis2018CommBiol,Goldberg2019TheorComputSci}. However, death-birth processes also give important insights into constant-selection evolutionary dynamics \cite{Hindersin2015PLoSCB,Kaveh2015RSocOSci,Pattni2015ProcRSocA,Allen2020PLoSCB}, and we briefly examine this with our two-layer network model in section~\ref{m1-db-rule}.

In a directed and weighted network, the edge direction indicates a one-way relationship between the two nodes. A network is an isothermal graph if the weighted in-degree (i.e., sum of the edge weight over all incoming edges to a node) is the same for all nodes. Unweighted regular graphs are examples of isothermal graphs. The fixation probability for an isothermal graph is given by Eq.~\eqref{fp-moran-bd} \cite{Lieberman2005nature, Nowak2006book}.

The fixation probability for a single mutant of the Moran process is $1/N$ at $r=1$ \cite{DonnellyWelsh1983MPCPS,MasudaOhtsuki2009NewJPhys,Broom2010ProcRSocA}. Relative to the Moran process, many networks are either amplifiers or suppressors of selection \cite{Hindersin2015PLoSCB, Cuesta2018PlosOne, Allen2021PlosComputBiol,Lieberman2005nature,Giakkoupis2016arxiv,Galanis2017JAcm,Pavlogiannis2017SciRep,Pavlogiannis2018CommBiol,Goldberg2019TheorComputSci,Cuesta2017PlosOne}. Amplifiers of selection are networks in which the fixation probability is larger than that for the Moran process (i.e.,
Eq.~\eqref{fp-moran-bd}) for any $r>1$ and smaller than that for the Moran process for any $r<1$. Suppressors of selection are networks in which the fixation probability is smaller and larger than that for the Moran process for any $r>1$ and $r<1$, respectively.

\section{Models}

We introduce two models of constant-selection evolutionary dynamics for a population of $N$ individuals in undirected and possibly weighted multilayer networks. The assumption of the undirected network is given for simplicity, and it is straightforward to generalize the following models to the case of directed multilayer networks. We assume a two-layer network as the population structure, whereas it is straightforward to generalize the models to the case of more than two layers. 
Each layer is assumed to be a connected network with $N$ nodes. Each layer represents one of the two types of relationships between individuals such as physical proximity or, in the case of human social networks, online social relationship. We call each node in one layer the replica node; there are $2N$ replica nodes in the entire two-layer network. Each replica node has a corresponding replica node in the other layer. A pair of the corresponding replica nodes, one in each layer, represents an individual (see Figure~\ref{fig:replica-node} for a schematic).
Each edge within a layer represents direct connectivity between two replica nodes in the same layer. Two individuals may be adjacent to each other in both layers, just one layer, or neither layer. For example, two people may directly interact both in person and online, or in only one of the two ways.
 
\begin{figure}[H]
  \centering
  \includegraphics[width=0.9\linewidth]{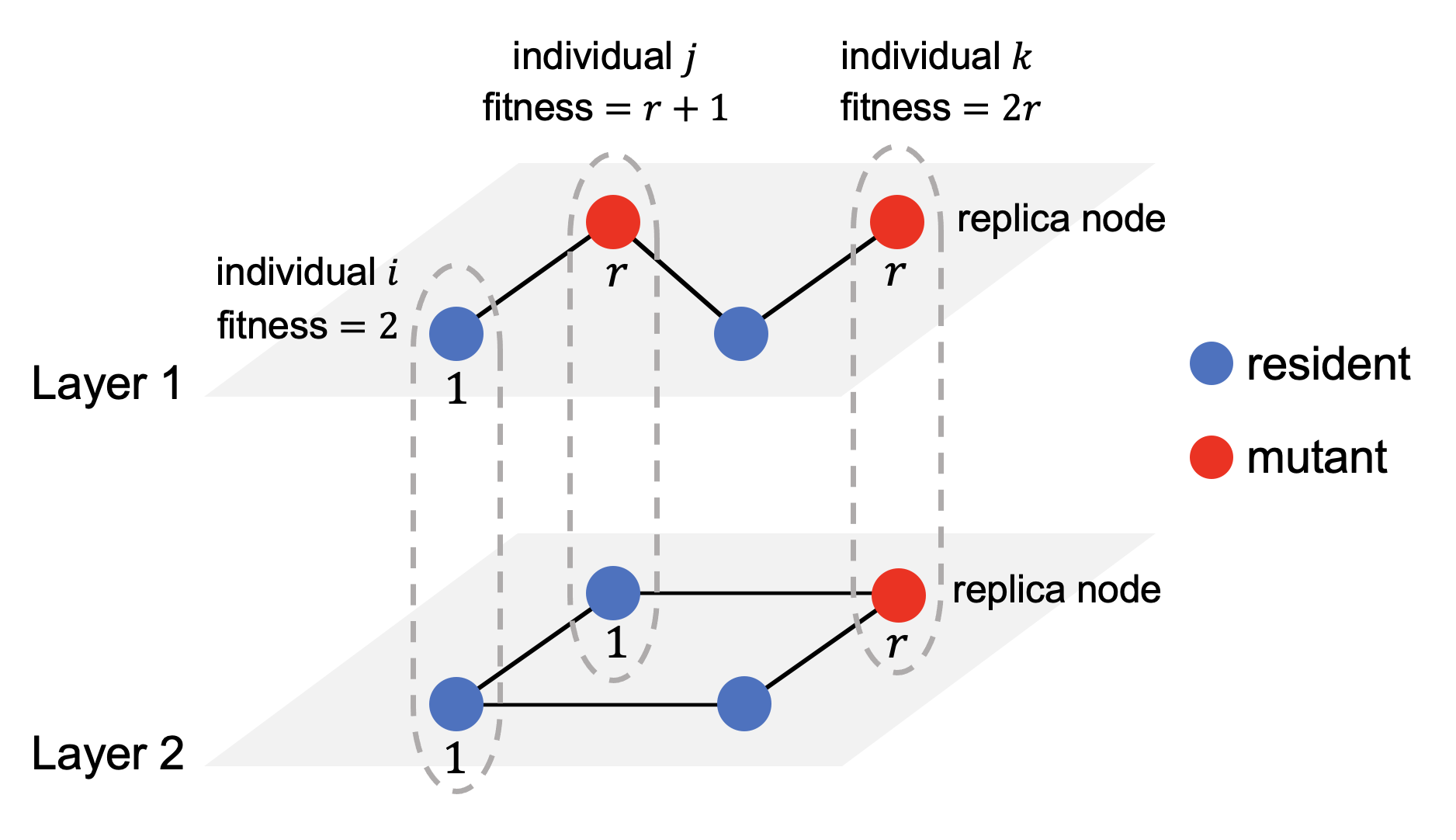}
\caption{An example of a two-layer network. Each individual occupies a replica node in layer 1 and the corresponding replica node in layer 2, as indicated by dashed lines. A resident replica node and a mutant replica node are shown in blue and red, respectively.
}
\label{fig:replica-node}
\end{figure}

Both models extend the Bd process on conventional (i.e., monolayer) networks and the Moran process in well-mixed populations to the case of two-layer networks. We assume that each of the $2N$ replica nodes takes either the resident or mutant type at any discrete time. The resident and mutant have fitness $1$ and $r$, respectively, which are constant over time. We define the fitness of each individual by the sum of the fitness of the corresponding replica nodes in both layers \cite{Su2022nhb}. In other words, the individual has fitness $2$ if it is of the resident type in both layers, $r+1$ if it is of the mutant type in one layer and the resident type in the other layer, and $2r$ if it is of the mutant type in both layers. 
We allow each individual to adopt different types in the opposite layers (i.e., the resident type in one layer and the mutant type in the other layer) because they may behave differently in different types of social relationships. Furthermore, success or failure of an individual in one type of social relationship may affect the same in the other domain, which motivates us to couple the fitness of each individual across the two layers~\cite{Su2022nhb}.

The model assumptions up to this point are shared by models 1 and 2.
Next, in model 1, in each time step, we select one individual (i.e., parent) for reproduction with probability proportional to its fitness. Then, we select one of the two layers to operate the Bd process with the equal probability, i.e., $1/2$. Then, the parent selects one of its neighbors in the selected layer with probability proportional to the weight of the edge between the two individuals. Finally, the parent converts the type of the selected neighbor into the parent's type in the selected layer. This concludes one time step of the Bd process. We repeat this procedure until the entire population settles into an absorbing state in which all individuals are of either resident or mutant type in each layer. It should be noted that the final state in the two layers may be different, i.e., resident in one layer and mutant in the other layer.
This phenomenon may represent the situation in which two opinions or behaviors, $O_1$ and $O_2$, are competing in one layer, and two others, $O_3$ and $O_4$, are competing in the other layer. Then, all individuals may adopt the combination of $O_1$ and $O_3$ in the end or the combination of $O_1$ and $O_4$, for example.

In each time step in model 2, we first select an individual $i$ as the parent with probability proportional to its fitness in each time step. 
This process is the same as that in model 1. However, differently from model 1, we then do not select the layer but draw a neighbor of $i$ in layer 1, denoted by $j$, with probability proportional to the edge weight $w_{ij}^{[1]}$, and $j$ copies $i$'s type. At the same time, we select an individual $k$ as another parent with probability proportional to its fitness. Then, we select a neighbor of $k$ in layer 2, denoted by $\ell$, with probability proportional to the edge weight $w_{k\ell}^{[2]}$, and $\ell$ copies $k$'s type. Individual $k$ may be the same as individual $i$. This model is the same as the main model proposed in \cite{Su2022nhb}, except that the latter model used a death-birth instead of birth-death process, and in the latter model the fitness for each individual was determined by two-player games in their model and therefore was not constant for each type in general. We consider model 2 in addition to model 1 because model 2 is a direct extension of the model proposed in \cite{Su2022nhb}. On the other hand, model 1 is more amenable to mathematical analysis of fixation dynamics than model 2.

\section{Theoretical results}

\subsection{Neutral drift}

In this section, we focus on the case of neutral mutants, i.e., $r=1$.
The fixation probability for the neutral mutant type when there is initially just one mutant node selected uniformly at random must be equal to $1/N$ to enable us to discuss amplifiers and suppressors of selection. We start by proving this property for two-layer networks.

\begin{theorem}
Consider model 1 under $r=1$. When there are initially $i$ mutants selected uniformly at random from the $N$ replica nodes in one layer, the fixation probability for the mutant for that layer is equal to $i/N$. 
\end{theorem}

\begin{proof}
When $r=1$, the fitness of each individual is always equal to 2. Then, the Bd process in layer 1 is independent of that in layer 2. Therefore, the proof is exactly the same as that for conventional networks as shown in \cite{DonnellyWelsh1983MPCPS,MasudaOhtsuki2009NewJPhys,Broom2010ProcRSocA,Liu2023PLoSCB}.
\end{proof}

\begin{remark}
This theorem also holds true for model 2 with the proof being unchanged.
\end{remark}

\subsection{Complete graph layer in a two-layer network is always a suppressor of selection}

In this section, we show that the complete graph layer in an arbitrary two-layer network is always a suppressor of selection under model 1.
To this end, we let $\xi_t \in \{0, 1\}^{2N}$, with $t \in \{0, 1, \ldots \}$, be the state of the Bd process on the two-layer network at time $t$. The initial condition is given by $\xi_0$. We conveniently define $t$ as the number of the state changes in layer 1, which we assumed to be the complete graph. In other words, when counting $t$, we ignore the updating steps in which a replica node in layer 1 is selected as the parent but does not induce the actual change of the state of the network (because the child node has the same type as that of the parent) or a replica node in layer 2 is selected as the parent (because there is then no change in the state in layer 1). 
We consider model 1 in the following text unless stated otherwise.

\begin{lemma}\label{yn-submartingale}
Consider the Bd process on the two-layer network in which layer 1 is the unweighted complete graph and layer 2 is an arbitrary connected network. We let $X_t$ be the number of mutants in the first layer at time $t$ and set $Y_t \equiv r^{-X_t}$. Then, sequence $\{Y_n\}$ is a submartingale for any $r>0$.
\end{lemma}

\begin{proof}
Let $\{\mathcal{B}_t\}$ be the filtration, i.e., an increasing sequence of the $\sigma$-algebras, generated by the Bd process on the two-layer network.
We obtain $X_{t+1} = X_t + 1$ or $X_{t+1} = X_t -1$ because we count the time $t$ if and only if the number of the mutants changes in the complete graph layer. For an arbitrary state of the two-layer network with $X_t$ mutants, $\xi_t$, we denote by $p(\xi_t)$ and $q(\xi_t)$ the probabilities with which $X_{t+1} = X_t + 1$ and $X_{t+1} = X_t - 1$, respectively. Note that $p(\xi_t) + q(\xi_t) = 1$.

To calculate $p(\xi_t)$ and $q(\xi_t)$, we denote by $N_1$ the number of individuals that have the mutant type in both layers, by $N_2$ the number of individuals that have the mutant type in layer 1 and the resident type in layer 2, by $N_3$ the number of individuals that have the resident type in layer 1 and the mutant type in layer 2, and by $N_4$ the number of individuals that have the resident type in both layers. Note that $N_1 + N_2 + N_3 + N_4 = N$. In a single time step of the original Bd process, $X_t$ increases by one with probability
\begin{equation}
p' = \frac{2r N_1 + (r+1) N_2}{2r N_1 + (r+1)(N_2 + N_3) + 2N_4} \cdot \frac{1}{2} \cdot \frac{N_3 + N_4}{N_1 + N_2 + N_3 + N_4 -1}
\label{eq:p'}
\end{equation}
and decreases by one with probability
\begin{equation}
q' = \frac{(r+1)N_3 + 2 N_4} {2r N_1 + (r+1)(N_2 + N_3) + 2N_4} \cdot \frac{1}{2} \cdot \frac{N_1 + N_2}{N_1 + N_2 + N_3 + N_4 -1}.
\label{eq:q'}
\end{equation}  
By combining Eqs.~\eqref{eq:p'} and \eqref{eq:q'} with $p(\xi_t)/q(\xi_t) = p'/q'$ and $p(\xi_t) + q(\xi_t) = 1$, we obtain
\begin{align}
p(\xi_t) =& \frac{r}{r+1} - \varepsilon,\\
q(\xi_t) =& \frac{1}{r+1} + \varepsilon,
\end{align}
where
\begin{equation}
\varepsilon = \frac{(r-1)\left[r N_1 N_3 + (r+1) N_2 N_3 + N_2 N_4\right]}
{(r+1) \left\{ \left[ 2rN_1 + (r+1)N_2\right](N_3+N_4) + \left[ (r+1) N_3 + 2N_4 \right](N_1+N_2) \right\}}.
\end{equation}

We obtain
\begin{align}
E[Y_{t+1} | \mathcal{B}_t] =& p(\xi_t) r^{-(X_t+1)} + q(\xi_t) r^{-(X_t-1)} \notag\\
=& \left[\left(\frac{r}{r+1} - \varepsilon \right) \frac{1}{r} + \left( \frac{1}{r+1} + \varepsilon \right) r \right] Y_t \notag\\
=& \left[ 1 + \left( r - \frac{1}{r} \right) \varepsilon \right] Y_t,
\label{eq:martingale-eq}
\end{align}
where $E[\cdot | \cdot]$ represents the conditional expectation.
If $r>1$, we obtain $E[Y_{t+1} | \mathcal{B}_t] \ge Y_t$ because $r-r^{-1} > 0$ and $\varepsilon \ge 0$.
If $r<1$, we also obtain $E[Y_{t+1} | \mathcal{B}_t] \ge Y_t$ because $r-r^{-1} < 0$ and $\varepsilon \le 0$.
Therefore, in both cases, $Y_t$ is a submartingale.
If $r=1$, we obtain $\varepsilon = 0$ such that $Y_t$ is a martingale, which is a submartingale.
\end{proof}

\begin{remark}
Our choice of $Y_t$ is inspired by the construction of a martingale for the biased random walk on $\mathbb{Z}$
(see, e.g., \cite{Feller1971book2,Durrett1996book}) and its application to constant-selection evolutionary dynamics \cite{Monk2014PRSocA,Adlam2015ProcRSocA,Monk2018JofTheorBiol}. 
\end{remark}

\begin{theorem}\label{complete-thm}
Consider the Bd process on the two-layer network in which layer 1 is the unweighted complete graph and layer 2 is an arbitrary connected network. Then, the complete graph layer is a suppressor of selection.
\end{theorem}

\begin{proof}
Sequence $\{Y_t\}$ is a submartingale and bounded because $r^{-N} \le Y_t \le 1$ $\forall t$ when $r \ge 1$ and $1\le Y_t \le r^{-N}$ $\forall t$ when $r\le 1$. Therefore, $Y_t$ converges almost surely, and $E[Y_{\infty}]$ is finite owing to the martingale convergence theorem~\cite{Feller1971book2,Durrett1996book}. The present Bd process has four absorbing states in which all the nodes in each layer are unanimously occupied by the resident or mutant. The two absorbing states in which all the nodes in layer 1 are occupied by the resident yields $X_t=0$. The other two absorbing states in which all the nodes in layer 1 are occupied by the mutant yields $X_t=N$.
Because an absorbing state is ultimately reached with probability $1$,
\begin{equation}
E[Y_{\infty}] \ge Y_0
\label{eq:martingale-convergence}
\end{equation}
yields
\begin{equation}
x(\xi_0) r^{-N} + \left[1-x(\xi_0)\right] r^{-0} \ge r^{- X_0},
\label{eq:stopping-time}
\end{equation}
where $x(\xi_0)$ is the fixation probability of the mutant under an initial condition $\xi_0$ with $X_0$ mutants in layer 1; therefore, there are initially $N-X_0$ residents in the same layer. Equation~\eqref{eq:stopping-time} yields
\begin{equation}
\begin{cases}
x(\xi_0) \le \frac{1-r^{-X_0}}{1-r^{-N}} & (r\ge 1),\\[1mm]
x(\xi_0) \ge \frac{1-r^{-X_0}}{1-r^{-N}} & (r < 1).
\end{cases}
\label{eq:not-amplifier}
\end{equation}

Our goal is to exclude the equalities in Eq.~\eqref{eq:not-amplifier} for $1\le X_0 \le N-1$ because then it will hold true that the complete graph layer is a suppressor of selection. To show this, we distinguish among the following three cases.

To state the first case, we note that
$\varepsilon = 0$ for $r\neq 1$ and $1\le X_0 \le N-1$ if and only if $N_2 = N_3 = 0$. Therefore, if the initial condition $\xi_0$ satisfies
$N_2 > 0$ or $N_3 > 0$, then
Eq.~\eqref{eq:martingale-eq} implies that
\begin{equation}
E[Y_1 | \xi_0] > Y_0
\label{eq:submartingale-exact-inequality}
\end{equation}
for $r\neq 1$. By combining
$E[Y_2 | \mathcal{B}_t] \ge Y_1$, which follows from
Lemma~\ref{yn-submartingale}, with Eq.~\eqref{eq:submartingale-exact-inequality}, we obtain
$E[Y_2 | \xi_0] > Y_0$.

The second and third cases concern the initial condition $\xi_0$ satisfying $N_2 = N_3 = 0$ such that each individual has the same type (i.e., resident or mutant) in both layers. Then, we obtain 
$E[Y_1 | \xi_0 ] = Y_0$ because $\varepsilon = 0$. 
As the second case, we consider the situation in which $\xi_0$ satisfies $N_1 \le N-2$
in addition to $N_2 = N_3 = 0$. In this case, the network's state after the first state transition, $\xi_1$, satisfies $(N_1, N_2, N_3, N_4) = (N_1, 1, 0, N-N_1 - 1)$ with probability $p(\xi_0)=r/(r+1)$.
Conditioned on this transition, we obtain $E[Y_2 | \xi_1] > Y_1$ for $r\neq 1$, which is an adaptation of
Eq.~\eqref{eq:submartingale-exact-inequality}. We obtain $E[Y_2 | \xi_1] > Y_1$ because this particular $\xi_1$ yields $N_2 = 1$, which implies $(r - r^{-1}) \varepsilon > 0$.
If we start from the same $\xi_0$, and a different $\xi_1$ is realized with probability $1-p(\xi_0)$, we still obtain
$E[Y_2 | \xi_1] \ge Y_1$ owing to Lemma~\ref{yn-submartingale}. Therefore, we obtain
$E[Y_2 | \xi_0] > Y_0$ when $\xi_0$ satisfies 
$N_1 \le N-2$ and $N_2 = N_3 = 0$.
 
As the third case, we consider the situation in which $\xi_0$ satisfies $N_1 = N-1$, which implies that
$N_2 = N_3 = 0$. In this case, $\xi_1$ satisfies $(N_1, N_2, N_3, N_4) = (N_1 - 2, 0, 1, N-N_1)$ with probability $q(\xi_0)=1/(r+1)$.
Conditioned on this transition, we obtain $E[Y_2 | \xi_1] > Y_1$ for $r\neq 1$ because this particular $\xi_1$ yields $N_3 = 1$, which implies $(r - r^{-1}) \varepsilon > 0$.
If we start from the same $\xi_0$, and a different $\xi_1$ is realized with probability $1-q(\xi_0)$, we still obtain
$E[Y_2 | \xi_1] \ge Y_1$ owing to Lemma~\ref{yn-submartingale}. Therefore, we obtain
$E[Y_2 | \xi_0] > Y_0$ when $\xi_0$ satisfies 
$N_1 = N-1$.

Because $E[Y_2 | \xi_0] > Y_0$ holds true in all three cases, we obtain
$E[Y_2 | \mathcal{B}_0] > Y_0$, which, together with
$E[Y_{t+1} | \mathcal{B}_t] \ge Y_t$ $\forall t \in \{2, 3, \ldots \}$, leads to
Eq.~\eqref{eq:martingale-convergence} with the strict inequality. Therefore, Eqs.~\eqref{eq:stopping-time} and \eqref{eq:not-amplifier} hold true with the strict inequality when $r\neq 1$.

\end{proof}

\subsection{Cycle graph layer in a two-layer network is always a suppressor of selection}

We use the same method as that for the complete graph layer to show that~\Cref{yn-submartingale} also holds true when one replaces the complete graph layer by the cycle graph. The cycle graph, which we assumed to form layer 1, is defined by $w_{ij}^{[1]} = 1$ if $j = i \pm 1 \mod N$, and $w_{ij}^{[1]} = 0$ otherwise.
For simplicity, we assume that the replica nodes of the mutant type are initially consecutive (i.e., forming just one connected component of mutants) in the cycle graph layer. 
\begin{lemma}\label{yn-submartingale-cycle}
Consider the Bd process on the two-layer network in which layer 1 is the unweighted cycle graph and layer 2 is an arbitrary connected network. We let $X_t$ be the number of mutants in layer 1 at time $t$ and set $Y_t \equiv r^{-X_t}$. The individuals of the mutant type are assumed to be initially located at consecutive replica nodes on the cycle. Then, sequence $\{Y_n\}$ is a submartingale for any $r>0$.
\end{lemma}
We prove Lemma~\ref{yn-submartingale-cycle} in section~S1. 
\begin{theorem}\label{cycle-thm}
Consider the Bd process on the two-layer network in which layer 1 is the unweighted cycle graph and layer 2 is an arbitrary connected network. Then, the cycle graph layer is a suppressor of selection, given that the individuals of the mutant type are initially located at consecutive replica nodes on the cycle.
\end{theorem}
We prove Theorem~\ref{cycle-thm} in section~S2. 

\subsection{Complete bipartite graph layer in a two-layer network}

In this section, we consider the two-layer network in which layer 1 is the unweighted complete bipartite graph and layer 2 is an arbitrary connected network. The complete bipartite graph, denoted by $K_{N_1,N_2}$, where $N_1+N_2=N$, consists of two disjoint subsets of nodes $V_1$ and $V_2$ with $N_1$ and $N_2$ nodes, respectively. It is defined by $w_{ij}^{[1]}=1$ if $i\in V_1$ and $j\in V_2$, or $i\in V_2$ and $j\in V_1$, and by $w_{ij}^{[1]}=0$ otherwise. We construct a similar proof to that for the complete graph or cycle graph layer to show that the complete bipartite graph layer in an arbitrary two-layer network is more suppressing than the one-layer complete bipartite graph.

\begin{lemma}\label{yn-submartingale-bipartite}
Consider the Bd process on the two-layer network in which layer 1 is the unweighted complete bipartite graph and layer 2 is an arbitrary connected network. We let $\bm{X_t}=[X_{1,t}, X_{2,t}]$, where $X_{1,t}$ and $X_{2,t}$ are the numbers of replica nodes in $V_1$ and $V_2$, respectively, that are occupied by the mutant in layer 1 at time $t$. We define $Y_t \equiv h_1^{X_1,t}h_2^{X_2,t}$, where
\begin{align}
h_1&=\frac{N_1+N_2r}{N_1r^2+N_2r},\\
h_2&=\frac{N_2+N_1r}{N_2r^2+N_1r}.
\end{align}
Then, sequence $\{Y_n\}$ is a submartingale for any $r>0$.
\end{lemma}
We prove Lemma~\ref{yn-submartingale-bipartite} in section~S3. 
\begin{remark}
Our choice of $Y_t$ is inspired by the application of martingales to the Bd process in one-layer complete bipartite graphs \cite{Monk2014PRSocA}.
\end{remark}

\begin{theorem}\label{complete-bi-thm}
Consider the Bd process on the two-layer network in which layer 1 is the unweighted complete bipartite graph and layer 2 is an arbitrary connected network. Then, the complete bipartite graph layer is more suppressing than the one-layer complete bipartite graph.
\end{theorem}
\begin{proof}
Equation~\eqref{eq:martingale-convergence} holds true in the present case as well. It is equivalent to
\begin{equation}
x(\xi_0) h_1^{N_1}h_2^{N_2} + \left[1-x(\xi_0)\right] h_1^0h_2^0 \ge h_1^{X_{1,0}}h_2^{X_{2,0}},
\label{eq:stopping-time-bipartite}
\end{equation}
where $x(\xi_0)$ is the fixation probability of the mutant type under an initial condition $\xi_0$ with $X_{1,0}$ mutants on the nodes in $V_1$ and $X_{2,0}$ mutants on the nodes in $V_2$;
there are initially $X_{1,0}+X_{2,0}$ mutants and $N-(X_{1,0}+X_{2,0})$ residents in the complete bipartite graph layer.
Equation~\eqref{eq:stopping-time-bipartite} yields 
\begin{equation}
\begin{cases}
x(\xi_0) \le \frac{h_1^{X_{1,0}}h_2^{X_{2,0}}-1}{h_1^{N_1}h_2^{N_2}-1} & (r\ge 1),\\
\\[-5pt]%
x(\xi_0) \ge \frac{h_1^{X_{1,0}}h_2^{X_{2,0}}-1}{h_1^{N_1}h_2^{N_2}-1} & (r < 1).
\end{cases}
\label{eq:not-amplifier-bipartite}
\end{equation}
To exclude the equalities in Eq.~\eqref{eq:not-amplifier-bipartite} for $1\le X_{1,0}+X_{2,0}\le N-1$, we distinguish 12 cases that are different in terms of the number of individuals in $V_1$ and in $V_2$ with different fitness values. We obtain
\begin{equation}\label{eq:submartingale-exact-inequality-bipartite}
E[Y_{3} | \xi_0]>Y_0
\end{equation}
for all 12 cases; for the proof, see section S4.

Because Eq.~\eqref{eq:submartingale-exact-inequality-bipartite} holds true in all the cases, we obtain
$E[Y_3 | \mathcal{B}_0] > Y_0$, which, together with
$E[Y_{t+1} | \mathcal{B}_t] \ge Y_t$ $\forall t \in \{3, 4, \ldots \}$, leads to
Eq.~\eqref{eq:martingale-convergence} with the strict inequality. Therefore, Eqs.~\eqref{eq:stopping-time-bipartite} and \eqref{eq:not-amplifier-bipartite} hold true with the strict inequality when $r\neq 1$, proving that the complete bipartite graph layer in a two-layer network is more suppressing than the monolayer complete bipartite graph. 
\end{proof}

\begin{remark}\label{star-thm}
If $N_1=1$ and $N_2=N-1$, the complete bipartite graph layer reduces to a star graph. Therefore, Lemma~\ref{yn-submartingale-bipartite} and Theorem~\ref{complete-bi-thm} also hold true when one layer of the two-layer network is a star graph.
\end{remark}

\begin{remark}
All lemmas and theorems also hold true for model 2 with the proof being essentially unchanged (for more details, see section S4).
\end{remark}

\section{Semi-analytical results for two-layer networks with high symmetry}

\subsection{Exact computation of the fixation probability in two-layer networks\label{fp-general-case}}

\sloppy In this section, we explain how to exactly calculate the fixation probability for the mutant type when there is initially one replica node of mutant type that is selected uniformly at random in layer 1, and one replica node of mutant type in layer 2. This initial state is the same as that assumed in \cite{Su2022nhb}. Let $s_i^{[1]}\in \{0, 1\}$ and $s_i^{[2]}\in \{0, 1\}$ be individual $i$'s type in layer 1 and layer 2, respectively, where values 0 and 1 represent resident and mutant, respectively. Then, the state of the evolutionary dynamics is specified by a $2N$-dimensional binary vector $\bm{s}=(s_1^{[1]}, \ldots, s_N^{[1]}, s_1^{[2]}, \ldots, s_N^{[2]})$. Therefore, there are $2^{2N}$ states in total. We number the states from 1 to $2^{2N}$ by a bijective map, denoted by $\varphi$, given by
\begin{align}
\varphi: S&\rightarrow \{1, \ldots, 2^{2N}\},\nonumber\\
\bm{s}&\mapsto \varphi(\bm{s}),
\end{align}
where $S$ is the set of all states.
Let $P=[p_{i,j}]$ denote the $2^{2N}\times2^{2N}$ transition probability matrix, where $p_{i,j}$ is the probability that the state moves from the $i$th state to the $j$th state in a time step of the birth-death process. Denote the probability that the mutant fixates in layer 1 by $x_i^{[1]}$ starting from the $i$th state, where
$i \in \{1, \ldots, 2^{2N}\}$. Similarly, denote the probability that the mutant fixates in layer 2 by $x_i^{[2]}$ starting from the $i$th state. We can obtain $x_i^{[1]}$ by solving the linear system
\begin{equation}\label{vec-eq-layer1}
\bm{x}^{[1]}=P\bm{x}^{[1]},
\end{equation}
where $\bm{x}^{[1]}=(x_1^{[1]}, \ldots, x_{2^{2N}}^{[1]})^{\top}$, and ${}^\top$ represents the transposition, with boundary conditions $x^{[1]}_{\varphi((1, \ldots, 1, 1, \ldots, 1))}=1$, $x^{[1]}_{\varphi((1, \ldots, 1, 0, \ldots, 0))}=1$, $x^{[1]}_{\varphi((0, \ldots, 0, 1, \ldots, 1))}=0$, and $x^{[1]}_{\varphi((0, \ldots, 0, 0, \ldots, 0))}=0$. Similarly, we can obtain $x_i^{[2]}$ by solving the same linear system
\begin{equation}\label{vec-eq-layer2}
\bm{x}^{[2]}=P\bm{x}^{[2]},
\end{equation}
where $\bm{x}^{[2]}=(x_1^{[2]}, \ldots, x_{2^{2N}}^{[2]})^{\top}$, with boundary conditions $x^{[2]}_{\varphi((1, \ldots, 1, 1, \ldots, 1))}=1$, $x^{[2]}_{\varphi((1, \ldots, 1, 0, \ldots, 0))}=0$, $x^{[2]}_{\varphi((0, \ldots, 0, 1, \ldots, 1))}=1$, and $x^{[2]}_{\varphi((0, \ldots, 0, 0, \ldots, 0))}=0$. Let $C \subset S$ be the set of initial states that contain only one replica node of mutant type in layer 1 and one replica node of mutant type in layer 2. The cardinality of $C$ is $N^2$. Denote the numerical labels of states in $C$ by $\{k_1, \ldots, k_{N^2}\}$. Then, the fixation probability for the mutant type in layer 1 and 2 starting with the initial configuration with just one mutant in each layer, denoted by $x^{[1]}_C$ and $x^{[2]}_C$, respectively, is given by
\begin{align}
x^{[1]}_C&=\sum_{i\in\{k_1, \ldots, k_{N^2}\}} \frac{x^{[1]}_i}{N^2},\\
x^{[2]}_C&=\sum_{i\in\{k_1, \ldots, k_{N^2}\}} \frac{x^{[2]}_i}{N^2}.
\end{align}

For an arbitrary two-layer network, we need to solve a linear system with $2^{2N}-4$ unknowns to obtain the fixation probability of the mutant type. This is computationally prohibitive when $N$ is large. 
Although we can exploit that $x^{[1]}_{\varphi((0, \ldots, 0, s_1^{[2]}, \ldots, s_N^{[2]}))}=0$ and
$x^{[1]}_{\varphi((1, \ldots, 1, s_1^{[2]}, \ldots, s_N^{[2]}))}=1$
for any $(s_1^{[2]}, \ldots, s_N^{[2]}) \in \{0, 1\}^N$ and similar relationships for $x^{[2]}$, the number of unknowns still scales with
$2^{2N}$ as $N$ increases.
Therefore, to drastically reduce the dimension of the linear system to be solved, we analyze two-layer networks with a highly symmetric structure for each layer, in which all or most replica nodes are structurally equivalent to other replica nodes. This strategy has been used for exactly calculating fixation probabilities on conventional networks \cite{Lieberman2005nature,Nowak2006book}, hypergraphs \cite{Liu2023PLoSCB}, and temporal networks \cite{Gyan2023JMB}.

In the following text, we consider model 1, except in subsection~\ref{sub:model2-semianalytical}, where we briefly consider model 2.

\subsection{Coupled complete graphs} \label{bd-ccn}

We first consider the case in which each layer is the complete graph with $N$ nodes. Because all nodes in each layer are structurally equivalent to one another, we only need to track the number of individuals with the mutant type in both layers, denoted by $i_1$; the number of individuals with the mutant type in layer 1 and the resident type in layer 2, denoted by $i_2$; the number of individuals with the resident type in layer 1 and the mutant type in layer 2, denoted by $i_3$; and the number of individuals with the resident type in both layers, denoted by $i_4$. One can specify the state of the evolutionary dynamics by a 4-tuple $\bm{i}=(i_1, i_2, i_3, i_4)$, where $i_1, i_2, i_3, i_4 \in \{0, 1, \ldots, N\}$ and $i_1+i_2+i_3+i_4=N$. Therefore, there are $\binom{N+3}{3}$ states in total, where $\binom{}{}$ represents the binomial coefficient. For visual clarity, we denote the transition probability matrix by $P=[p_{\bm{i} \to \bm{j}}]$, where $p_{\bm{i} \to \bm{j}}$ is the probability that the state moves from $\bm{i}=(i_1, i_2, i_3, i_4)$ to $\bm{j}=(j_1, j_2, j_3, j_4)$ in a time step. Assume that the current state is $\bm{i}=(i_1, i_2, i_3, i_4)$. There are nine types of events that can occur next.

In the first type of event, an individual who has the mutant type in layer 1 (and either type in layer 2) is selected as the parent, which occurs with probability $[2ri_1+(1+r)i_2]/[2ri_1+(1+r)(i_2+i_3)+2i_4]$, and layer 1 is selected for the reproduction event with probability $1/2$. Then, we select a neighbor of the parent in layer 1 for death, and the selected individual, which we refer to as the child, has the resident type in layer 1 and the mutant type in layer 2 with probability $i_3/(N-1)$. Then, the child copies the parent's type in layer 1. The state after this event is $(i_1+1, i_2, i_3-1, i_4)$. Therefore, we obtain
\begin{equation}
p_{(i_1,i_2,i_3,i_4) \to (i_1+1, i_2, i_3-1, i_4)} = \frac{2ri_1+(1+r)i_2}{2ri_1+(1+r)(i_2+i_3)+2i_4}\cdot\frac{1}{2}\cdot\frac{i_3}{N-1}.
\end{equation}

In the second type of event, an individual who has the mutant type in layer 1 is selected as the parent, which occurs with probability $[2ri_1+(1+r)i_2]/[2ri_1+(1+r)(i_2+i_3)+2i_4]$, and layer 1 is selected for reproduction with probability $1/2$. Then, we select a neighbor of the parent in layer 1 as the child, and the child has the resident type in both layers, which occurs with probability $i_4/(N-1)$. Then, the child copies the parent's type in layer 1. The state after this event is $(i_1, i_2+1, i_3, i_4-1)$. Therefore, we obtain
\begin{equation}
p_{(i_1,i_2,i_3,i_4) \to (i_1, i_2+1, i_3, i_4-1)} = \frac{2ri_1+(1+r)i_2}{2ri_1+(1+r)(i_2+i_3)+2i_4}\cdot\frac{1}{2}\cdot\frac{i_4}{N-1}.
\end{equation}

In the third type of event, an individual who has the resident type in layer 1 is selected as the parent, which occurs with probability $[(1+r)i_3+2i_4]/[2ri_1+(1+r)(i_2+i_3)+2i_4]$, and layer 1 is selected for reproduction with probability $1/2$. Then, we select a neighbor of the parent in layer 1 as the child, and the child has the mutant type in both layers, which occurs with probability $i_1/(N-1)$. The state after this event is $(i_1-1, i_2, i_3+1, i_4)$. Therefore, we obtain
\begin{equation}
p_{(i_1,i_2,i_3,i_4) \to (i_1-1, i_2, i_3+1, i_4)} = \frac{(1+r)i_3+2i_4}{2ri_1+(1+r)(i_2+i_3)+2i_4}\cdot\frac{1}{2}\cdot\frac{i_1}{N-1}.
\end{equation}

In the fourth type of event, an individual who has the resident type in layer 1 is selected as the parent, which occurs with probability $[(1+r)i_3+2i_4]/[2ri_1+(1+r)(i_2+i_3)+2i_4]$, and layer 1 is selected for reproduction with probability $1/2$. Then, we select a neighbor of the parent in layer 1 as the child, and the child has the mutant type in layer 1 and the resident type in layer 2, which occurs with probability $i_2/(N-1)$. The state after this event is $(i_1, i_2-1, i_3, i_4+1)$. Therefore, we obtain
\begin{equation}
p_{(i_1,i_2,i_3,i_4) \to (i_1, i_2-1, i_3, i_4+1)} = \frac{(1+r)i_3+2i_4}{2ri_1+(1+r)(i_2+i_3)+2i_4}\cdot\frac{1}{2}\cdot\frac{i_2}{N-1}.
\end{equation}

In the fifth type of event, an individual who has the mutant type in layer 2 is selected as the parent, which occurs with probability $[2ri_1+(1+r)i_3]/[2ri_1+(1+r)(i_2+i_3)+2i_4]$, and layer 2 is selected for reproduction with probability $1/2$. Then, we select a neighbor of the parent in layer 2 as the child, and the child has the mutant type in layer 1 and the resident type in layer 2, which occurs with probability $i_2/(N-1)$. The state after this event is $(i_1+1, i_2-1, i_3, i_4)$. Therefore, we obtain
\begin{equation}
p_{(i_1,i_2,i_3,i_4) \to (i_1+1, i_2-1, i_3, i_4)} = \frac{2ri_1+(1+r)i_3}{2ri_1+(1+r)(i_2+i_3)+2i_4}\cdot\frac{1}{2}\cdot\frac{i_2}{N-1}.
\end{equation}

In the sixth type of event, an individual who has the mutant type in layer 2 is selected as the parent, which occurs with probability $[2ri_1+(1+r)i_3]/[2ri_1+(1+r)(i_2+i_3)+2i_4]$, and layer 2 is selected for reproduction with probability $1/2$. Then, we select a neighbor of the parent in layer 2 as the child, and the child has the resident type in both layers, which occurs with probability $i_4/(N-1)$. The state after this event is $(i_1, i_2, i_3+1, i_4-1)$. Therefore, we obtain
\begin{equation}
p_{(i_1,i_2,i_3,i_4) \to (i_1, i_2, i_3+1, i_4-1)} = \frac{2ri_1+(1+r)i_3}{2ri_1+(1+r)(i_2+i_3)+2i_4}\cdot\frac{1}{2}\cdot\frac{i_4}{N-1}.
\end{equation}

In the seventh type of event, an individual who has the resident type in layer 2 is selected as the parent, which occurs with probability $[(1+r)i_2+2i_4]/[2ri_1+(1+r)(i_2+i_3)+2i_4]$, and layer 2 is selected for reproduction with probability $1/2$. Then, we select a neighbor of the parent in layer 2 as the child, and the child has the mutant type in both layers, which occurs with probability $i_1/(N-1)$. The state after this event is $(i_1-1, i_2+1, i_3, i_4)$. Therefore, we obtain
\begin{equation}
p_{(i_1,i_2,i_3,i_4) \to (i_1-1, i_2+1, i_3, i_4)} = \frac{(1+r)i_2+2i_4}{2ri_1+(1+r)(i_2+i_3)+2i_4}\cdot\frac{1}{2}\cdot\frac{i_1}{N-1}.
\end{equation}

In the eighth type of event, an individual who has the resident type in layer 2 is selected as the parent, which occurs with probability $[(1+r)i_2+2i_4]/[2ri_1+(1+r)(i_2+i_3)+2i_4]$, and layer 2 is selected for reproduction with probability $1/2$. Then, we select a neighbor of the parent in layer 2 as the child, and the child has the resident type in layer 1 and the mutant type in layer 2, which occurs with probability $i_3/(N-1)$. The state after this event is $(i_1, i_2, i_3-1, i_4+1)$. Therefore, we obtain
\begin{equation}
p_{(i_1,i_2,i_3,i_4) \to (i_1, i_2, i_3-1, i_4+1)} = \frac{(1+r)i_2+2i_4}{2ri_1+(1+r)(i_2+i_3)+2i_4}\cdot\frac{1}{2}\cdot\frac{i_3}{N-1}.
\end{equation}

If any other event occurs, the state remains unchanged. Therefore, we obtain 
\begin{align}\label{eq-ccn-unchanged}
p_{(i_1,i_2,i_3,i_4) \to (i_1, i_2, i_3, i_4)} = 1&-p_{(i_1,i_2,i_3,i_4) \to (i_1+1, i_2, i_3-1, i_4)}-p_{(i_1,i_2,i_3,i_4) \to (i_1, i_2+1, i_3, i_4-1)}\nonumber\\
&-p_{(i_1,i_2,i_3,i_4) \to (i_1-1, i_2, i_3+1, i_4)}-p_{(i_1,i_2,i_3,i_4) \to (i_1, i_2-1, i_3, i_4+1)}\nonumber\\
&-p_{(i_1,i_2,i_3,i_4) \to (i_1+1, i_2-1, i_3, i_4)}-p_{(i_1,i_2,i_3,i_4) \to (i_1, i_2, i_3+1, i_4-1)}\nonumber\\
&-p_{(i_1,i_2,i_3,i_4) \to (i_1-1, i_2+1, i_3, i_4)}-p_{(i_1,i_2,i_3,i_4) \to (i_1, i_2, i_3-1, i_4+1)}.
\end{align}

By slightly adapting the notation introduced in section~\ref{fp-general-case},
we denote by $x^{[1]}_{\bm{i}}$ and $x^{[2]}_{\bm{i}}$ the fixation probability of the mutant type in layer 1 and layer 2, respectively, when the initial state is $\bm{i}=(i_1,i_2,i_3,i_4)$. To obtain fixation probabilities in layer 1, we solve Eq.~\eqref{vec-eq-layer1}, where $\bm{x}^{[1]}$ is a column vector in which each entry
is the fixation probability for the mutant type starting from one of the $\binom{N+3}{3}$ initial states. The boundary conditions are given by $x^{[1]}_{(N,0,0,0)}=1$, $x^{[1]}_{(0,N,0,0)}=1$, $x^{[1]}_{(0,0,N,0)}=0$, and $x^{[1]}_{(0,0,0,N)}=0$. To obtain fixation probabilities in layer 2, we solve Eq.~\eqref{vec-eq-layer2} with boundary conditions $x^{[2]}_{(N,0,0,0)}=1$, $x^{[2]}_{(0,N,0,0)}=0$, $x^{[2]}_{(0,0,N,0)}=1$, and $x^{[2]}_{(0,0,0,N)}=0$. There are two initial states with one mutant in each layer, i.e., $(1, 0, 0, N-1)$ and $(0, 1, 1, N-2)$. These initial states occur with probability $1/N$ and $(N-1)/N$, respectively. Therefore, we obtain
\begin{equation}
x^{[\ell]}_C = \frac{1}{N} x^{[\ell]}_{(1, 0, 0, N-1)} + \frac{N-1}{N} x^{[\ell]}_{(0, 1, 1, N-2)},\quad \ell \in \{ 1, 2 \},
\end{equation}
where we note that $x^{[\ell]}_C$ is the fixation probability for the mutant type in layer $\ell$ when there is initially one mutant in each layer.

We obtained $x^{[1]}_C$ $(= x^{[2]}_C)$ by numerically solving Eq.~\eqref{vec-eq-layer1} for $N=6$ and $N=30$. The results shown in Figure~\ref{fig:m1-m2-complete}(a) and \ref{fig:m1-m2-complete}(b) for $N=6$ and $N=30$, respectively, indicate that these coupled complete graphs are suppressors of selection. This result is consistent with Theorem~\ref{complete-thm}.

\captionsetup{labelfont=bf}
\begin{figure}[H]
  \centering
  \includegraphics[width=1.0\linewidth]{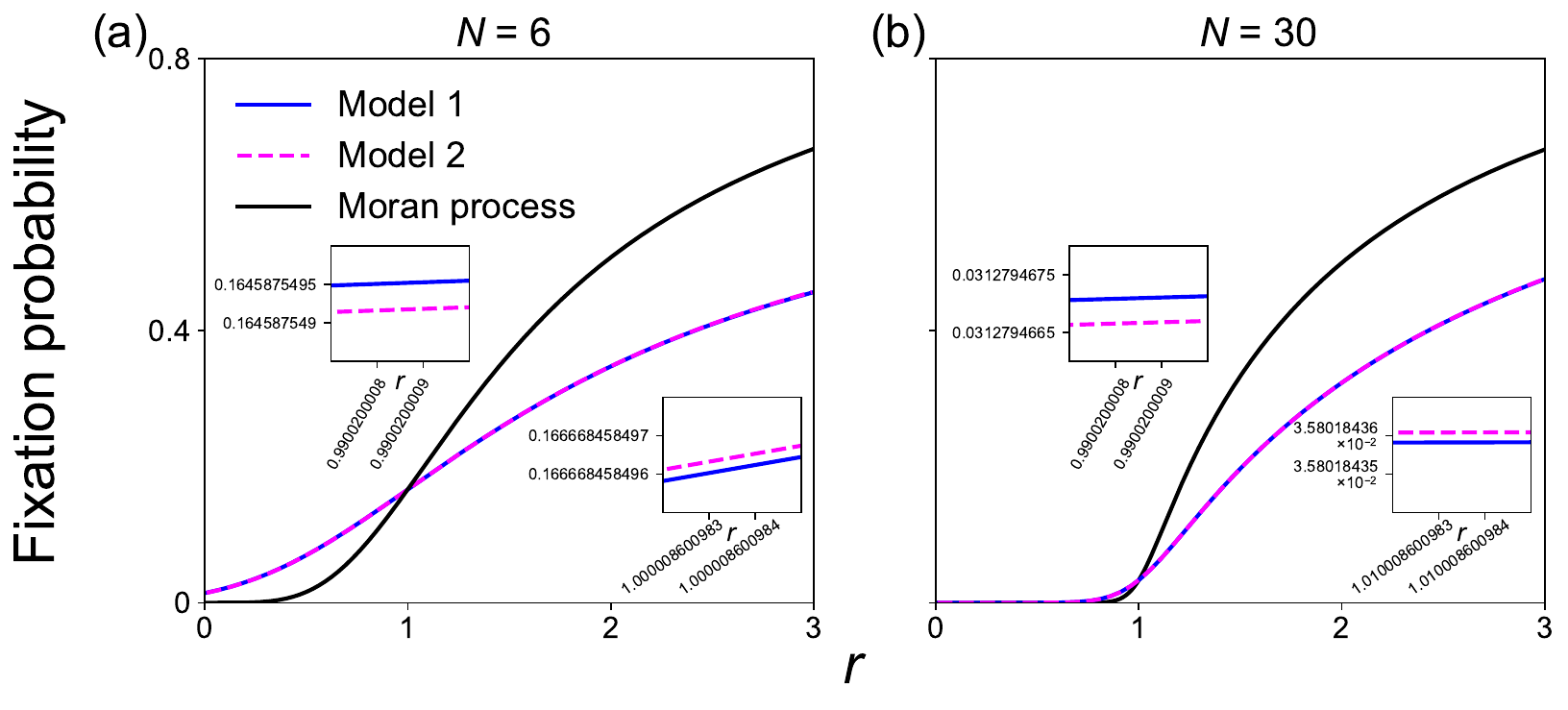}
\caption{Fixation probability for coupled complete graphs under models 1 and 2. (a) $N=6$. (b) $N=30$. The insets to the left within each panel magnify the results for $r$ values less than and close to $r=1$. Those to the right within each panel magnify the results for $r$ values greater than and close to $r=1$.}
\label{fig:m1-m2-complete}
\end{figure}

\subsection{Combination of the complete graph and star graph}\label{complete-star-bilayer}

Next, we consider the two-layer network in which layer 1 is the complete graph and layer 2 is the star graph. All the $N$ nodes in the complete graph are structurally equivalent, as are all the $N-1$ leaf nodes (i.e., nodes with degree 1) in the star graph. Therefore, we represent the state of the evolutionary dynamics by $\bm{i}=(h_1,h_2,i_1,i_2,i_3,i_4)$, i.e., an ordered 6-tuple, where $h_1=0$ or 1 if the individual that is the hub node (i.e., the replica node with degree $N-1$) in layer 2 is of resident or mutant type in layer 1, respectively; $h_2=0$ or 1 if the hub node in layer 2 is of the resident or mutant type, respectively; $i_1$ is the number of the remaining $N-1$ individuals that have the mutant type in both layers; $i_2$ is the number of the remaining $N-1$ individuals that have the mutant type in layer 1 and the resident type in layer 2; $i_3$ is the number of the remaining $N-1$ individuals that have the resident type in layer 1 and the mutant type in layer 2; and $i_4$ is the number of the remaining $N-1$ individuals that have the resident type in both layers.
There are $2^2\binom{N+2}{3}$ states in total.

Similarly to the case in which both layers are the complete graph, we distinguish nine types of state transitions from each state. We derive the probability of each state transition in section~S5.

We use the same numerical method for solving Eqs.~\eqref{vec-eq-layer1} and \eqref{vec-eq-layer2} as that for the coupled complete graphs. We show the fixation probability for the mutant type for $N=6$ and $N=30$ in Figure~\ref{fig:m1-complete-star}(a) and \ref{fig:m1-complete-star}(b), respectively. The figure suggests that the two-layer networks composed of a complete graph layer and a star graph layer are suppressors of selection.

\begin{figure}[H]
  \centering
  \includegraphics[width=1.0\linewidth]{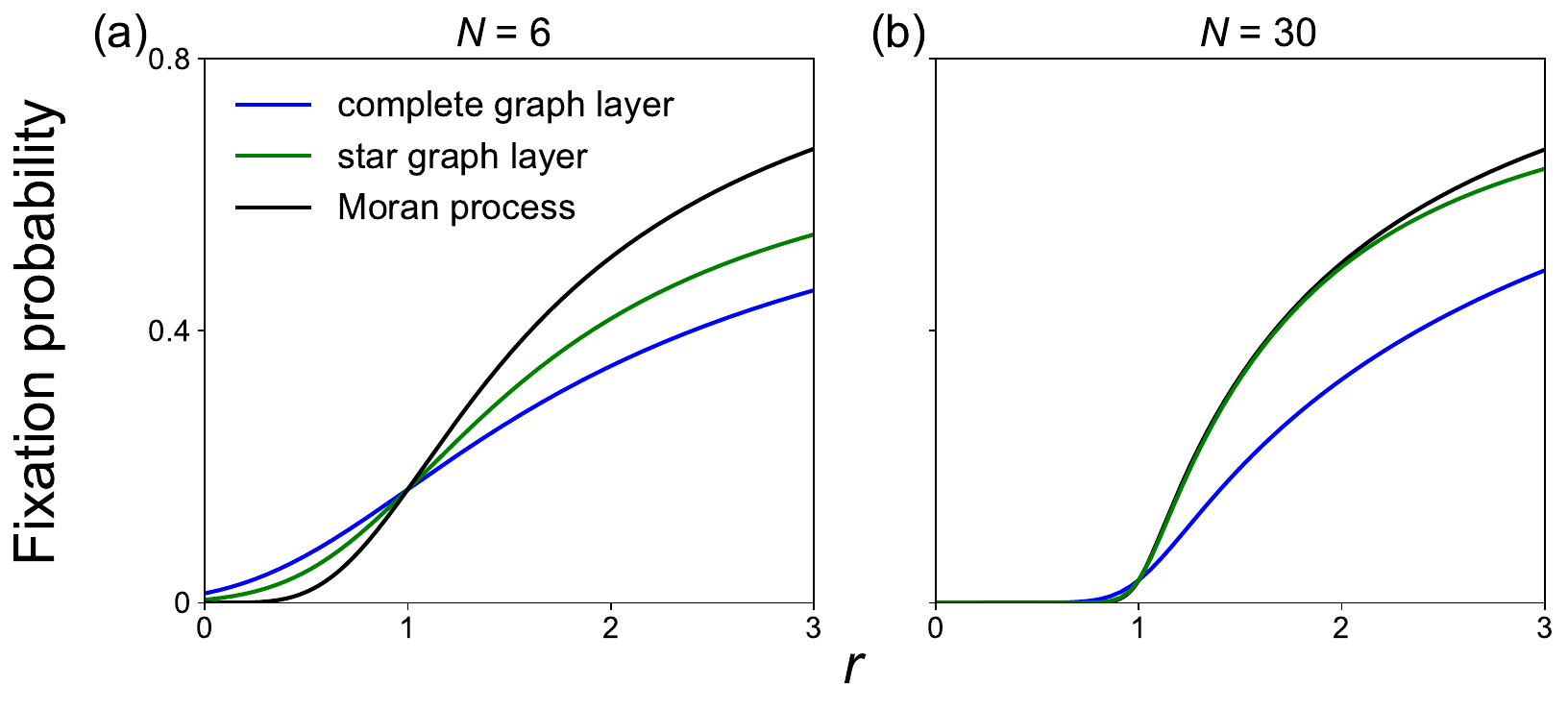}
\caption{Fixation probability for two-layer networks composed of a complete graph layer and a star graph layer under model 1. (a) $N=6$. (b) $N=30$.
}
\label{fig:m1-complete-star}
\end{figure}

\subsection{Coupled star graphs}\label{bd-ssn}

Here we consider the case in which each layer is the star graph with $N$ nodes. We assume that the hub replica node in layer 1 corresponds to a leaf replica node in layer 2 and vice versa. Then, we can specify the network's state by an ordered 8-tuple $\bm i=(h_1, h_2, h_3, h_4, i_1, i_2, i_3, i_4)$, where $h_1=0$ or 1 if the hub node in layer 1 is of resident or mutant type, respectively; $h_2=0$ or 1 if the individual that is the hub node in layer 1 is of resident or mutant type in layer 2, respectively; $h_3=0$ or 1 if the individual that is the hub node in layer 2 is of resident or mutant type in layer 1, respectively; $h_4=0$ or 1 if the hub node in layer 2 is of resident or mutant type, respectively. We reuse $i_1$, $i_2$, $i_3$, and $i_4$ as defined in section~\ref{complete-star-bilayer} but with a slight difference. Here, we count $i_1$, $i_2$, $i_3$, and $i_4$ among the $N-2$ individuals that are leaf nodes in both layers. There are $2^4\binom{N+1}{3}$ states in total.

We distinguish all types of state transitions from each state and derive the probability of each state transition in section~S6. The number of the types of state transitions varies between seven and nine and depends on the current  state.

We use the same numerical method for solving Eqs.~\eqref{vec-eq-layer1} and \eqref{vec-eq-layer2} as that for the coupled complete graphs. We show the fixation probability for the mutant type for $N=6$ and $N=30$ in Figure~\ref{fig:m1-star-star}(a) and \ref{fig:m1-star-star}(b), respectively. Figure~\ref{fig:m1-star-star} indicates that the coupled star graph is a suppressor of selection when $N=6$ but is neither a suppressor nor an amplifier of selection when $N=30$. However, in both cases, the coupled star graph is more suppressing than the one-layer star graph, which is a strong amplifier of selection. This last result is consistent with Theorem~\ref{complete-bi-thm} and Remark~\ref{star-thm}.

\begin{figure}[H]
  \centering
  \includegraphics[width=1.0\linewidth]{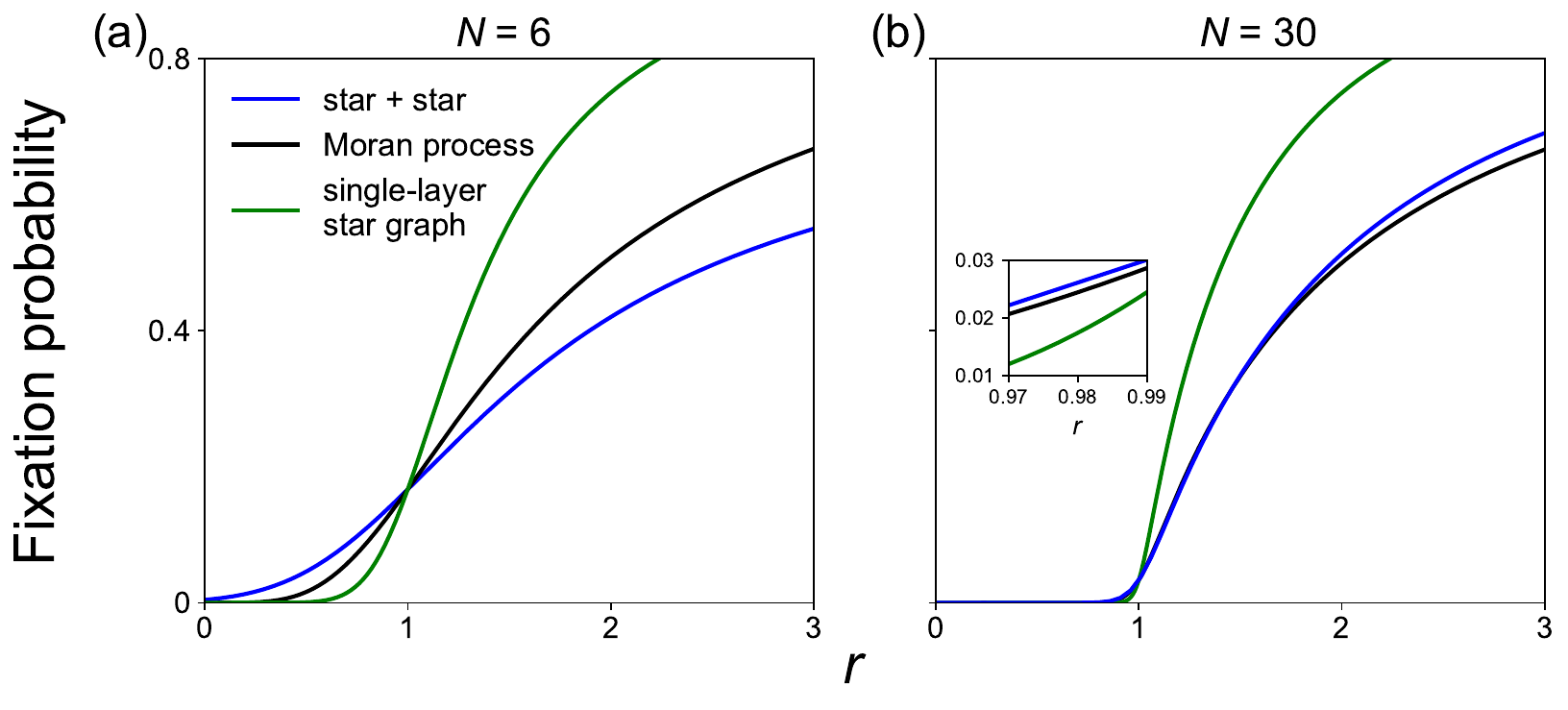}
\caption{Fixation probability for coupled star graphs under model 1. We compare the results for the coupled star graphs, shown in blue, with those for the Moran process, shown in black, and those for the single-layer star graphs, shown in green. (a) $N=6$. (b) $N=30$. The inset in (b) magnifies the result for $r$ values less than and close to $r=1$.
}
 \label{fig:m1-star-star}
\end{figure}

\subsection{Combination of the complete graph and complete bipartite graph}

Consider two-layer networks in which layer 1 is the complete graph and layer 2 is the complete bipartite graph $K_{N_1, N_2}$, where $N_1+N_2=N$. The complete bipartite graph $K_{N_1, N_2}$ has two disjoint subsets of nodes $V_1$ and $V_2$ with $N_1$ and $N_2$ nodes, respectively. Each node in $V_1$ is adjacent to each node in $V_2$. We can describe the state of the evolutionary dynamics by an 8-tuple. For each 8-tuple state,
we distinguish 17 types of transition events and can derive the transition probability of each state to each state.
We show the calculations of the transition probability matrix in section~S7.

We use the same numerical method to solve \eqref{vec-eq-layer1} and \eqref{vec-eq-layer2} for this two-layer network. We show the fixation probability for $N=6$ in Figure~\ref{fig:m1-complete-bipartite}(a) and \ref{fig:m1-complete-bipartite}(b), and $N=20$ in Figure~\ref{fig:m1-complete-bipartite}(c) and \ref{fig:m1-complete-bipartite}(d), respectively. We reduce the larger $N$ value to $20$ due to large memory requirement for this network. We set $N_1 = N_2 = N/2$ in Figure~\ref{fig:m1-complete-bipartite}(a) and \ref{fig:m1-complete-bipartite}(c). In this case, a one-layer network $K_{N_1, N_2}$ is a regular graph and therefore an isothermal graph. We set $N_2 \approx 2 N_1$, where $\approx$ represents ``approximately equal to'', in Figure~\ref{fig:m1-complete-bipartite}(b) and \ref{fig:m1-complete-bipartite}(d). Figure~\ref{fig:m1-complete-bipartite} shows that these two-layer networks are suppressors of selection.

\begin{figure}[H]
  \centering
  \includegraphics[width=1.0\linewidth]{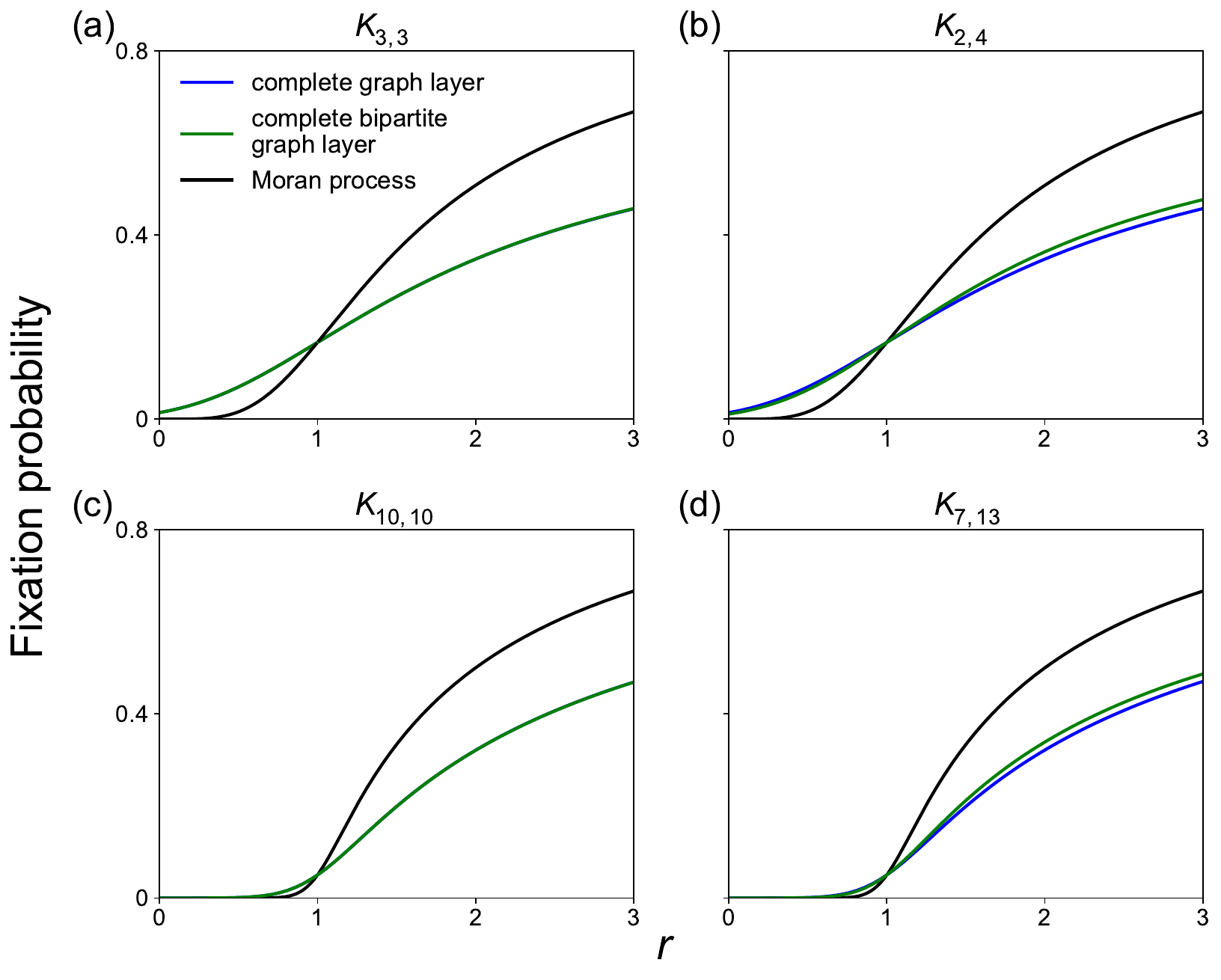}
\caption{Fixation probability for two-layer networks composed of a complete graph layer and a complete bipartite graph layer under model 1. (a) $N=6$ with $K_{3,3}$. (b) $N=6$ with $K_{2,4}$. (c) $N=20$ with $K_{10,10}$. (d) $N=20$ with $K_{7,13}$. In panels (a) and (c), the results for the complete graph layer are close to those for the complete bipartite graph layer such that the blue lines are almost hidden behind the green lines.
}
\label{fig:m1-complete-bipartite}
\end{figure}

\subsection{Combination of the complete graph and two-community networks}

Empirical networks often have community (i.e., group) structure \cite{Fortunato2010PhysRep}.
Therefore, in this section, we consider two-layer networks in which layer 1 is the complete graph and layer 2 is a weighted network with two communities. Specifically, layer 2 is composed of two disjoint sets of nodes $V_1$ and $V_2$ with $N_1$ and $N_2$ nodes, respectively, where $N_1+N_2=N$. Each set of nodes forms a clique (i.e., complete graph as a subgraph) with edge weight $1$. In addition, each node in $V_1$ is connected to each node in $V_2$ with edge weight $\overline{\epsilon}$. A small $\overline{\epsilon}$ implies a strong community structure. It should also be noted that the combination of the complete graph and the complete bipartite graph corresponds to this model in the limit of $\overline{\epsilon} \to \infty$. Similarly to the case of the combination of the complete graph and the complete bipartite network, here we use an 8-tuple 
and distinguish 17 types of transition events from each state to another. We show the calculations of the transition probability matrix in section~S8.

We numerically solve Eqs.~\eqref{vec-eq-layer1} and \eqref{vec-eq-layer2} for $N=6$ and $N=20$ with $\overline{\epsilon}=0.1$. We show the results for $N=6$ in Figure~\ref{fig:m1-complete-community}(a), \ref{fig:m1-complete-community}(b) and $N=20$ in Figure~\ref{fig:m1-complete-community}(c), \ref{fig:m1-complete-community}(d). In Figure~\ref{fig:m1-complete-community}(a), \ref{fig:m1-complete-community}(c), we set $N_1 = N_2 = N/2$, and the two-community network is an isothermal graph. In Figure~\ref{fig:m1-complete-community}(b), \ref{fig:m1-complete-community}(d), we set $N_2 \approx 2 N_2$. Figure~\ref{fig:m1-complete-community} indicates that these two-layer networks are suppressors of selection.

\begin{figure}[H]
  \centering
  \includegraphics[width=1.0\linewidth]{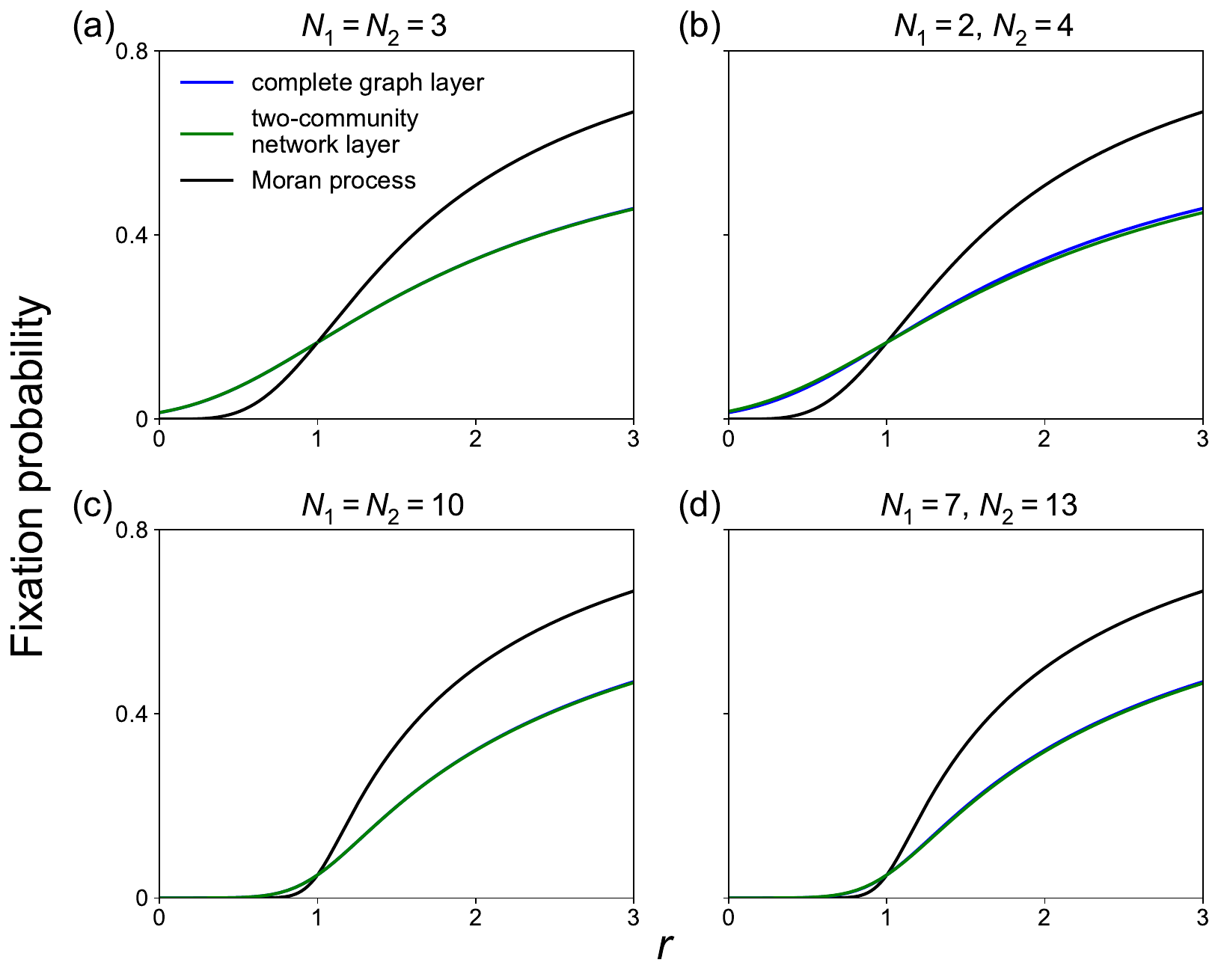}
\caption{Fixation probability for two-layer networks composed of a complete graph layer and a two-community network layer under model 1. (a) $N_1=N_2=3$ and $\overline{\epsilon}=0.1$. (b) $N_1=2$, $N_2=4$, and $\overline{\epsilon}=0.1$. (c) $N_1=N_2=10$ and $\overline{\epsilon}=0.1$. (d) $N_1=7$, $N_2=13$, and $\overline{\epsilon}=0.1$.
}
\label{fig:m1-complete-community}
\end{figure}

\subsection{Death-birth process variant of model 1}\label{m1-db-rule}

We have analyzed model 1, which is a two-layer Bd process. To assess the robustness of our main result that two-layer networks are mostly suppressors of selection,
in this section we consider a variant of model 1 in which we replace the Bd updating rule by the death-birth updating rule with selection on the birth, often referred to as the dB rule \cite{Ohtsuki2006Nature,Masuda2009JTB, Shakarian2012Biosys, Pattni2015ProcRSocA}. According to the dB rule, we select an individual uniformly at random for death in each time step. Then, the neighbors of the dying individual compete to reproduce its type on the vacant site with probability proportional to their fitness. The fixation probability for this death-birth process in the case of the well-mixed population (i.e., unweighted complete graph) is~\cite{Kaveh2015RSocOSci}
%
%
\begin{equation}
\rho^{\text{dB}} = \left(1-\frac{1}{N}\right)\frac{1-\frac{1}{r}}{1-\frac{1}{r^{N-1}}}.
\label{eq:fixation-probability-Moran-dB} 
\end{equation}

We extend the dB process to the case of two-layer networks. For simplicity, we only consider the case in which both layers are complete graphs (i.e., coupled complete graph). In each time step, an individual selected uniformly at random (i.e., with probability $1/N$) dies. We then select one of the two layers to operate the dB process with equal probability, i.e., $1/2$. The neighbors of the dying individual in the selected layer compete to fill the empty site with probability proportional to the product of their fitness and the edge weight (which we set to $1$ because we are considering unweighted complete graphs for both layers). As we show in section~S9, we can derive the set of $\binom{N+3}{3}-4$ linear equations with which to calculate the fixation probability similarly to the case of the Bd process on the coupled complete graph.

We show the fixation probability for this death-birth process on coupled complete graphs with $N=6$ and $N=30$ in Figure~\ref{fig:m1-complete-complete-dB}(a) and \ref{fig:m1-complete-complete-dB}(b), respectively. We find that these coupled complete graphs are suppressors of selection under the dB rule, relative to the Moran process.
The green lines in Figure~\ref{fig:m1-complete-complete-dB} represent Eq.~\eqref{eq:fixation-probability-Moran-dB}. We find that the coupled complete graphs under the dB rule are also more suppressing than the one-layer complete graphs under the same dB rule.

\begin{figure}[H]
  \centering
  \includegraphics[width=1.0\linewidth]{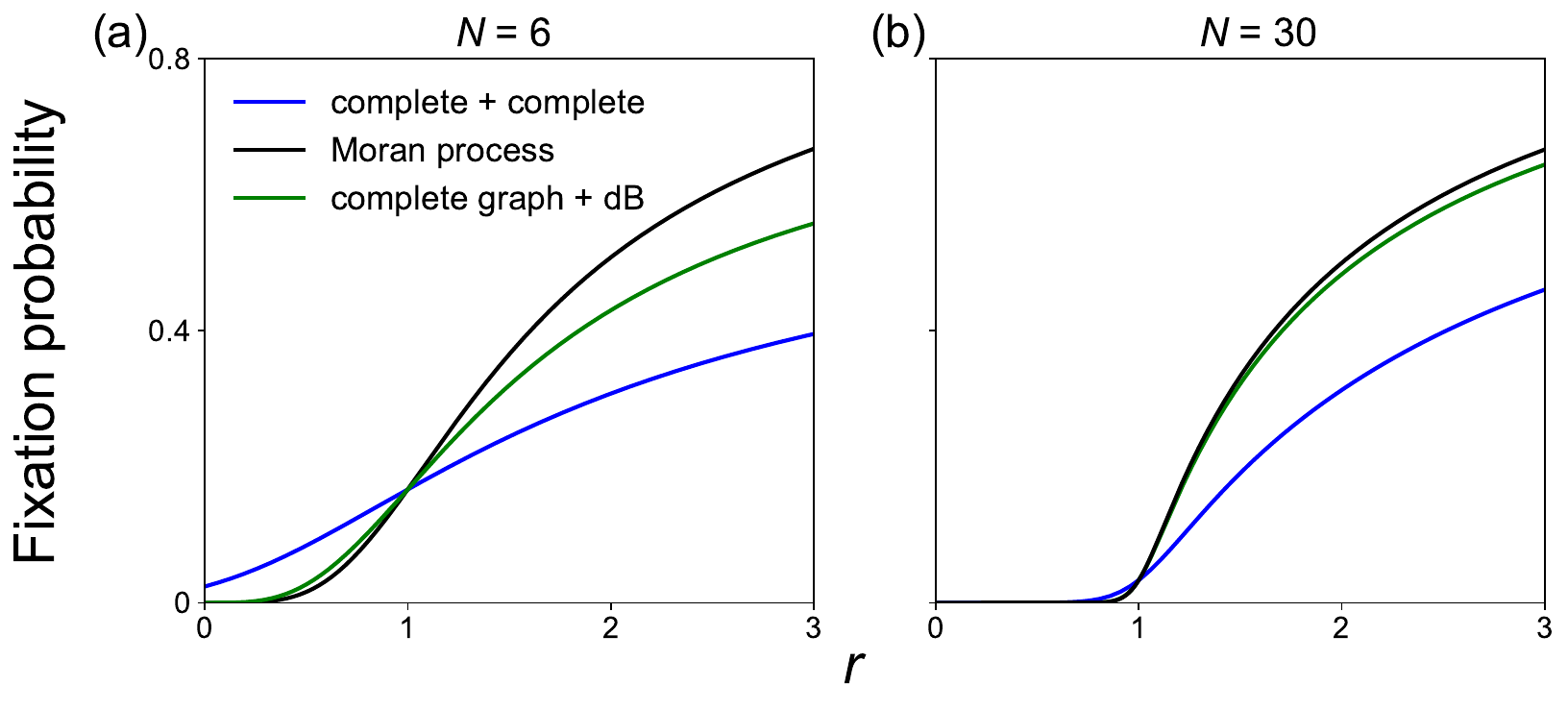}
\caption{Fixation probability for coupled complete graphs under the variant of model 1 with the dB updating. (a) $N=6$. (b) $N=30$. The green lines represent Eq.~\eqref{eq:fixation-probability-Moran-dB}.}
\label{fig:m1-complete-complete-dB}
\end{figure}

\subsection{Model 2 on coupled complete graphs\label{sub:model2-semianalytical}}

In this section, we consider model 2 when both layers are complete graphs. We can describe the state of the network
by the same 4-tuple $\bm{i}=(i_1, i_2, i_3, i_4)$ that we used in section~\ref{bd-ccn} and derive a set of $\binom{N+3}{3}-4$ linear equations to determine the fixation probability. For each state $(i_1,i_2,i_3,i_4)$, we distinguish 21 types of events and obtain the transition probability from each state to another state, as shown in section~S10. 

We show the fixation probability for the mutant type for $N=6$ and $N=30$ by the dashed lines in Figure~\ref{fig:m1-m2-complete}(a) and \ref{fig:m1-m2-complete}(b), respectively. We find that these coupled complete graphs are suppressors of selection. Furthermore, the results for model 2 are close to those for model 1, while model 2 is slightly less suppressing than model 1.

\section{Numerical results}

In this section, we carry out numerical simulations of the Bd process on four two-layer networks without particular symmetry, i.e., a coupled Erd\H{o}s--R\'{e}nyi (ER) random graph, a coupled Barab\'{a}si--Albert (BA) network, and two empirical two-layer networks. To generate a two-layer ER random graph with $N=100$ individuals, in each layer we connected each pair of nodes with probability $0.1$. We iterated generating networks from the ER random graph with $N=100$ nodes until we obtained two connected networks, which we used as two layers. The two generated networks had $M_1=498$ edges and $M_2=500$ edges, respectively. To generate a two-layer BA network, we sampled two networks with $N=100$ nodes each from the BA model \cite{Barabasi1999Sci}. In the network growth process of the BA model, each incoming node is connected to five already existing nodes according to the linear preferential attachment rule. We use the star graph on six nodes as the initial network in each layer. Each of the two generated networks was more heterogeneous than the ER graph in terms of the node's degree, was connected, and had $M=475$ edges. Without loss of generality, we uniformly randomly permuted the label of all nodes in layer 2. Otherwise, there would be a strong positive correlation between the degrees of the two replica nodes of the same individual. One empirical network is the Vickers--Chan 7th Graders (VC7) network, which is a two-layer network of scholastic and friendship relationships among $N = 29$ seventh grade students in a school in Victoria, Australia, with $M_1=126$ edges in layer 1 and $M_2=152$ edges in layer 2 \cite{Vickers1981}. The second empirical network is the Lazega Law Firm (LLF) network, which is a two-layer network of professional and cooperative relationships among $N=71$ partners at the LLF, with $M_1=717$ edges in layer 1 and $M_2=726$ edges in layer 2 \cite{Emmanuel2001}.

We focus on model 1 and examine whether these two-layer networks tend to be suppressors of selection.
We initially placed a mutant on just one replica node in each layer. Therefore, there are $N\times N$ possible initial states. To calculate the fixation probability for a single mutant for each layer, we run the Bd process until the mutant type or the resident type fixates in the selected layer. For each value of $r$, we run $20N^2$ simulations starting from each of the $N^2$ initial conditions 20 times. We obtain the fixation probability of the mutant type for each layer as the number of runs in which the mutant type has fixated in the selected layer divided by $20 N^2$.

Figure~\ref{fig:four-networks}(a)--(d) shows the relationship between the fixation probability for a single mutant and $r$ for the four two-layer networks. The figure shows that both layers are suppressors of selection in all four two-layer networks. Unexpectedly, we also find that the fixation probability as a function of $r$ is similar between the two layers, which are different networks in terms of edges, in all four two-layer networks.

\begin{figure}[H]
  \centering
  \includegraphics[width=0.9\linewidth]{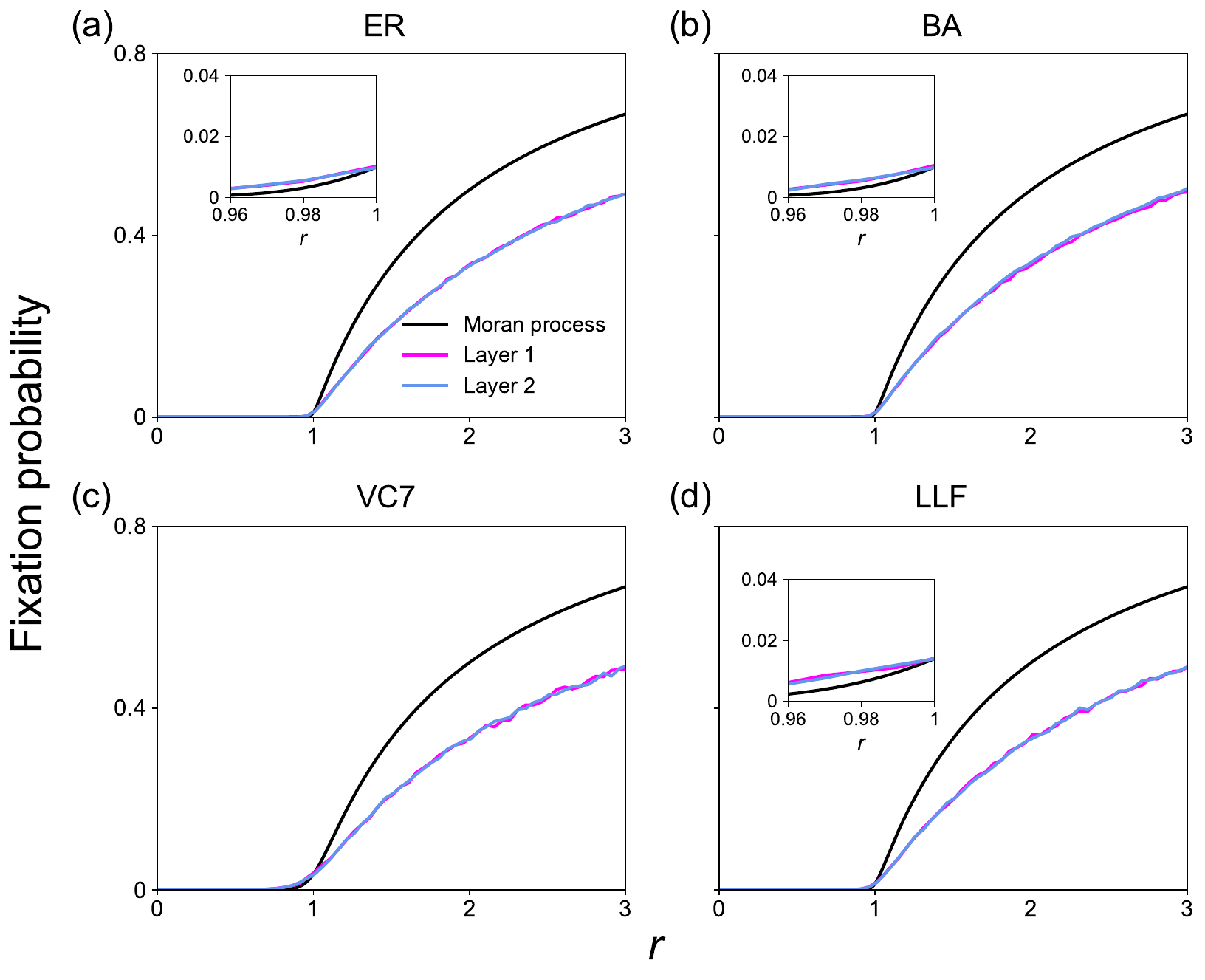}
\caption{Fixation probability in model and empirical two-layer networks. (a) Two-layer ER graph. (b) Two-layer BA model. (c) Vickers--Chan 7th Graders network (VC7). (d) Lazega Law Firm network (LLF). The insets of (a), (b), and (d) magnify the results for values of $r$ less than and close to $r=1$.
}
\label{fig:four-networks}
\end{figure}

\section{Discussion\label{sec:discussion}}

Inspired by an evolutionary game model on two-layer networks~\cite{Su2022nhb}, we formulated and analyzed constant-selection dynamics on two-layer networks in which each individual's fitness is defined to be the sum of the fitness of the replica nodes over the two layers. Using martingales, we proved that two-layer networks are suppressors of selection if one layer is a particular network with high symmetry--at least relative to that network considered as a single-layer network. The single-layer regular graphs, including the complete graph, cycle, and bipartite complete graphs in which the two parts have the same number of nodes,
%
%
are isothermal graphs~\cite{Lieberman2005nature,Broom2008ProcRoySocA}, i.e., equivalent to the Moran process. Therefore, these theorems show that two-layer networks that have any of these networks as one layer are suppressors of selection regardless of the second layer. Furthermore, we semi-analytically analyzed some two-layer networks in which both layers are highly symmetric networks to show that they are also suppressors of selection, except the coupled star graph with $N=30$ nodes. Nonetheless, the couple star graph with $N=30$ is more suppressing than the single star graph with $N=30$. Numerical simulations of stochastic evolutionary dynamics on four larger two-layer networks without particular symmetry have also shown that these networks are suppressors of selection. Overall, we have provided mathematical results and compelling numerical evidence that two-layer networks are suppressors of selection unless both layers are strong amplifiers of selection, such as the star graph (see Figure~\ref{fig:m1-star-star}).

We argue that the intuitive reason behind this result is the key assumption of our model that the total fitness of a replica node depends on the fitness of the corresponding replica node in the other layer as well as its own fitness. Suppose that $r>1$ and that a replica node $i$ in layer 1 is of resident type. Then, intuitively, it is more likely to be invaded by a mutant type if a neighbor is a mutant, than vice versa, because the mutant's fitness ($=r$) is higher than the resident's fitness ($=1$). However, if the replica node $i$ in layer 2 is of mutant type, the total fitness for the $i$th individual is equal to $1+r$. Therefore, the mutant type in layer 2 boosts the likelihood that the $i$th individual reproduces in layer 1 relative to the case of a single-layer network. In this manner, the two-layer nature of the model blurs the effect of fitness due to the interference of one layer into constant-selection dynamics in the other layer. This is why two-layer networks are expected to be suppressors of selection--at least relative to their one-layer counterparts. We note that we exploited this intuition in formulating and proving our theorems using martingales.

As we reviewed in section~\ref{sec:introduction}, most networks are amplifiers of selection under the Bd process and uniform initialization. However, small directed networks~\cite{Masuda2009JTB}, metapopulation model networks~\cite{Yagoobi2021SciRep,Marrec2021PRL}, a type of temporal network called a switching network (i.e., in which the network switches between two static network with regular or irregular time intervals) when $N$ is small~\cite{Gyan2023JMB}, and hypergraphs~\cite{Liu2023PLoSCB} tend to be suppressors of selection under the same conditions (i.e., the Bd process with uniform initialization) even if the undirected variant of them is an amplifier of selection.  Here we add two-layer networks as another case in which suppressors of selection are common. These results altogether suggest that amplifiers of selection under the Bd process with uniform initialization are not as common as was initially considered. It is straightforward to extend our models to the case of more than two layers. Constant-selection dynamics under adaptive networks (i.e., time-varying networks in which network changes are induced by the state, or type, of the nodes) are underexplored~\cite{Tkadlec2021arxiv,Melissourgos2022JCSS}. Determining whether or not these networks are amplifying or suppressing would be a worthwhile investigation.

We have exploited some highly symmetric networks to be used as network layers with the aim of reducing the number of linear equations to be solved from $O(2^{2N})$ to a polynomial order of $N$. We used the same strategy to analyze hypergraphs~\cite{Liu2023PLoSCB} and switching temporal networks~\cite{Gyan2023JMB}. The same technique was exploited for analytically solving the fixation probability in the complete bipartite graphs~\cite{Zhang2012DisDynNatSoc,Voorhees2013ProcRoySocA}, stars~\cite{Lieberman2005nature,Broom2008ProcRoySocA}, and so-called superstars \cite{Lieberman2005nature}.
However, the size of the two-layer networks for which we exactly calculated the fixation probability is still modest, i.e., up to $N=20$ or $N=30$ depending on the network. This is because the network in our models has two layers, and we need to track the type (i.e., resident or mutant) of the replica nodes in both layers to specify the state of an individual. It would be ideal if this type of mathematical technique were to lead to analytical solutions of the fixation probability rather than just to reducing the dimension of the problem. This is left for future work.

\section*{Acknowledgments}

We thank Alex McAvoy for fruitful discussions.
N.M. acknowledges support from the 
Japan Science and Technology Agency (JST) Moonshot R\&D (under Grant No.\,JPMJMS2021), the National Science Foundation (under Grant Nos.\,2052720 and 2204936), and JSPS KAKENHI (under grant Nos.\,JP 21H04595 and 23H03414).


\newpage

\counterwithout{equation}{section}
\renewcommand{\theequation}{S\arabic{equation}}

\newtheoremstyle{rmk}
{3pt}
{3pt}
{}
{}
{\itshape}
{:}
{.5em}
{}
\newtheorem{rmk}{Remark}
\renewcommand{\thermk}{S\arabic{rmk}}
                                                                                                                                   
\part*{}

\renewcommand{\tablename}{TABLE}
\renewcommand{\thefigure}{S\arabic{figure}}
\renewcommand{\thetable}{S\arabic{table}}
\renewcommand{\thesection}{S\arabic{section}}
\renewcommand{\thesubsection}{\Alph{subsection}}

\section*{\centering{\Large{Supporting Information}}}

\section{Proof of Lemma~6}\label{submartingale-proof-cycle}

In this section, we show that $E[Y_{t+1} | \mathcal{B}_t] \ge Y_t$ when one layer is the cycle graph.

We assume that the replica nodes of the mutant type are initially consecutive. Then, at any time step, all replica nodes of the mutant type are consecutively numbered in layer 1, forming what we call the segment, without being interrupted by the replica nodes of the resident type (see Figure~\ref{fig:cyc-star} for a schematic). Therefore, $p(\xi_t)$ only depends on the states of the two individuals at the two endpoints of the segment with the mutant type in layer 1, i.e., individuals $j$ and $\ell$ in Figure~\ref{fig:cyc-star}. Note that individuals $j$ and $\ell$ are identical when $X_0=1$. Similarly, $q(\xi_t)$ only depends on the states of the two individuals that are next to the two endpoints of the same segment and have the resident type in layer 1, i.e., individuals $i$ and $k$ in Figure~\ref{fig:cyc-star}. Note that individuals $i$ and $k$ are identical when $X_0=N-1$. We next show that $E[Y_{t+1} | \mathcal{B}_t] \ge Y_t$, regardless of the states of individuals at or adjacent to the two endpoints of the segment.

\renewcommand{\thefigure}{S1}
\begin{figure}[h]
  \centering
  \includegraphics[width=0.8\linewidth]{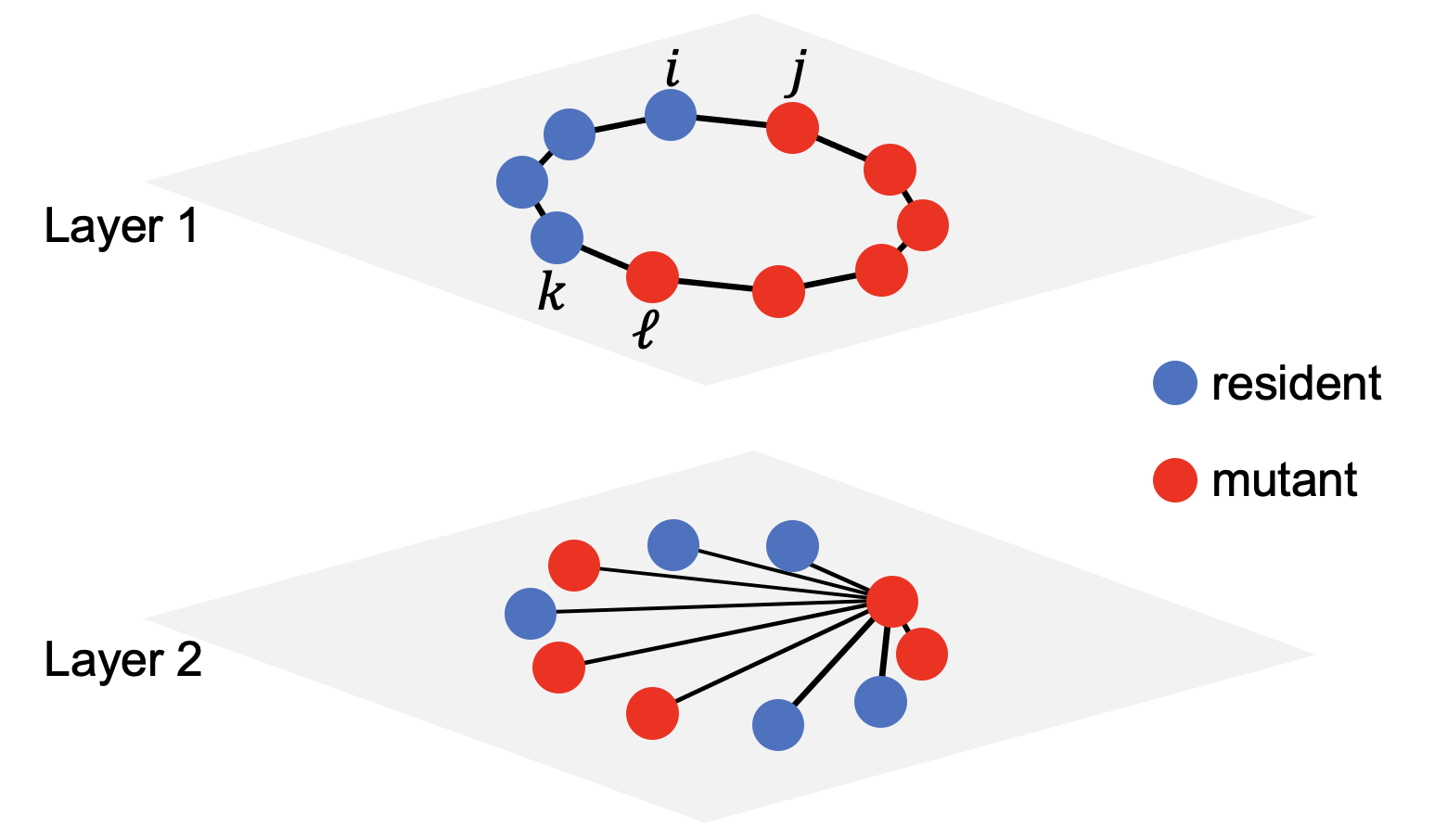}
\caption{A two-layer network in which layer 1 is the unweighted cycle graph and layer 2 is the unweighted star graph.
}
\label{fig:cyc-star}
\end{figure}

We distinguish among the following nine cases that are different in terms of the type of the individuals $i$, $j$, $k$, and $\ell$.

In the first case, $j$ and $\ell$, which by definition have the mutant type in layer 1, have the mutant type in layer 2, and $i$ and $k$, which by definition have the resident type in layer 1, have the resident type in layer 2. We reuse $N_1$, $N_2$, $N_3$, and $N_4$ defined in the proof of Lemma~3. Then, in a single time step of the original Bd process, $X_t$ increases by one with probability
\begin{equation}
p' = \frac{4r}{2r N_1 + (r+1)(N_2 + N_3) + 2N_4} \cdot \frac{1}{2} \cdot \frac{1}{2}
\label{eq:p'-cycle-1}
\end{equation}
and decreases by one with probability
\begin{equation}
q' = \frac{4} {2r N_1 + (r+1)(N_2 + N_3) + 2N_4} \cdot \frac{1}{2} \cdot \frac{1}{2}.
\label{eq:q'-cycle-1}
\end{equation}
This derivation is valid only if $1<X_0<N-1$. However, in fact, Eqs.~\eqref{eq:p'-cycle-1} and \eqref{eq:q'-cycle-1} also hold true for $X_0=1$ and $X_0=N-1$, respectively, although $j$ and $\ell$ are identical individuals when $X_0=1$, and $i$ and $k$ are identical individuals when $X_0=N-1$. Similarly, the probability that $X_t$ increases or decreases by one in the other eight cases is derived in the following text for $1<X_0<N-1$ but holds true for $X_0 = 1$ or $X_0 = N-1$ as well. Note that 
$X_0=1$ occurs only in the first (i.e., present), second, third, fourth, fifth, and sixth cases and that $X_0 = N-1$ occurs only in the first, second, fourth, fifth, seventh, and eighth cases.

By combining Eqs.~\eqref{eq:p'-cycle-1} and \eqref{eq:q'-cycle-1} with $p(\xi_t)/q(\xi_t) = p'/q'$ and $p(\xi_t) + q(\xi_t) = 1$, we obtain
\begin{align}
p(\xi_t) =& \frac{r}{r+1},\\
q(\xi_t) =& \frac{1}{r+1}.
\end{align}
Therefore, we obtain 
\begin{align}
E[Y_{t+1} | \mathcal{B}_t] =& p(\xi_t) r^{-(X_t+1)} + q(\xi_t) r^{-(X_t-1)} \notag\\
=& \left(\frac{r}{r+1} \cdot \frac{1}{r} +  \frac{1}{r+1} \cdot r  \right)   Y_t \notag\\
=&  Y_t.
\label{eq:martingale-eq-cycle-case1}
\end{align}

In the second case, the two individuals of the mutant type in layer 1 have the mutant type in layer 2, and the two individuals of the resident type in layer 1 have the mutant type in layer 2. Then, in a single time step of the original Bd process, $X_t$ increases by one with probability $p'$ and decreases by one with probability
\begin{equation}
q'' = \frac{2+2r} {2r N_1 + (r+1)(N_2 + N_3) + 2N_4} \cdot \frac{1}{2} \cdot \frac{1}{2}.
\label{eq:q'-cycle-2}
\end{equation}  
Therefore, we obtain
\begin{align}
p(\xi_t) =& \frac{2r}{3r+1},\\
q(\xi_t) =& \frac{r+1}{3r+1},
\end{align}
and
\begin{align}
E[Y_{t+1} | \mathcal{B}_t] =& \left(\frac{2r}{3r+1} \cdot \frac{1}{r} +  \frac{r+1}{3r+1} \cdot r  \right)   Y_t \notag\\
%
%
%
=& \left[1+\frac{(1-r)^2}{3r+1}\right] Y_t\notag\\
\geq & Y_t.
\end{align}

In the third case, the two individuals of the mutant type in layer 1 have the mutant type in layer 2, one individual of the resident type in layer 1 has the mutant type in layer 2, and the other individual of the resident type in layer 1 has the resident type in layer 2. Then, in a single time step of the original Bd process, $X_t$ increases by one with probability $p'$ and decreases by one with probability
\begin{equation}
q''' = \frac{3+r} {2r N_1 + (r+1)(N_2 + N_3) + 2N_4} \cdot \frac{1}{2} \cdot \frac{1}{2}.
\label{eq:q'-cycle-3}
\end{equation}  
Therefore, we obtain
\begin{align}
p(\xi_t) =& \frac{4r}{5r+3},\\
q(\xi_t) =& \frac{r+3}{5r+3},
\end{align}
and
\begin{align}
E[Y_{t+1} | \mathcal{B}_t] =& \left(\frac{4r}{5r+3} \cdot \frac{1}{r} +  \frac{r+3}{5r+3} \cdot r  \right)   Y_t \notag\\
%
%
=& \left[1+\frac{(1-r)^2}{5r+3}\right] Y_t\notag\\
\geq & Y_t.
\end{align}

In the fourth case, the two individuals of the mutant type in layer 1 have the resident type in layer 2, and the two individuals of the resident type in layer 1 have the resident type in layer 2. Then, in a single time step of the original Bd process, $X_t$ increases by one with probability $q''$ and decreases by one with probability $q'$. Therefore, we obtain
\begin{align}
p(\xi_t) =& \frac{r+1}{r+3},\\
q(\xi_t) =& \frac{2}{r+3},
\end{align}
and
\begin{align}
E[Y_{t+1} | \mathcal{B}_t] =& \left(\frac{r+1}{r+3} \cdot \frac{1}{r} +  \frac{2}{r+3} \cdot r  \right)   Y_t \notag\\
%
%
=& \left[1+\frac{(1-r)^2}{r^2+3r}\right] Y_t\notag\\
\geq & Y_t.
\end{align}

In the fifth case, the two individuals of the mutant type in layer 1 have the resident type in layer 2, and the two individuals of the resident type in layer 1 have the mutant type in layer 2. Then, in a single time step of the original Bd process, $X_t$ increases by one with probability $q''$ and decreases by one with the same probability $q''$. Therefore, we obtain $p(\xi_t)=q(\xi_t)=1/2$ and
\begin{align}
E[Y_{t+1} | \mathcal{B}_t] =& \left(\frac{1}{2} \cdot \frac{1}{r} +  \frac{1}{2} \cdot r  \right)   Y_t \notag\\
%
%
=& \left[1+\frac{(1-r)^2}{2r}\right] Y_t\notag\\
\geq & Y_t.
\end{align}

In the sixth case, the two individuals of the mutant type in layer 1 have the resident type in layer 2, one individual of the resident type in layer 1 has the mutant type in layer 2, and the other individual of the resident type in layer 1 has the resident type in layer 2. Then, in a single time step of the original Bd process, $X_t$ increases by one with probability $q''$ and decreases by one with probability $q'''$. Therefore, we obtain
\begin{align}
p(\xi_t) =& \frac{2r+2}{3r+5},\\
q(\xi_t) =& \frac{r+3}{3r+5},
\end{align}
and
\begin{align}
E[Y_{t+1} | \mathcal{B}_t] =& \left(\frac{2r+2}{3r+5} \cdot \frac{1}{r} +  \frac{r+3}{3r+5} \cdot r  \right)   Y_t \notag\\
%
%
=& \left[1+\frac{(r-1)^2(r+2)}{3r^2+5r}\right] Y_t\notag\\
\geq & Y_t.
\end{align}

In the seventh case, one individual of the mutant type in layer 1 has the mutant type in layer 2, the other individual of the mutant type in layer 1 has the resident type in layer two, and the two individuals of the resident type in layer 1 have the resident type in layer 2. Then, in a single time step of the original Bd process, $X_t$ increases by one with probability
\begin{equation}
p'' = \frac{3r+1} {2r N_1 + (r+1)(N_2 + N_3) + 2N_4} \cdot \frac{1}{2} \cdot \frac{1}{2}
\label{eq:p'-cycle-7}
\end{equation}
and decreases by one with probability $q'$. Therefore, we obtain
\begin{align}
p(\xi_t) =& \frac{3r+1}{3r+5},\\
q(\xi_t) =& \frac{4}{3r+5},
\end{align}
and
\begin{align}
E[Y_{t+1} | \mathcal{B}_t] =& \left(\frac{3r+1}{3r+5} \cdot \frac{1}{r} +  \frac{4}{3r+5} \cdot r  \right)   Y_t \notag\\
%
%
=& \left[1+\frac{(r-1)^2}{3r^2+5r}\right] Y_t\notag\\
\geq & Y_t.
\end{align}

In the eighth case, one individual of the mutant type in layer 1 has the mutant type in layer 2, the other individual of the mutant type in layer 1 has the resident type in layer 2, and the two individuals of the resident type in layer 1 have the mutant type in layer 2. Then, in a single time step of the original Bd process, $X_t$ increases by one with probability $p''$ and decreases by one with probability $q''$.
Therefore, we obtain
\begin{align}
p(\xi_t) =& \frac{3r+1}{5r+3},\\
q(\xi_t) =& \frac{2r+2}{5r+3},
\end{align}
and
\begin{align}
E[Y_{t+1} | \mathcal{B}_t] =& \left(\frac{3r+1}{5r+3} \cdot \frac{1}{r} +  \frac{2r+2}{5r+3} \cdot r  \right)   Y_t \notag\\
%
%
=& \left[1+\frac{(r-1)^2(2r+1)}{5r^2+3r}\right] Y_t\notag\\
\geq & Y_t.
\end{align}

In the ninth case, one individual of the mutant type in layer 1 has the mutant type in layer 2, the other individual of the mutant type in layer 1 has the resident type in layer 2, one individual of the resident type in layer 1 has the mutant type in layer 2, and the other individual of the resident type in layer 1 has the resident type in layer 2. Then, in a single time step of the original Bd process, $X_t$ increases by one with probability $p''$ and decreases by one with probability $q'''$. Therefore, we obtain
\begin{align}
p(\xi_t) =& \frac{3r+1}{4r+4},\\
q(\xi_t) =& \frac{r+3}{4r+4},
\end{align}
and
\begin{align}
E[Y_{t+1} | \mathcal{B}_t] =& \left(\frac{3r+1}{4r+4} \cdot \frac{1}{r} +  \frac{r+3}{4r+4} \cdot r  \right)   Y_t \notag\\
%
%
=& \left[1+\frac{(r-1)^2(r+1)}{4r^2+4r}\right] Y_t\notag\\
\geq & Y_t.
\end{align}

Therefore, $Y_t$ is a submartingale.

\section{Proof of Theorem~7}\label{cycle-thm-proof}

We observe that Eqs.~(4.7), (4.8), and (4.9) hold true in the present case as well. To exclude the equalities in Eq.~(4.9) for $1\le X_0 \le N-1$, we note that $E[Y_{t+1} | \mathcal{B}_t]=Y_t$ for $r\ne 1$ and $1\le X_0 \le N-1$ if and only if the first out of the nine cases introduced in the proof of Lemma 6, shown in section~\ref{submartingale-proof-cycle}, occurs. 

If the initial condition, $\xi_0$, satisfies one of the other eight cases, then we obtain Eq.~(4.10) if $r\ne 1$. Combination of $E[Y_{2} | \mathcal{B}_1] \ge Y_1$, which follows from Lemma~6, and Eq.~(4.10) yields $E[Y_{2} | \xi_0]>Y_0$.

Otherwise, i.e., if $\xi_0$ satisfies the first case, then we obtain $E[Y_{1} | \xi_0]=Y_0$ from Eq.~\eqref{eq:martingale-eq-cycle-case1}. In this case, we distinguish among the following three subcases.
First, if $2\le X_0 \le N-2$, then the network's state after the first state transition, $\xi_1$, satisfies one of the latter eight cases with probability 1. Therefore, we obtain $E[Y_{2} | \xi_1]>Y_1$ for $r\ne 1$, which is an adaptation of Eq.~(4.10). By combining $E[Y_{2} | \xi_1]>Y_1$ with $E[Y_{1} | \xi_0]=Y_0$, we obtain $E[Y_{2} | \xi_0]>Y_0$.
Second, if $X_0=1$, then $\xi_1$ satisfies one of the latter eight cases with probability $p(\xi_0)=r/(r+1)$.
Conditioned on this transition, we obtain $E[Y_{2} | \xi_1]>Y_1$ for $r\ne 1$. 
If a different $\xi_1$ is realized with probability $1 - p(\xi_0)$, then we obtain $E[Y_{2} | \xi_1] \ge Y_1$ owing to Lemma~6. Therefore, $E[Y_{2} | \xi_0]>Y_0$.
Third, if $X_0 = N-1$, then $\xi_1$ satisfies one of the latter eight cases with probability $q(\xi_0)=1/(r+1)$.
Conditioned on this transition, we obtain $E[Y_{2} | \xi_1]>Y_1$ for $r\ne 1$.
If a different $\xi_1$ is realized with probability $1 - q(\xi_0)$, then we obtain $E[Y_{2} | \xi_1] \ge Y_1$. Therefore, 
$E[Y_{2} | \xi_0]>Y_0$.

Because $E[Y_2 | \xi_0] > Y_0$ holds true in all the cases, we obtain
$E[Y_2 | \mathcal{B}_0] > Y_0$, which together with
$E[Y_{t+1} | \mathcal{B}_t] \ge Y_t$, $\forall t \in \{2, 3, \ldots \}$ leads to
Eq.~(4.7) with the strict inequality. Therefore, Eqs.~(4.8) and (4.9) hold true with the strict inequality when $r\neq 1$, proving that
the cycle graph layer in a two-layer network is a suppressor of selection.

\section{Proof of Lemma~8}\label{submartingale-proof-bipartite}

In this section, we show that $E[Y_{t+1} | \mathcal{B}_t] \ge Y_t$ for any $r>0$ when one layer is the complete bipartite graph.

After a time step, we obtain $\bm X_{t+1}=[X_{1,t}+1, X_{2,t}]$, $\bm X_{t+1}=[X_{1,t}, X_{2,t}+1]$, $\bm X_{t+1}=[X_{1,t}-1, X_{2,t}]$, or $\bm X_{t+1}=[X_{1,t}, X_{2,t}-1]$ because we count the time $t$ if and only if there is a change in the number of the mutants in the complete bipartite graph layer. We denote by $p_1(\xi_t)$, $p_2(\xi_t)$, $q_1(\xi_t)$, and $q_2(\xi_t)$ the probabilities with which $\bm X_{t+1}=[X_{1,t}+1, X_{2,t}]$, $\bm X_{t+1}=[X_{1,t}, X_{2,t}+1]$, $\bm X_{t+1}=[X_{1,t}-1, X_{2,t}]$, and $\bm X_{t+1}=[X_{1,t}, X_{2,t}-1]$, respectively. Note that $p_1(\xi_t)+p_2(\xi_t)+q_1(\xi_t)+q_2(\xi_t)=1$.

To calculate $p_1(\xi_t)$, $p_2(\xi_t)$, $q_1(\xi_t)$, and $q_2(\xi_t)$, we denote by $C_1$ the number of individuals in $V_1$ that have the mutant type in both layers, by $C_2$ the number of individuals in $V_1$ that have the mutant type in layer 1 and the resident type in layer 2, by $C_3$ the number of individuals in $V_1$ that have the resident type in layer 1 and the mutant type in layer 2, and by $C_4$ the number of individuals in $V_1$ that have the resident type in both layers. Note that $C_1 + C_2 + C_3 + C_4 = N_1$. Similarly, we also count the number of four different kinds of individuals in $V_2$ and denote them by $D_1$, $D_2$, $D_3$, and $D_4$, respectively. Then, $D_1 + D_2 + D_3 + D_4 = N_2$. After a time step of the original Bd process, we obtain $\bm X_{t+1}=[X_{1,t}+1, X_{2,t}]$ with probability
\begin{equation}
p_1' = \frac{2r D_1 + (r+1) D_2}{2r (C_1+D_1) + (r+1)(C_2 + C_3+D_2+D_3) + 2(C_4+D_4)} \cdot \frac{1}{2} \cdot \frac{C_3 + C_4}{N_1},
\label{eq:p1'}
\end{equation}
$\bm X_{t+1}=[X_{1,t}, X_{2,t}+1]$ with probability
\begin{equation}
p_2' = \frac{2r C_1 + (r+1) C_2}{2r (C_1+D_1) + (r+1)(C_2 + C_3+D_2+D_3) + 2(C_4+D_4)} \cdot \frac{1}{2} \cdot \frac{D_3 + D_4}{N_2},
\label{eq:p2'}
\end{equation}
$\bm X_{t+1}=[X_{1,t}-1, X_{2,t}]$ with probability
\begin{equation}
q_1' = \frac{(r+1) D_3+2D_4}{2r (C_1+D_1) + (r+1)(C_2 + C_3+D_2+D_3) + 2(C_4+D_4)} \cdot \frac{1}{2} \cdot \frac{C_1 + C_2}{N_1},
\label{eq:q1'}
\end{equation}
and $\bm X_{t+1}=[X_{1,t}, X_{2,t}-1]$ with probability
\begin{equation}
q_2' = \frac{(r+1) C_3+2C_4}{2r (C_1+D_1) + (r+1)(C_2 + C_3+D_2+D_3) + 2(C_4+D_4)} \cdot \frac{1}{2} \cdot \frac{D_1 + D_2}{N_2}.
\label{eq:q2'}
\end{equation}
By combining Eqs.~\eqref{eq:p1'},~\eqref{eq:p2'},~\eqref{eq:q1'}, and~\eqref{eq:q2'} with $p_1(\xi_t) / p_1' = p_2(\xi_t) / p_2' = q_1(\xi_t) / q_1' = q_2(\xi_t) / q_2'$ and $p_1(\xi_t)+p_2(\xi_t)+q_1(\xi_t)+q_2(\xi_t)=1$, we obtain
\begin{align}
p_1(\xi_t) =&\frac{1}{\delta N_1}\left[2rD_1+(1+r)D_2\right](C_3+C_4) ,\label{eq:p1}\\
p_2(\xi_t) =&\frac{1}{\delta N_2}\left[2rC_1+(1+r)C_2\right](D_3+D_4), \label{eq:p2}\\
q_1(\xi_t) =&\frac{1}{\delta N_1}\left[(1+r)D_3+2D_4\right](C_1+C_2) ,\label{eq:q1}\\
q_2(\xi_t) =&\frac{1}{\delta N_2}\left[(1+r)C_3+2C_4\right](D_1+D_2), \label{eq:q2} 
\end{align}
where
\begin{align}
\delta =& \frac{1}{N_1}\Bigl\{\left[2rD_1+(1+r)D_2\right](C_3+C_4)+\left[(1+r)D_3+2D_4\right](C_1+C_2)\Bigr\} \notag\\
& + \frac{1}{N_2}\Bigl\{\left[2rC_1+(1+r)C_2\right](D_3+D_4)+\left[(1+r)C_3+2C_4\right](D_1+D_2)\Bigr\}.
\end{align}
Therefore, we obtain
\begin{align}
E[Y_{t+1} | \mathcal{B}_t] &= p_1(\xi_t) h_1^{X_{1,t}+1}h_2^{X_{2,t}} + p_2(\xi_t) h_1^{X_{1,t}}h_2^{X_{2,t}+1} + q_1(\xi_t) h_1^{X_{1,t}-1}h_2^{X_{2,t}} + q_2(\xi_t) h_1^{X_{1,t}}h_2^{X_{2,t}-1} \notag\\
&=\left(p_1(\xi_t)h_1+p_2(\xi_t)h_2+\frac{q_1(\xi_t)}{h_1}+\frac{q_2(\xi_t)}{h_2}\right) Y_t \notag\\
&=\left(1+(1-r)^2(1+r)N_1N_2\Bigl[\delta N_1N_2r(N_1+N_2r)(N_2+N_1r)\Bigr]^{-1}\biggl\{(N_1+N_2r)\right. \notag\\
& \phantom{=\;\;}\left.\cdot\Bigl[rD_1C_3+(1+r)D_2C_3+D_2C_4\Bigr]+(N_2+N_1r)\Bigl[rC_1D_3+(1+r)C_2D_3+C_2D_4\Bigr]\biggr\}\right) Y_t \notag\\
&\ge Y_t.
\label{eq:yn-submartingale-bipartite}
\end{align}
Therefore, $Y_t$ is a submartingale.

\section{Proof of Eq.~(4.15)}

We distinguish among 12 cases that are different in terms of $C_1$, $C_2$, $C_3$, $C_4$, $D_1$, $D_2$, $D_3$, and $D_4$, at $t=0$. Recall that $C_i$ and $D_i$, where $i\in \{1, 2, 3, 4\}$, are defined in section~\ref{submartingale-proof-bipartite}.

To state the first case, we note that Eq.~\eqref{eq:yn-submartingale-bipartite} implies that $E[Y_{t+1} | \mathcal{B}_t]>Y_t$ for $r\ne 1$ and $1\le X_{1,0}+X_{2,0}\le N-1$ if and only if at least one of the following products, i.e., $D_1C_3$, $D_2C_3$, $D_2C_4$, $C_1D_3$, $C_2D_3$, and $C_2D_4$, is not zero. If $\xi_0$ satisfies this condition, then we obtain Eq.~(4.10) for $r\ne 1$. By combining $E[Y_{3} | \xi_2]\ge Y_2$ and $E[Y_{2} | \xi_1]\ge Y_1$, which follow from Lemma~8, and Eq.~(4.10), we obtain Eq.~(4.15).

As the second case, we assume that $C_2=C_3=C_4=D_1=D_2=D_3=0$, i.e., all the individuals in $V_1$ have the mutant type in both layers, and all the individuals in $V_2$ have the resident type in both layers, at $t=0$. Then, we obtain $E[Y_{1} | \xi_0]=Y_0$ from Eq.~\eqref{eq:yn-submartingale-bipartite}. The network's state after the first two state transitions, $\xi_2$, satisfies $(C_1,C_2,C_3,C_4,D_1,D_2,D_3,D_4)=(N_1-1,0,1,0,0,1,0,N_2-1)$, which leads to $D_2C_3\ne 0$, with probability 
\begin{equation}
p_2(\xi_0)q_1(\xi_1)+q_1(\xi_0)p_2(\xi^\prime_1)=\frac{2rN_1N_2}{(rN_1+N_2)^2},
\end{equation}
where $\xi_1$ satisfies $(C_1,C_2,C_3,C_4,D_1,D_2,D_3,D_4)=(N_1,0,0,0,0,1,0,N_2-1)$ and $\xi^\prime_1$ satisfies $(C_1,C_2,C_3,C_4,$ $D_1,D_2,D_3,D_4)=(N_1-1,0,1,0,0,0,0,N_2)$. Note that $p_1$, $p_2$, $q_1$, and $q_2$ are defined in Eqs.~\eqref{eq:p1}, \eqref{eq:p2}, \eqref{eq:q1}, and \eqref{eq:q2}, respectively. Conditioned on this transition, we obtain $E[Y_{3} | \xi_2]>Y_2$ for $r\ne 1$. If a different $\xi_2$ is realized with probability $1 - [p_2(\xi_0)q_1(\xi_1)+q_1(\xi_0)p_2(\xi^{\prime}_1)]$, then we obtain $E[Y_{3} | \xi_2] \ge Y_2$ owing to Lemma~8. By combining these results with $E[Y_{2} | \xi_1]\ge Y_1$, which follows from Lemma~8, we obtain Eq.~(4.15).

As the third case, we assume that $C_1=C_2=C_3=D_2=D_3=D_4=0$, i.e., all the individuals in $V_1$ have the resident type in both layers, and all the individuals in $V_2$ have the mutant type in both layers, at $t=0$. Then, we obtain $E[Y_{1} | \xi_0]=Y_0$ from Eq.~\eqref{eq:yn-submartingale-bipartite}. The network's state after the first two state transitions, $\xi_2$, satisfies $(C_1,C_2,C_3,C_4,D_1,D_2,D_3,D_4)=(0,1,0,N_1-1,N_2-1,0,1,0)$, which leads to $C_2D_3\ne 0$, with probability 
\begin{equation}
p_1(\xi_0)q_2(\xi_1)+q_2(\xi_0)p_1(\xi^{\prime}_1)=\frac{2rN_1N_2}{(N_1+rN_2)^2},
\end{equation}
where $\xi_1$ satisfies $(C_1,C_2,C_3,C_4,D_1,D_2,D_3,D_4)=(0,1,0,N_1-1,N_2,0,0,0)$ and $\xi^\prime_1$ satisfies $(C_1,C_2,C_3,C_4,$ $D_1,D_2,D_3,D_4)=(0,0,0,N_1,N_2-1,0,1,0)$. Conditioned on this transition, we obtain $E[Y_{3} | \xi_2]>Y_2$ for $r\ne 1$. If a different $\xi_2$ is realized with probability $1 - [p_1(\xi_0)q_2(\xi_1)+q_2(\xi_0)p_1(\xi^{\prime}_1)]$, then we obtain $E[Y_{3} | \xi_2] \ge Y_2$ owing to Lemma~8. By combining these results with $E[Y_{2} | \xi_1]\ge Y_1$, we obtain Eq.~(4.15).

As the fourth case, we assume that $C_2=C_3=C_4=D_2=D_3=0$ and $0<D_1<N_2$, i.e., all the individuals in $V_1$ have the mutant type in both layers, and each individual in $V_2$ has the same type (i.e., resident or mutant) in both layers but excluding $(D_1, D_4) = (0, N_2)$ and $(D_1, D_4) = (N_2, 0)$, at $t=0$. Then, we obtain $E[Y_{1} | \xi_0]=Y_0$ from Eq.~\eqref{eq:yn-submartingale-bipartite}. The network's state after the first state transition, $\xi_1$, satisfies $(C_1,C_2,C_3,C_4,D_1,D_2,D_3,D_4)=(N_1-1,0,1,0,D_1,0,0,N_2-D_1)$ with probability 
\begin{equation}\label{eq:q1-case4}
q_1(\xi_0)=\frac{N_2}{rN_1+N_2},
\end{equation}
which leads to $D_1C_3\ne 0$. Conditioned on this transition, we obtain $E[Y_{2} | \xi_1]>Y_1$ for $r\ne 1$. If a different $\xi_1$ is realized with probability $1 - q_1(\xi_0)$, then we obtain $E[Y_{2} | \xi_1] \ge Y_1$ owing to Lemma~8. By combining these results with $E[Y_{3} | \xi_2]\ge Y_2$, we obtain Eq.~(4.15).

As the fifth case, we assume that $C_1=C_2=C_3=D_2=D_3=0$ and $0<D_1<N_2$, i.e., all the individuals in $V_1$ have the resident type in both layers, and each individual in $V_2$ has the same type (i.e., resident or mutant) in both layers but excluding $(D_1, D_4) = (0, N_2)$ and $(D_1, D_4) = (N_2, 0)$, at $t=0$. Then, we obtain $E[Y_{1} | \xi_0]=Y_0$ from Eq.~\eqref{eq:yn-submartingale-bipartite}. The network's state after the first state transition, $\xi_1$, satisfies $(C_1,C_2,C_3,C_4,D_1,D_2,D_3,D_4)=(0,1,0,N_1-1,D_1,0,0,N_2-D_1)$ with probability 
\begin{equation}\label{eq:p1-case5}
p_1(\xi_0)=\frac{rN_2}{N_1+rN_2},
\end{equation}
which leads to $C_2D_4\ne 0$. Conditioned on this transition, we obtain $E[Y_{2} | \xi_1]>Y_1$ for $r\ne 1$. If a different $\xi_1$ is realized with probability $1 - p_1(\xi_0)$, then we obtain $E[Y_{2} | \xi_1] \ge Y_1$ owing to Lemma~8. By combining these results with $E[Y_{3} | \xi_2]\ge Y_2$, we obtain Eq.~(4.15).

As the sixth case, we assume that $C_2=C_3=D_2=D_3=D_4=0$ and $0<C_1<N_1$, i.e., each individual in $V_1$ has the same type (i.e., resident or mutant) in both layers but excluding $(C_1, C_4) = (0, N_1)$ and $(C_1, C_4) = (N_1, 0)$, and all the individuals in $V_2$ have the mutant type in both layers, at $t=0$. Then, we obtain $E[Y_{1} | \xi_0]=Y_0$ from Eq.~\eqref{eq:yn-submartingale-bipartite}. The network's state after the first state transition, $\xi_1$, satisfies $(C_1,C_2,C_3,C_4,D_1,D_2,D_3,D_4)=(C_1,0,0,N_1-C_1,N_2-1,0,1,0)$ with probability 
\begin{equation}\label{eq:q2-case6}
q_2(\xi_0)=\frac{N_1}{N_1+rN_2},
\end{equation}
which leads to $C_1D_3\ne 0$. Conditioned on this transition, we obtain $E[Y_{2} | \xi_1]>Y_1$ for $r\ne 1$. If a different $\xi_1$ is realized with probability $1 - q_2(\xi_0)$, then we obtain $E[Y_{2} | \xi_1] \ge Y_1$ owing to Lemma~8. By combining these results with $E[Y_{3} | \xi_2]\ge Y_2$, we obtain Eq.~(4.15).

As the seventh case, we assume that $C_2=C_3=D_1=D_2=D_3=0$ and $0<C_1<N_1$, i.e., each individual in $V_1$ has the same type (i.e., resident or mutant) in both layers but excluding $(C_1, C_4) = (0, N_1)$ and $(C_1, C_4) = (N_1, 0)$, and all the individuals in $V_2$ have the resident type in both layers, at $t=0$. Then, we obtain $E[Y_{1} | \xi_0]=Y_0$ from Eq.~\eqref{eq:yn-submartingale-bipartite}. The network's state after the first state transition, $\xi_1$, satisfies $(C_1,C_2,C_3,C_4,D_1,D_2,D_3,D_4)=(C_1,0,0,N_1-C_1,0,1,0,N_2-1)$ with probability 
\begin{equation}\label{eq:p2-case7}
p_2(\xi_0)=\frac{rN_1}{rN_1+N_2},
\end{equation}
which leads to $D_2C_4\ne 0$. Conditioned on this transition, we obtain $E[Y_{2} | \xi_1]>Y_1$ for $r\ne 1$. If a different $\xi_1$ is realized with probability $1 - p_2(\xi_0)$, then we obtain $E[Y_{2} | \xi_1] \ge Y_1$ owing to Lemma~8. By combining these results with $E[Y_{3} | \xi_2]\ge Y_2$, we obtain Eq.~(4.15).

As the eighth case, we assume that $C_2=C_3=D_2=D_3=0$, $0<C_1<N_1$, and $0<D_1<N_2$, i.e., each individual in $V_1$ and $V_2$ has the same type (i.e., resident or mutant) in both layers but excluding $(C_1, C_4, D_1, D_4) = (0, N_1, 0, N_2)$, $(C_1, C_4, D_1, D_4) = (N_1, 0, 0, N_2)$, $(C_1, C_4, D_1, D_4) = (0, N_1, N_2, 0)$, and $(C_1, C_4, D_1, D_4) = (N_1, 0, N_2, 0)$, at $t=0$. Then, we obtain $E[Y_{1} | \xi_0]=Y_0$ from Eq.~\eqref{eq:yn-submartingale-bipartite}. The network's state after the first state transition, $\xi_1$, satisfies $(C_1,C_2,C_3,C_4,D_1,D_2,D_3,D_4)=(C_1,0,0,N_1-C_1,D_1-1,0,1,N_2-D_1)$ with probability 
\begin{equation}
q_2(\xi_0)=\frac{N_1D_1(N_1-C_1)}{rN_1N_2(C_1+D_1)-(r+1)(N_1+N_2)C_1D_1+N^2_1D_1+N^2_2C_1},
\end{equation}
which leads to $C_1D_3\ne 0$. Conditioned on this transition, we obtain $E[Y_{2} | \xi_1]>Y_1$ for $r\ne 1$. If a different $\xi_1$ is realized with probability $1 - q_2(\xi_0)$, then we obtain $E[Y_{2} | \xi_1] \ge Y_1$ owing to Lemma~8. By combining these results with $E[Y_{3} | \xi_2]\ge Y_2$, we obtain Eq.~(4.15).

As the ninth case, we assume that $C_1=C_2=C_3=D_2=0$, $D_1>0$, $D_3>0$, and $D_1+D_3\le N_2$, i.e., all the individuals in $V_1$ have the resident type in both layers, and no individual in $V_2$ simultaneously has the mutant type in layer 1 and the resident type in layer 2 but excluding $D_1D_3=0$, at $t=0$. Then, we obtain $E[Y_{1} | \xi_0]=Y_0$ from Eq.~\eqref{eq:yn-submartingale-bipartite}. The network's state after the first state transition, $\xi_1$, satisfies $(C_1,C_2,C_3,C_4,D_1,D_2,D_3,D_4)=(0,1,0,N_1-1,D_1,0,D_3,N_2-D_1-D_3)$ with probability 
$p_1(\xi_0)$ given by Eq.~\eqref{eq:p1-case5},
which leads to $C_2D_3\ne 0$. Conditioned on this transition, we obtain $E[Y_{2} | \xi_1]>Y_1$ for $r\ne 1$. If a different $\xi_1$ is realized with probability $1 - p_1(\xi_0)$, then we obtain $E[Y_{2} | \xi_1] \ge Y_1$ owing to Lemma~8. By combining these results with $E[Y_{3} | \xi_2]\ge Y_2$, we obtain Eq.~(4.15).

As the tenth case, we assume that $C_2=C_3=C_4=D_3=0$, $D_2>0$, $D_4>0$, and $D_2+D_4\le N_2$, i.e., all the individuals in $V_1$ have the mutant type in both layers, and no individual in $V_2$ simultaneously has the resident type in layer 1 and the mutant type in layer 2 but excluding $D_2D_4=0$, at $t=0$. Then, we obtain $E[Y_{1} | \xi_0]=Y_0$ from Eq.~\eqref{eq:yn-submartingale-bipartite}. The network's state after the first state transition, $\xi_1$, satisfies $(C_1,C_2,C_3,C_4,D_1,D_2,D_3,D_4)=(N_1-1,0,1,0,N_2-D_2-D_4,D_2,0,D_4)$ with probability 
$q_1(\xi_0)$ given by Eq.~\eqref{eq:q1-case4},
which leads to $D_2C_3\ne 0$. Conditioned on this transition, we obtain $E[Y_{2} | \xi_1]>Y_1$ for $r\ne 1$. If a different $\xi_1$ is realized with probability $1 - q_1(\xi_0)$, then we obtain $E[Y_{2} | \xi_1] \ge Y_1$ owing to Lemma~8. By combining these results with $E[Y_{3} | \xi_2]\ge Y_2$, we obtain Eq.~(4.15).

As the eleventh case, we assume that $C_2=D_1=D_2=D_3=0$, $C_1>0$, $C_3>0$, and $C_1+C_3\le N_1$, i.e., no individual in $V_1$ simultaneously has the mutant type in layer 1 and the resident type in layer 2 but excluding $C_1C_3=0$, and all the individuals in $V_2$ have the resident type in both layers, at $t=0$. Then, we obtain $E[Y_{1} | \xi_0]=Y_0$ from Eq.~\eqref{eq:yn-submartingale-bipartite}. The network's state after the first state transition, $\xi_1$, satisfies $(C_1,C_2,C_3,C_4,D_1,D_2,D_3,D_4)=(C_1,0,C_3,N_1-C_1-C_3,0,1,0,N_2-1)$ with probability 
$p_2(\xi_0)$ given by Eq.~\eqref{eq:p2-case7},
which leads to $D_2C_3\ne 0$. Conditioned on this transition, we obtain $E[Y_{2} | \xi_1]>Y_1$ for $r\ne 1$. If a different $\xi_1$ is realized with probability $1 - p_2(\xi_0)$, then we obtain $E[Y_{2} | \xi_1] \ge Y_1$ owing to Lemma~8. By combining these results with $E[Y_{3} | \xi_2]\ge Y_2$, we obtain Eq.~(4.15).

As the twelfth case, we assume that $C_3=D_2=D_3=D_4=0$, $C_2>0$, $C_4>0$, and $C_2+C_4\le N_1$, i.e., no individual in $V_1$ simultaneously has the resident type in layer 1 and the mutant type in layer 2 but excluding $C_2C_4=0$, and all the individuals in $V_2$ have the mutant type in both layers, at $t=0$. Then, we obtain $E[Y_{1} | \xi_0]=Y_0$ from Eq.~\eqref{eq:yn-submartingale-bipartite}. The network's state after the first state transition, $\xi_1$, satisfies $(C_1,C_2,C_3,C_4,D_1,D_2,D_3,D_4)=(N_1-C_2-C_4,C_2,0,C_4,N_2-1,0,1,0)$ with probability 
$q_2(\xi_0)$ given by Eq.~\eqref{eq:q2-case6},
which leads to $C_2D_3\ne 0$. Conditioned on this transition, we obtain $E[Y_{2} | \xi_1]>Y_1$ for $r\ne 1$. If a different $\xi_1$ is realized with probability $1 - q_2(\xi_0)$, then we obtain $E[Y_{2} | \xi_1] \ge Y_1$ owing to Lemma~8. By combining these results with $E[Y_{3} | \xi_2]\ge Y_2$, we obtain Eq.~(4.15).

\begin{rmk}
All lemmas and theorems in the main text also hold true for model 2 with the proof being essentially unchanged. The only difference is that, for model 2, we do not multiply $1/2$ to $p'$s and $q'$s (e.g., Eqs.~\eqref{eq:p1'},~\eqref{eq:p2'},~\eqref{eq:q1'}, and~\eqref{eq:q2'}) because both layers are selected for reproduction in a single step in model 2.
\end{rmk}

\section{Derivation of the transition probability matrix for model 1 on two-layer networks composed of a complete graph layer and a star graph layer\label{complete-star-m1-derivation}}

Assume that the current state is $(1,1,i_1,i_2,i_3,i_4)$. There are nine types of events that can occur next.

In the first type of event, an individual that has the mutant type in layer 1 is selected as the parent, which occurs with probability $[2r(i_1+1)+(1+r)i_2]/[2r(i_1+1)+(1+r)(i_2+i_3)+2i_4]$, and the layer 1 is selected for reproduction with probability $1/2$. Then, we select a neighbor of the parent in layer 1 as the child, and the child has the resident type in layer 1 and the mutant type in layer 2, which occurs with probability $i_3/(N-1)$. Then, the child copies the parent's type in layer 1. The state after this event is $(1,1,i_1+1,i_2,i_3-1,i_4)$. Therefore, we obtain
\begin{equation}
p_{(1,1,i_1,i_2,i_3,i_4) \to (1,1,i_1+1,i_2,i_3-1,i_4)} = \frac{2r(i_1+1)+(1+r)i_2}{2r(i_1+1)+(1+r)(i_2+i_3)+2i_4}\cdot\frac{1}{2}\cdot\frac{i_3}{N-1}.
\label{eq:p1-complete-star-state1}
\end{equation}
The first row of Table~\ref{table-complete-star-state1} (except the header rows) represents this state transition event and Eq.~\eqref{eq:p1-complete-star-state1}.
The remaining seven types of events are listed in Table~\ref{table-complete-star-state1}. If any other event than the eight types shown in Table~\ref{table-complete-star-state1} occurs, the state remains unchanged. The probability of this case is 1 minus  the sum of all entries in the last column of Table~\ref{table-complete-star-state1}. 

\begin{table}[H]
\hspace*{-3em}
\begin{tabular}{ | w{c}{1.1em} | w{c}{1.1em} | w{c}{1em} | c | w{c}{1.2em} | w{c}{1em} | w{c}{1em} | w{c}{1em} | w{c}{1.3em} | c | M{10em} |} 
  \hline
  \multicolumn{4}{|c|}{parent} & \multirow{3}{*}{\shortstack{sele-\\cted\\ layer}} & \multicolumn{4}{c|}{child} & \multirow{3}{*}{state after transition} & \multirow{3}{*}{transition probability}\\
  \cline{1-4}
  \cline{6-9}
   \multicolumn{2}{|c|}{type} & \multirow{2}{*}{h/l} & \multirow{2}{*}{probability} & & \multicolumn{2}{c|}{type} & \multirow{2}{*}{h/l} & \multirow{2}{*}{prob.} & & \\
   \cline{1-2}
   \cline{6-7}      
    $\ell_1$ & $\ell_2$ &  &  &  & $\ell_1$ & $\ell_2$ &  &  &  & \\
  \hline
  \rule{0pt}{20pt}m & m/r & h/l & $\frac{2r(i_1+1)+(1+r)i_2}{2r(i_1+1)+(1+r)(i_2+i_3)+2i_4}$ & $\ell_1$ & r & m & l & $\frac{i_3}{N-1}$ & $(1,1,i_1+1,i_2,i_3-1,i_4)$ & $\frac{2r(i_1+1)+(1+r)i_2}{2r(i_1+1)+(1+r)(i_2+i_3)+2i_4}\cdot\frac{1}{2}\cdot\frac{i_3}{N-1}$  \\[2.5ex]
  \hline
  \rule{0pt}{20pt}m & m/r & h/l & $\frac{2r(i_1+1)+(1+r)i_2}{2r(i_1+1)+(1+r)(i_2+i_3)+2i_4}$ & $\ell_1$ & r & r & l & $\frac{i_4}{N-1}$ & $(1,1,i_1,i_2+1,i_3,i_4-1)$ & $\frac{2r(i_1+1)+(1+r)i_2}{2r(i_1+1)+(1+r)(i_2+i_3)+2i_4}\cdot\frac{1}{2}\cdot\frac{i_4}{N-1}$  \\[2.5ex]
  \hline
  \rule{0pt}{20pt}r & m/r & l & $\frac{(1+r)i_3+2i_4}{2r(i_1+1)+(1+r)(i_2+i_3)+2i_4}$ & $\ell_1$ & m & m & h & $\frac{1}{N-1}$ & $(0,1,i_1,i_2,i_3,i_4)$ & $\frac{(1+r)i_3+2i_4}{2r(i_1+1)+(1+r)(i_2+i_3)+2i_4}\cdot\frac{1}{2}\cdot\frac{1}{N-1}$ \\[2.5ex] 
  \hline
  \rule{0pt}{20pt}r & m/r & l & $\frac{(1+r)i_3+2i_4}{2r(i_1+1)+(1+r)(i_2+i_3)+2i_4}$ & $\ell_1$ & m & m & l & $\frac{i_1}{N-1}$ & $(1,1,i_1-1,i_2,i_3+1,i_4)$ & $\frac{(1+r)i_3+2i_4}{2r(i_1+1)+(1+r)(i_2+i_3)+2i_4}\cdot\frac{1}{2}\cdot\frac{i_1}{N-1}$ \\[2.5ex] 
  \hline
  \rule{0pt}{20pt}r & m/r & l & $\frac{(1+r)i_3+2i_4}{2r(i_1+1)+(1+r)(i_2+i_3)+2i_4}$ & $\ell_1$ & m & r & l & $\frac{i_2}{N-1}$ & $(1,1,i_1,i_2-1,i_3,i_4+1)$ & $\frac{(1+r)i_3+2i_4}{2r(i_1+1)+(1+r)(i_2+i_3)+2i_4}\cdot\frac{1}{2}\cdot\frac{i_2}{N-1}$ \\[2.5ex] 
  \hline
  \rule{0pt}{20pt}m & m & h & $\frac{2r}{2r(i_1+1)+(1+r)(i_2+i_3)+2i_4}$ & $\ell_2$ & m & r & l & $\frac{i_2}{N-1}$ & $(1,1,i_1+1,i_2-1,i_3,i_4)$ & $\frac{2r}{2r(i_1+1)+(1+r)(i_2+i_3)+2i_4}\cdot\frac{1}{2}\cdot\frac{i_2}{N-1}$ \\[2.5ex] 
  \hline
  \rule{0pt}{20pt}m & m & h & $\frac{2r}{2r(i_1+1)+(1+r)(i_2+i_3)+2i_4}$ & $\ell_2$ & r & r & l & $\frac{i_4}{N-1}$ & $(1,1,i_1,i_2,i_3+1,i_4-1)$ & $\frac{2r}{2r(i_1+1)+(1+r)(i_2+i_3)+2i_4}\cdot\frac{1}{2}\cdot\frac{i_4}{N-1}$ \\[2.5ex] 
  \hline
  \rule{0pt}{20pt}m/r & r & l & $\frac{(1+r)i_2+2i_4}{2r(i_1+1)+(1+r)(i_2+i_3)+2i_4}$ & $\ell_2$ & m & m & h & $1$ & $(1,0,i_1,i_2,i_3,i_4)$ & $\frac{(1+r)i_2+2i_4}{2r(i_1+1)+(1+r)(i_2+i_3)+2i_4}\cdot\frac{1}{2}$ \\[2.5ex] 
  \hline
\end{tabular}
\captionof{table}{\label{table-complete-star-state1}Eight types of state transition from state $(1,1,i_1,i_2,i_3,i_4)$ and their probabilities under model 1 on the two-layer network composed of a complete graph layer and a star graph layer.
``prob.'' abbreviates probability; ``m'' and ``r'' abbreviate the mutant and resident, respectively; ``h'' and ``l'' abbreviate the hub and leaf nodes, respectively; $\ell_1$ and $\ell_2$ represent layers 1 and 2, respectively.
}
\end{table}

\newpage
Assume that the current state is $(1,0,i_1,i_2,i_3,i_4)$. There are nine types of events that can occur next.
In the first type of event, an individual that has the mutant type in layer 1 is selected as the parent, which occurs with probability $[2ri_1+(1+r)(i_2+1)]/[2ri_1+(1+r)(i_2+i_3+1)+2i_4]$, and the layer 1 is selected for reproduction with probability $1/2$. Then, we select a neighbor of the parent in layer 1 as the child, and the child has the resident type in layer 1 and the mutant type in layer 2, which occurs with probability $i_3/(N-1)$. The state after this event is $(1,0,i_1+1,i_2,i_3-1,i_4)$. Therefore, we obtain
\begin{equation}
p_{(1,0,i_1,i_2,i_3,i_4) \to (1,0,i_1+1,i_2,i_3-1,i_4)} = \frac{2ri_1+(1+r)(i_2+1)}{2ri_1+(1+r)(i_2+i_3+1)+2i_4}\cdot\frac{1}{2}\cdot\frac{i_3}{N-1}.
\label{eq:p1-complete-star-state2}
\end{equation}
The first row of Table~\ref{table-complete-star-state2} (except the header rows) represents this state transition event and Eq.~\eqref{eq:p1-complete-star-state2}.
The remaining seven rows of the table represent the other types of state transition. If any other event than the eight types shown in Table~\ref{table-complete-star-state2} occurs, the state remains unchanged. The probability of this case is 1 minus  the sum of all entries in the last column of Table~\ref{table-complete-star-state2}. 

\begin{table}[H]
\hspace*{-3em}
\begin{tabular}{ | w{c}{1.1em} | w{c}{1.1em} | w{c}{1em} | c | w{c}{1.2em} | w{c}{1em} | w{c}{1em} | w{c}{1em} | w{c}{1.3em} | c | M{10em} |}
  \hline
  \multicolumn{4}{|c|}{parent} & \multirow{3}{*}{\shortstack{sele-\\cted\\ layer}} & \multicolumn{4}{c|}{child} & \multirow{3}{*}{state after transition} & \multirow{3}{*}{transition probability}\\
  \cline{1-4}
  \cline{6-9}
    \multicolumn{2}{|c|}{type} & \multirow{2}{*}{h/l} & \multirow{2}{*}{probability} & & \multicolumn{2}{c|}{type} & \multirow{2}{*}{h/l} & \multirow{2}{*}{prob.} & & \\
   \cline{1-2}
   \cline{6-7}      
    $\ell_1$ & $\ell_2$ &  &  &  & $\ell_1$ & $\ell_2$ &  &  &  & \\
  \hline
  \rule{0pt}{20pt}m & m/r & h/l & $\frac{2ri_1+(1+r)(i_2+1)}{2ri_1+(1+r)(i_2+i_3+1)+2i_4}$ & $\ell_1$ & r & m & l & $\frac{i_3}{N-1}$ & $(1,0,i_1+1,i_2,i_3-1,i_4)$ & $\frac{2ri_1+(1+r)(i_2+1)}{2ri_1+(1+r)(i_2+i_3+1)+2i_4}\cdot\frac{1}{2}\cdot\frac{i_3}{N-1}$  \\[2.5ex] 
  \hline
  \rule{0pt}{20pt}m & m/r & h/l & $\frac{2ri_1+(1+r)(i_2+1)}{2ri_1+(1+r)(i_2+i_3+1)+2i_4}$ & $\ell_1$ & r & r & l & $\frac{i_4}{N-1}$ & $(1,0,i_1,i_2+1,i_3,i_4-1)$ & $\frac{2ri_1+(1+r)(i_2+1)}{2ri_1+(1+r)(i_2+i_3+1)+2i_4}\cdot\frac{1}{2}\cdot\frac{i_4}{N-1}$  \\[2.5ex] 
  \hline
  \rule{0pt}{20pt}r & m/r & l & $\frac{(1+r)i_3+2i_4}{2ri_1+(1+r)(i_2+i_3+1)+2i_4}$ & $\ell_1$ & m & r & h & $\frac{1}{N-1}$ & $(0,0,i_1,i_2,i_3,i_4)$ & $\frac{(1+r)i_3+2i_4}{2ri_1+(1+r)(i_2+i_3+1)+2i_4}\cdot\frac{1}{2}\cdot\frac{1}{N-1}$ \\[2.5ex] 
  \hline
  \rule{0pt}{20pt}r & m/r & l & $\frac{(1+r)i_3+2i_4}{2ri_1+(1+r)(i_2+i_3+1)+2i_4}$ & $\ell_1$ & m & m & l & $\frac{i_1}{N-1}$ & $(1,0,i_1-1,i_2,i_3+1,i_4)$ & $\frac{(1+r)i_3+2i_4}{2ri_1+(1+r)(i_2+i_3+1)+2i_4}\cdot\frac{1}{2}\cdot\frac{i_1}{N-1}$ \\[2.5ex] 
  \hline
  \rule{0pt}{20pt}r & m/r & l & $\frac{(1+r)i_3+2i_4}{2ri_1+(1+r)(i_2+i_3+1)+2i_4}$ & $\ell_1$ & m & r & l & $\frac{i_2}{N-1}$ & $(1,0,i_1,i_2-1,i_3,i_4+1)$ & $\frac{(1+r)i_3+2i_4}{2ri_1+(1+r)(i_2+i_3+1)+2i_4}\cdot\frac{1}{2}\cdot\frac{i_2}{N-1}$ \\[2.5ex] 
  \hline
  \rule{0pt}{20pt}m/r & m & l & $\frac{2ri_1+(1+r)i_3}{2ri_1+(1+r)(i_2+i_3+1)+2i_4}$ & $\ell_2$ & m & r & h & $1$ & $(1,1,i_1,i_2,i_3,i_4)$ & $\frac{2ri_1+(1+r)i_3}{2ri_1+(1+r)(i_2+i_3+1)+2i_4}\cdot\frac{1}{2}$ \\[2.5ex] 
  \hline
  \rule{0pt}{20pt}m & r & h & $\frac{1+r}{2ri_1+(1+r)(i_2+i_3+1)+2i_4}$ & $\ell_2$ & r & m & l & $\frac{i_3}{N-1}$ & $(1,0,i_1,i_2,i_3-1,i_4+1)$ & $\frac{1+r}{2ri_1+(1+r)(i_2+i_3+1)+2i_4}\cdot\frac{1}{2}\cdot\frac{i_3}{N-1}$ \\[2.5ex] 
  \hline
  \rule{0pt}{20pt}m & r & h & $\frac{1+r}{2ri_1+(1+r)(i_2+i_3+1)+2i_4}$ & $\ell_2$ & m & m & l & $\frac{i_1}{N-1}$ & $(1,0,i_1-1,i_2+1,i_3,i_4)$ & $\frac{1+r}{2ri_1+(1+r)(i_2+i_3+1)+2i_4}\cdot\frac{1}{2}\cdot\frac{i_1}{N-1}$ \\[2.5ex] 
  \hline
\end{tabular}
\captionof{table}{\label{table-complete-star-state2}Eight types of state transition from state $(1,0,i_1,i_2,i_3,i_4)$ under model 1 on the two-layer network composed of a complete graph layer and a star graph layer.}
\end{table}
\newpage
Assume that the current state is $(0,1,i_1,i_2,i_3,i_4)$. There are nine types of events that can occur next.
%
%
Table~\ref{table-complete-star-state3} shows the eight types of events that accompany a state transition and their probabilities. If any other event than these eight types occurs, the state remains unchanged. The probability of this case is 1 minus the sum of all entries in the last column of Table~\ref{table-complete-star-state3}. 

\begin{table}[H]
\hspace*{-3em}
\begin{tabular}{ | w{c}{1.1em} | w{c}{1.1em} | w{c}{1em} | c | w{c}{1.2em} | w{c}{1em} | w{c}{1em} | w{c}{1em} | w{c}{1.3em} | c | M{10em} |}
  \hline
  \multicolumn{4}{|c|}{parent} & \multirow{3}{*}{\shortstack{sele-\\cted\\ layer}} & \multicolumn{4}{c|}{child} & \multirow{3}{*}{state after transition} & \multirow{3}{*}{transition probability}\\
  \cline{1-4}
  \cline{6-9}
    \multicolumn{2}{|c|}{type} & \multirow{2}{*}{h/l} & \multirow{2}{*}{probability} & & \multicolumn{2}{c|}{type} & \multirow{2}{*}{h/l} & \multirow{2}{*}{prob.} & & \\
   \cline{1-2}
   \cline{6-7}      
    $\ell_1$ & $\ell_2$ &  &  &  & $\ell_1$ & $\ell_2$ &  &  &  & \\
  \hline
  \rule{0pt}{20pt}m & m/r & l & $\frac{2ri_1+(1+r)i_2}{2ri_1+(1+r)(i_2+i_3+1)+2i_4}$ & $\ell_1$ & r & m & h & $\frac{1}{N-1}$ & $(1,1,i_1,i_2,i_3,i_4)$ & $\frac{2ri_1+(1+r)i_2}{2ri_1+(1+r)(i_2+i_3+1)+2i_4}\cdot\frac{1}{2}\cdot\frac{1}{N-1}$  \\[2.5ex] 
  \hline
  \rule{0pt}{20pt}m & m/r & l & $\frac{2ri_1+(1+r)i_2}{2ri_1+(1+r)(i_2+i_3+1)+2i_4}$ & $\ell_1$ & r & m & l & $\frac{i_3}{N-1}$ & $(0,1,i_1+1,i_2,i_3-1,i_4)$ & $\frac{2ri_1+(1+r)i_2}{2ri_1+(1+r)(i_2+i_3+1)+2i_4}\cdot\frac{1}{2}\cdot\frac{i_3}{N-1}$  \\[2.5ex] 
  \hline
  \rule{0pt}{20pt}m & m/r & l & $\frac{2ri_1+(1+r)i_2}{2ri_1+(1+r)(i_2+i_3+1)+2i_4}$ & $\ell_1$ & r & r & l & $\frac{i_4}{N-1}$ & $(0,1,i_1,i_2+1,i_3,i_4-1)$ & $\frac{2ri_1+(1+r)i_2}{2ri_1+(1+r)(i_2+i_3+1)+2i_4}\cdot\frac{1}{2}\cdot\frac{i_4}{N-1}$ \\[2.5ex] 
  \hline
  \rule{0pt}{20pt}r & m/r & h/l & $\frac{(1+r)(i_3+1)+2i_4}{2ri_1+(1+r)(i_2+i_3+1)+2i_4}$ & $\ell_1$ & m & m & l & $\frac{i_1}{N-1}$ & $(0,1,i_1-1,i_2,i_3+1,i_4)$ & $\frac{(1+r)(i_3+1)+2i_4}{2ri_1+(1+r)(i_2+i_3+1)+2i_4}\cdot\frac{1}{2}\cdot\frac{i_1}{N-1}$ \\[2.5ex] 
  \hline
  \rule{0pt}{20pt}r & m/r & h/l & $\frac{(1+r)(i_3+1)+2i_4}{2ri_1+(1+r)(i_2+i_3+1)+2i_4}$ & $\ell_1$ & m & r & l & $\frac{i_2}{N-1}$ & $(0,1,i_1,i_2-1,i_3,i_4+1)$ & $\frac{(1+r)(i_3+1)+2i_4}{2ri_1+(1+r)(i_2+i_3+1)+2i_4}\cdot\frac{1}{2}\cdot\frac{i_2}{N-1}$ \\[2.5ex] 
  \hline
  \rule{0pt}{20pt}r & m & h & $\frac{1+r}{2ri_1+(1+r)(i_2+i_3+1)+2i_4}$ & $\ell_2$ & m & r & l & $\frac{i_2}{N-1}$ & $(0,1,i_1+1,i_2-1,i_3,i_4)$ & $\frac{1+r}{2ri_1+(1+r)(i_2+i_3+1)+2i_4}\cdot\frac{1}{2}\cdot\frac{i_2}{N-1}$ \\[2.5ex] 
  \hline
  \rule{0pt}{20pt}r & m & h & $\frac{1+r}{2ri_1+(1+r)(i_2+i_3+1)+2i_4}$ & $\ell_2$ & r & r & l & $\frac{i_4}{N-1}$ & $(0,1,i_1,i_2,i_3+1,i_4-1)$ & $\frac{1+r}{2ri_1+(1+r)(i_2+i_3+1)+2i_4}\cdot\frac{1}{2}\cdot\frac{i_4}{N-1}$ \\[2.5ex] 
  \hline
  \rule{0pt}{20pt}m/r & r & l & $\frac{(1+r)i_2+2i_4}{2ri_1+(1+r)(i_2+i_3+1)+2i_4}$ & $\ell_2$ & r & m & h & $1$ & $(0,0,i_1,i_2,i_3,i_4)$ & $\frac{(1+r)i_2+2i_4}{2ri_1+(1+r)(i_2+i_3+1)+2i_4}\cdot\frac{1}{2}$ \\[2.5ex] 
  \hline
\end{tabular}
\captionof{table}{\label{table-complete-star-state3}Eight types of state transition from state $(0,1,i_1,i_2,i_3,i_4)$ under model 1 on the two-layer network composed of a complete graph layer and a star graph layer.}
\end{table}
\newpage
Assume that the current state is $(0,0,i_1,i_2,i_3,i_4)$. There are nine types of events that can occur next.
%
%
Table~\ref{table-complete-star-state4} shows the eight types of events that accompany a state transition and their probabilities. If any other event than these eight types occurs, the state remains unchanged. The probability of this case is 
1 minus  the sum of all entries in the last column of Table~\ref{table-complete-star-state4}.  

\begin{table}[H]
\hspace*{-3em}
\begin{tabular}{ | w{c}{1.1em} | w{c}{1.1em} | w{c}{1em} | c | w{c}{1.2em} | w{c}{1em} | w{c}{1em} | w{c}{1em} | w{c}{1.3em} | c | M{10em} |} 
  \hline
  \multicolumn{4}{|c|}{parent} & \multirow{3}{*}{\shortstack{sele-\\cted\\ layer}} & \multicolumn{4}{c|}{child} & \multirow{3}{*}{state after transition} & \multirow{3}{*}{transition probability}\\
  \cline{1-4}
  \cline{6-9}
    \multicolumn{2}{|c|}{type} & \multirow{2}{*}{h/l} & \multirow{2}{*}{probability} & & \multicolumn{2}{c|}{type} & \multirow{2}{*}{h/l} & \multirow{2}{*}{prob.} & & \\
   \cline{1-2}
   \cline{6-7}      
    $\ell_1$ & $\ell_2$ &  &  &  & $\ell_1$ & $\ell_2$ &  &  &  & \\
  \hline
  \rule{0pt}{20pt}m & m/r & l & $\frac{2ri_1+(1+r)i_2}{2ri_1+(1+r)(i_2+i_3)+2(i_4+1)}$ & $\ell_1$ & r & r & h & $\frac{1}{N-1}$ & $(1,0,i_1,i_2,i_3,i_4)$ & $\frac{2ri_1+(1+r)i_2}{2ri_1+(1+r)(i_2+i_3)+2(i_4+1)}\cdot\frac{1}{2}\cdot\frac{1}{N-1}$  \\[2.5ex] 
  \hline
  \rule{0pt}{20pt}m & m/r & l & $\frac{2ri_1+(1+r)i_2}{2ri_1+(1+r)(i_2+i_3)+2(i_4+1)}$ & $\ell_1$ & r & m & 1 & $\frac{i_3}{N-1}$ & $(0,0,i_1+1,i_2,i_3-1,i_4)$ & $\frac{2ri_1+(1+r)i_2}{2ri_1+(1+r)(i_2+i_3)+2(i_4+1)}\cdot\frac{1}{2}\cdot\frac{i_3}{N-1}$  \\[2.5ex] 
  \hline
  \rule{0pt}{20pt}m & m/r & l & $\frac{2ri_1+(1+r)i_2}{2ri_1+(1+r)(i_2+i_3)+2(i_4+1)}$ & $\ell_1$ & r & r & l & $\frac{i_4}{N-1}$ & $(0,0,i_1,i_2+1,i_3,i_4-1)$ & $\frac{2ri_1+(1+r)i_2}{2ri_1+(1+r)(i_2+i_3)+2(i_4+1)}\cdot\frac{1}{2}\cdot\frac{i_4}{N-1}$ \\[2.5ex] 
  \hline
  \rule{0pt}{20pt}r & m/r & h/l & $\frac{(1+r)i_3+2(i_4+1)}{2ri_1+(1+r)(i_2+i_3)+2(i_4+1)}$ & $\ell_1$ & m & m & l & $\frac{i_1}{N-1}$ & $(0,0,i_1-1,i_2,i_3+1,i_4)$ & $\frac{(1+r)i_3+2(i_4+1)}{2ri_1+(1+r)(i_2+i_3)+2(i_4+1)}\cdot\frac{1}{2}\cdot\frac{i_1}{N-1}$ \\[2.5ex] 
  \hline
  \rule{0pt}{20pt}r & m/r & h/l & $\frac{(1+r)i_3+2(i_4+1)}{2ri_1+(1+r)(i_2+i_3)+2(i_4+1)}$ & $\ell_1$ & m & r & l & $\frac{i_2}{N-1}$ & $(0,0,i_1,i_2-1,i_3,i_4+1)$ & $\frac{(1+r)i_3+2(i_4+1)}{2ri_1+(1+r)(i_2+i_3)+2(i_4+1)}\cdot\frac{1}{2}\cdot\frac{i_2}{N-1}$ \\[2.5ex] 
  \hline
  \rule{0pt}{20pt}m/r & m & l & $\frac{2ri_1+(1+r)i_3}{2ri_1+(1+r)(i_2+i_3)+2(i_4+1)}$ & $\ell_2$ & r & r & h & $1$ & $(0,1,i_1,i_2,i_3,i_4)$ & $\frac{2ri_1+(1+r)i_3}{2ri_1+(1+r)(i_2+i_3)+2(i_4+1)}\cdot\frac{1}{2}$ \\[2.5ex] 
  \hline
  \rule{0pt}{20pt}r & r & h & $\frac{2}{2ri_1+(1+r)(i_2+i_3)+2(i_4+1)}$ & $\ell_2$ & m & m & l & $\frac{i_1}{N-1}$ & $(0,0,i_1-1,i_2+1,i_3,i_4)$ & $\frac{2}{2ri_1+(1+r)(i_2+i_3)+2(i_4+1)}\cdot\frac{1}{2}\cdot\frac{i_1}{N-1}$ \\[2.5ex] 
  \hline
  \rule{0pt}{20pt}r & r & h & $\frac{2}{2ri_1+(1+r)(i_2+i_3)+2(i_4+1)}$ & $\ell_2$ & r & m & l & $\frac{i_3}{N-1}$ & $(0,0,i_1,i_2,i_3-1,i_4+1)$ & $\frac{2}{2ri_1+(1+r)(i_2+i_3)+2(i_4+1)}\cdot\frac{1}{2}\cdot\frac{i_3}{N-1}$ \\[2.5ex] 
  \hline
\end{tabular}
\captionof{table}{\label{table-complete-star-state4}Eight types of state transition from state $(0,0,i_1,i_2,i_3,i_4)$ under model 1 on the two-layer network composed of a complete graph layer and a star graph layer.}
\end{table}

\newpage

\section{Derivation of the transition probability matrix for model 1 on the coupled star graph\label{star-star-m1-derivation}}

Assume that the current state is $(1,1,1,1,i_1,i_2,i_3,i_4)$. There are seven types of events that can occur next.
In the first type of event, the individual that is the hub node in layer 1 is selected as the parent, which occurs with probability $2r/[2r(i_1+2)+(1+r)(i_2+i_3)+2i_4]$, and the layer 1 is selected with probability $1/2$. Then, we select a neighbor of the parent in layer 1 as the child, and the child that is a leaf node in both layers has the resident type in layer 1 and the mutant type in layer 2, which occurs with probability $i_3/(N-1)$. The state after this event is $(1,1,1,1,i_1+1,i_2,i_3-1,i_4)$. Therefore, we obtain
\begin{equation}
p_{(1,1,1,1,i_1,i_2,i_3,i_4) \to (1,1,1,1,i_1+1,i_2,i_3-1,i_4)} = \frac{2r}{2r(i_1+2)+(1+r)(i_2+i_3)+2i_4}\cdot\frac{1}{2}\cdot\frac{i_3}{N-1}.
\label{eq:p1-star-star-state1}
\end{equation}
The first row of Table~\ref{table-star-star-state1} (except the header rows) represents this state transition event and Eq.~\eqref{eq:p1-star-star-state1}.
The remaining five types of events are listed in Table~\ref{table-star-star-state1}. If any other event than the six types shown in Table~\ref{table-star-star-state1} occurs, the state remains unchanged. The probability of this case is 1 minus the sum of all entries in the last column of Table~\ref{table-star-star-state1}.  

\begin{table}[H]
\hspace*{-2.5em}
\begin{tabular}{ | w{c}{1.1em} | w{c}{1.1em} | M{2.5em} | c | w{c}{1.1em} | w{c}{1.1em} | w{c}{1.1em} | M{2.5em} | w{c}{1.3em} | M{7em} | M{10em} |} 
  \hline
  \multicolumn{4}{|c|}{parent} & \multirow{3}{*}{layer} & \multicolumn{4}{c|}{child} & \multirow{3}{*}{\shortstack{state after\\ transition}} & \multirow{3}{*}{transition probability}\\
  \cline{1-4}
  \cline{6-9}
    \multicolumn{2}{|c|}{type} & \multirow{2}{*}{h/l} & \multirow{2}{*}{probability} & & \multicolumn{2}{c|}{type} & \multirow{2}{*}{h/l} & \multirow{2}{*}{prob.} & & \\
   \cline{1-2}
   \cline{6-7}      
    $\ell_1$ & $\ell_2$ &  &  &  & $\ell_1$ & $\ell_2$ &  &  &  & \\
  \hline
  \rule{0pt}{20pt}m & m & h in $\ell_1$ & $\frac{2r}{2r(i_1+2)+(1+r)(i_2+i_3)+2i_4}$ & $\ell_1$ & r & m & l in $\ell_1$, $\ell_2$ & $\frac{i_3}{N-1}$ & $(1,1,1,1,i_1+$\par$1,i_2,i_3-1,i_4)$ & $\frac{2r}{2r(i_1+2)+(1+r)(i_2+i_3)+2i_4}\cdot\frac{1}{2}\cdot\frac{i_3}{N-1}$  \\[2.5ex] 
  \hline
  \rule{0pt}{20pt}m & m & h in $\ell_1$ & $\frac{2r}{2r(i_1+2)+(1+r)(i_2+i_3)+2i_4}$ & $\ell_1$ & r & r & l in $\ell_1$, $\ell_2$ & $\frac{i_4}{N-1}$ & $(1,1,1,1,i_1,$\par$i_2+1,i_3,i_4-1)$ & $\frac{2r}{2r(i_1+2)+(1+r)(i_2+i_3)+2i_4}\cdot\frac{1}{2}\cdot\frac{i_4}{N-1}$  \\[2.5ex] 
  \hline
  \rule{0pt}{20pt}r & m/r & l in $\ell_1$, $\ell_2$ & $\frac{(1+r)i_3+2i_4}{2r(i_1+2)+(1+r)(i_2+i_3)+2i_4}$ & $\ell_1$ & m & m & h in $\ell_1$ & $1$ & $(0,1,1,1,i_1,$\par$i_2,i_3,i_4)$ & $\frac{(1+r)i_3+2i_4}{2r(i_1+2)+(1+r)(i_2+i_3)+2i_4}\cdot\frac{1}{2}$ \\[2.5ex] 
  \hline
  \rule{0pt}{20pt}m & m & h in $\ell_2$ & $\frac{2r}{2r(i_1+2)+(1+r)(i_2+i_3)+2i_4}$ & $\ell_2$ & m & r & l in $\ell_1$, $\ell_2$ & $\frac{i_2}{N-1}$ & $(1,1,1,1,i_1+$\par$1,i_2-1,i_3,i_4)$ & $\frac{2r}{2r(i_1+2)+(1+r)(i_2+i_3)+2i_4}\cdot\frac{1}{2}\cdot\frac{i_2}{N-1}$ \\[2.5ex] 
  \hline
  \rule{0pt}{20pt}m & m & h in $\ell_2$ & $\frac{2r}{2r(i_1+2)+(1+r)(i_2+i_3)+2i_4}$ & $\ell_2$ & r & r & l in $\ell_1$, $\ell_2$ & $\frac{i_4}{N-1}$ & $(1,1,1,1,i_1,$\par$i_2,i_3+1,i_4-1)$ & $\frac{2r}{2r(i_1+2)+(1+r)(i_2+i_3)+2i_4}\cdot\frac{1}{2}\cdot\frac{i_4}{N-1}$ \\[2.5ex] 
  \hline
  \rule{0pt}{20pt}m/r & r & l in $\ell_1$, $\ell_2$ & $\frac{(1+r)i_2+2i_4}{2r(i_1+2)+(1+r)(i_2+i_3)+2i_4}$ & $\ell_2$ & m & m & h in $\ell_2$ & $1$ & $(1,1,1,0,i_1,$\par$i_2,i_3,i_4)$ & $\frac{(1+r)i_2+2i_4}{2r(i_1+2)+(1+r)(i_2+i_3)+2i_4}\cdot\frac{1}{2}$ \\[2.5ex] 
  \hline
\end{tabular}
\captionof{table}{\label{table-star-star-state1}Six types of state transition from state $(1,1,1,1,i_1,i_2,i_3,i_4)$ under model 1 on the coupled star graph.}
\end{table}
\newpage
Assume that the current state is $(1,1,1,0,i_1,i_2,i_3,i_4)$. There are eight types of events that can occur next.
%
%
Table~\ref{table-star-star-state2} shows the seven types of events that accompany a state transition and their probabilities. If any other event than these seven types occurs, the state remains unchanged. The probability of this case is 1 minus the sum of all entries in the last column of Table~\ref{table-star-star-state2}.  

\begin{table}[H]
\hspace*{-3.5em}
\begin{tabular}{ | w{c}{1.1em} | w{c}{1.1em} | M{2.5em} | c | w{c}{1.1em} | w{c}{1.1em} | w{c}{1.1em} | M{2.5em} | w{c}{1.3em} | M{7em} | M{11em} |}  
  \hline
  \multicolumn{4}{|c|}{parent} & \multirow{3}{*}{layer} & \multicolumn{4}{c|}{child} & \multirow{3}{*}{\shortstack{state after\\ transition}} & \multirow{3}{*}{transition probability}\\
  \cline{1-4}
  \cline{6-9}
    \multicolumn{2}{|c|}{type} & \multirow{2}{*}{h/l} & \multirow{2}{*}{probability} & & \multicolumn{2}{c|}{type} & \multirow{2}{*}{h/l} & \multirow{2}{*}{prob.} & & \\
   \cline{1-2}
   \cline{6-7}      
    $\ell_1$ & $\ell_2$ &  &  &  & $\ell_1$ & $\ell_2$ &  &  &  & \\
  \hline
  \rule{0pt}{20pt}m & m & h in $\ell_1$ & $\frac{2r}{2r(i_1+1)+(1+r)(i_2+i_3+1)+2i_4}$ & $\ell_1$ & r & m & l in $\ell_1$, $\ell_2$ & $\frac{i_3}{N-1}$ & $(1,1,1,0,i_1+$\par$1,i_2,i_3-1,i_4)$ & $\frac{2r}{2r(i_1+1)+(1+r)(i_2+i_3+1)+2i_4}\cdot\frac{1}{2}\cdot\frac{i_3}{N-1}$  \\[2.5ex] 
  \hline
  \rule{0pt}{20pt}m & m & h in $\ell_1$ & $\frac{2r}{2r(i_1+1)+(1+r)(i_2+i_3+1)+2i_4}$ & $\ell_1$ & r & r & l in $\ell_1$, $\ell_2$ & $\frac{i_4}{N-1}$ & $(1,1,1,0,i_1,$\par$i_2+1,i_3,i_4-1)$ & $\frac{2r}{2r(i_1+1)+(1+r)(i_2+i_3+1)+2i_4}\cdot\frac{1}{2}\cdot\frac{i_4}{N-1}$  \\[2.5ex] 
  \hline
  \rule{0pt}{20pt}m/r & m & l in $\ell_1$, $\ell_2$ & $\frac{(1+r)i_3+2i_4}{2r(i_1+1)+(1+r)(i_2+i_3+1)+2i_4}$ & $\ell_1$ & m & m & h in $\ell_1$ & $1$ & $(0,1,1,0,i_1,$\par$i_2,i_3,i_4)$ & $\frac{(1+r)i_3+2i_4}{2r(i_1+1)+(1+r)(i_2+i_3+1)+2i_4}\cdot\frac{1}{2}$ \\[2.5ex] 
  \hline
  \rule{0pt}{20pt}r & m/r & l in $\ell_2$ & $\frac{2r(i_1+1)+(1+r)i_3}{2r(i_1+1)+(1+r)(i_2+i_3+1)+2i_4}$ & $\ell_2$ & m & r & h in $\ell_2$ & $1$ & $(1,1,1,1,i_1,$\par$i_2,i_3,i_4)$ & $\frac{2r(i_1+1)+(1+r)i_3}{2r(i_1+1)+(1+r)(i_2+i_3+1)+2i_4}\cdot\frac{1}{2}$ \\[2.5ex] 
  \hline
  \rule{0pt}{20pt}m & r & h in $\ell_2$ & $\frac{1+r}{2r(i_1+1)+(1+r)(i_2+i_3+1)+2i_4}$ & $\ell_2$ & m & m & h in $\ell_1$ & $\frac{1}{N-1}$ & $(1,0,1,0,i_1,$\par$i_2,i_3,i_4)$ & $\frac{1+r}{2r(i_1+1)+(1+r)(i_2+i_3+1)+2i_4}\cdot\frac{1}{2}\cdot\frac{1}{N-1}$ \\[2.5ex] 
  \hline
  \rule{0pt}{20pt}m & r & h in $\ell_2$ & $\frac{1+r}{2r(i_1+1)+(1+r)(i_2+i_3+1)+2i_4}$ & $\ell_2$ & m & m & l in $\ell_1$, $\ell_2$ & $\frac{i_1}{N-1}$ & $(1,1,1,0,i_1-$\par$1,i_2+1,i_3,i_4)$ & $\frac{1+r}{2r(i_1+1)+(1+r)(i_2+i_3+1)+2i_4}\cdot\frac{1}{2}\cdot\frac{i_1}{N-1}$ \\[2.5ex] 
  \hline
  \rule{0pt}{20pt}m & r & h in $\ell_2$ & $\frac{1+r}{2r(i_1+1)+(1+r)(i_2+i_3+1)+2i_4}$ & $\ell_2$ & r & m & l in $\ell_1$, $\ell_2$ & $\frac{i_3}{N-1}$ & $(1,1,1,0,i_1,$\par$i_2,i_3-1,i_4+1)$ & $\frac{1+r}{2r(i_1+1)+(1+r)(i_2+i_3+1)+2i_4}\cdot\frac{1}{2}\cdot\frac{i_3}{N-1}$ \\[2.5ex] 
  \hline
\end{tabular}
\captionof{table}{\label{table-star-star-state2}Seven types of state transition from state $(1,1,1,0,i_1,i_2,i_3,i_4)$ under model 1 on the coupled star graph.}
\end{table}
\newpage
Assume that the current state is $(1,1,0,1,i_1,i_2,i_3,i_4)$. There are eight types of events that can occur next.
%
%
Table~\ref{table-star-star-state3} shows the seven types of events that accompany a state transition and their probabilities. If any other event than these seven types occurs, the state remains unchanged. The probability of this case is 1 minus the sum of all entries in the last column of Table~\ref{table-star-star-state3}. 

\begin{table}[H]
\hspace*{-3.5em}
\begin{tabular}{ | w{c}{1.1em} | w{c}{1.1em} | M{2.5em} | c | w{c}{1.1em} | w{c}{1.1em} | w{c}{1.1em} | M{2.5em} | w{c}{1.3em} | M{7em} | M{11em} |} 
  \hline
  \multicolumn{4}{|c|}{parent} & \multirow{3}{*}{layer} & \multicolumn{4}{c|}{child} & \multirow{3}{*}{\shortstack{state after\\ transition}} & \multirow{3}{*}{transition probability}\\
  \cline{1-4}
  \cline{6-9}
    \multicolumn{2}{|c|}{type} & \multirow{2}{*}{h/l} & \multirow{2}{*}{probability} & & \multicolumn{2}{c|}{type} & \multirow{2}{*}{h/l} & \multirow{2}{*}{prob.} & & \\
   \cline{1-2}
   \cline{6-7}      
    $\ell_1$ & $\ell_2$ &  &  &  & $\ell_1$ & $\ell_2$ &  &  &  & \\
  \hline
  \rule{0pt}{20pt}m & m & h in $\ell_1$ & $\frac{2r}{2r(i_1+1)+(1+r)(i_2+i_3+1)+2i_4}$ & $\ell_1$ & r & m & h in $\ell_2$ & $\frac{1}{N-1}$ & $(1,1,1,1,i_1,$\par$i_2,i_3,i_4)$ & $\frac{2r}{2r(i_1+1)+(1+r)(i_2+i_3+1)+2i_4}\cdot\frac{1}{2}\cdot\frac{1}{N-1}$  \\[2.5ex] 
  \hline
  \rule{0pt}{20pt}m & m & h in $\ell_1$ & $\frac{2r}{2r(i_1+1)+(1+r)(i_2+i_3+1)+2i_4}$ & $\ell_1$ & r & m & l in $\ell_1$, $\ell_2$ & $\frac{i_3}{N-1}$ & $(1,1,0,1,i_1+$\par$1,i_2,i_3-1,i_4)$ & $\frac{2r}{2r(i_1+1)+(1+r)(i_2+i_3+1)+2i_4}\cdot\frac{1}{2}\cdot\frac{i_3}{N-1}$  \\[2.5ex] 
  \hline
  \rule{0pt}{20pt}m & m & h in $\ell_1$ & $\frac{2r}{2r(i_1+1)+(1+r)(i_2+i_3+1)+2i_4}$ & $\ell_1$ & r & r & l in $\ell_1$, $\ell_2$ & $\frac{i_4}{N-1}$ & $(1,1,0,1,i_1,$\par$i_2+1,i_3,i_4-1)$ & $\frac{2r}{2r(i_1+1)+(1+r)(i_2+i_3+1)+2i_4}\cdot\frac{1}{2}\cdot\frac{i_4}{N-1}$ \\[2.5ex] 
  \hline
  \rule{0pt}{20pt}r & m/r & l in $\ell_1$ & $\frac{(1+r)(i_3+1)+2i_4}{2r(i_1+1)+(1+r)(i_2+i_3+1)+2i_4}$ & $\ell_1$ & m & m & h in $\ell_1$ & $1$ & $(0,1,0,1,i_1,$\par$i_2,i_3,i_4)$ & $\frac{(1+r)(i_3+1)+2i_4}{2r(i_1+1)+(1+r)(i_2+i_3+1)+2i_4}\cdot\frac{1}{2}$ \\[2.5ex] 
  \hline
  \rule{0pt}{20pt}r & m & h in $\ell_2$ & $\frac{1+r}{2r(i_1+1)+(1+r)(i_2+i_3+1)+2i_4}$ & $\ell_2$ & m & r & l in $\ell_1$, $\ell_2$ & $\frac{i_2}{N-1}$ & $(1,1,0,1,i_1+$\par$1,i_2-1,i_3,i_4)$ & $\frac{1+r}{2r(i_1+1)+(1+r)(i_2+i_3+1)+2i_4}\cdot\frac{1}{2}\cdot\frac{i_2}{N-1}$ \\[2.5ex] 
  \hline
  \rule{0pt}{20pt}r & m & h in $\ell_2$ & $\frac{1+r}{2r(i_1+1)+(1+r)(i_2+i_3+1)+2i_4}$ & $\ell_2$ & r & r & l in $\ell_1$, $\ell_2$ & $\frac{i_4}{N-1}$ & $(1,1,0,1,i_1,$\par$i_2,i_3+1,i_4-1)$ & $\frac{1+r}{2r(i_1+1)+(1+r)(i_2+i_3+1)+2i_4}\cdot\frac{1}{2}\cdot\frac{i_4}{N-1}$ \\[2.5ex] 
  \hline
  \rule{0pt}{20pt}m/r & r & l in $\ell_1$, $\ell_2$ & $\frac{(1+r)i_2+2i_4}{2r(i_1+1)+(1+r)(i_2+i_3+1)+2i_4}$ & $\ell_2$ & r & m & h in $\ell_2$ & $1$ & $(1,1,0,0,i_1,$\par$i_2,i_3,i_4)$ & $\frac{(1+r)i_2+2i_4}{2r(i_1+1)+(1+r)(i_2+i_3+1)+2i_4}\cdot\frac{1}{2}$ \\[2.5ex] 
  \hline
\end{tabular}
\captionof{table}{\label{table-star-star-state3}Seven types of state transition from state $(1,1,0,1,i_1,i_2,i_3,i_4)$ under model 1 on the coupled star graph.}
\end{table}
\newpage
Assume that the current state is $(1,1,0,0,i_1,i_2,i_3,i_4)$. There are nine types of events that can occur next.
%
%
Table~\ref{table-star-star-state4} shows the eight types of events that accompany a state transition and their probabilities. If any other event than these eight types occurs, the state remains unchanged. The probability of this case is 1 minus the sum of all entries in the last column of Table~\ref{table-star-star-state4}. 

\begin{table}[H]
\hspace*{-4em}
\begin{tabular}{ | w{c}{1.1em} | w{c}{1.1em} | M{2.5em} | c | w{c}{1.1em} | w{c}{1.1em} | w{c}{1.1em} | M{2.5em} | w{c}{1.3em} | M{7em} | M{11.6em} |} 
  \hline
  \multicolumn{4}{|c|}{parent} & \multirow{3}{*}{layer} & \multicolumn{4}{c|}{child} & \multirow{3}{*}{\shortstack{state after\\ transition}} & \multirow{3}{*}{transition probability}\\
  \cline{1-4}
  \cline{6-9}
    \multicolumn{2}{|c|}{type} & \multirow{2}{*}{h/l} & \multirow{2}{*}{probability} & & \multicolumn{2}{c|}{type} & \multirow{2}{*}{h/l} & \multirow{2}{*}{prob.} & & \\
   \cline{1-2}
   \cline{6-7}      
    $\ell_1$ & $\ell_2$ &  &  &  & $\ell_1$ & $\ell_2$ &  &  &  & \\
  \hline
  \rule{0pt}{20pt}m & m & h in $\ell_1$ & $\frac{2r}{2r(i_1+1)+(1+r)(i_2+i_3)+2(i_4+1)}$ & $\ell_1$ & r & r & h in $\ell_2$ & $\frac{1}{N-1}$ & $(1,1,1,0,i_1,$\par$i_2,i_3,i_4)$ & $\frac{2r}{2r(i_1+1)+(1+r)(i_2+i_3)+2(i_4+1)}\cdot\frac{1}{2}\cdot\frac{1}{N-1}$  \\[2.5ex] 
  \hline
  \rule{0pt}{20pt}m & m & h in $\ell_1$ & $\frac{2r}{2r(i_1+1)+(1+r)(i_2+i_3)+2(i_4+1)}$ & $\ell_1$ & r & m & l in $\ell_1$, $\ell_2$ & $\frac{i_3}{N-1}$ & $(1,1,0,0,i_1+$\par$1,i_2,i_3-1,i_4)$ & $\frac{2r}{2r(i_1+1)+(1+r)(i_2+i_3)+2(i_4+1)}\cdot\frac{1}{2}\cdot\frac{i_3}{N-1}$  \\[2.5ex] 
  \hline
  \rule{0pt}{20pt}m & m & h in $\ell_1$ & $\frac{2r}{2r(i_1+1)+(1+r)(i_2+i_3)+2(i_4+1)}$ & $\ell_1$ & r & r & l in $\ell_1$, $\ell_2$ & $\frac{i_4}{N-1}$ & $(1,1,0,0,i_1,$\par$i_2+1,i_3,i_4-1)$ & $\frac{2r}{2r(i_1+1)+(1+r)(i_2+i_3)+2(i_4+1)}\cdot\frac{1}{2}\cdot\frac{i_4}{N-1}$ \\[2.5ex] 
  \hline
  \rule{0pt}{20pt}r & m/r & l in $\ell_1$ & $\frac{(1+r)i_3+2(i_4+1)}{2r(i_1+1)+(1+r)(i_2+i_3)+2(i_4+1)}$ & $\ell_1$ & m & m & h in $\ell_1$ & $1$ & $(0,1,0,0,i_1,$\par$i_2,i_3,i_4)$ & $\frac{(1+r)i_3+2(i_4+1)}{2r(i_1+1)+(1+r)(i_2+i_3)+2(i_4+1)}\cdot\frac{1}{2}$ \\[2.5ex] 
  \hline
  \rule{0pt}{20pt}m/r & m & l in $\ell_2$ & $\frac{2r(i_1+1)+(1+r)i_3}{2r(i_1+1)+(1+r)(i_2+i_3)+2(i_4+1)}$ & $\ell_2$ & r & r & h in $\ell_2$ & $1$ & $(1,1,0,1,i_1,$\par$i_2,i_3,i_4)$ & $\frac{2r(i_1+1)+(1+r)i_3}{2r(i_1+1)+(1+r)(i_2+i_3)+2(i_4+1)}\cdot\frac{1}{2}$ \\[2.5ex] 
  \hline
  \rule{0pt}{20pt}r & r & h in $\ell_2$ & $\frac{2}{2r(i_1+1)+(1+r)(i_2+i_3)+2(i_4+1)}$ & $\ell_2$ & m & m & h in $\ell_1$ & $\frac{1}{N-1}$ & $(1,0,0,0,i_1,$\par$i_2,i_3,i_4)$ & $\frac{2}{2r(i_1+1)+(1+r)(i_2+i_3)+2(i_4+1)}\cdot\frac{1}{2}\cdot\frac{1}{N-1}$ \\[2.5ex] 
  \hline
  \rule{0pt}{20pt}r & r & h in $\ell_2$ & $\frac{2}{2r(i_1+1)+(1+r)(i_2+i_3)+2(i_4+1)}$ & $\ell_2$ & m & m & l in $\ell_1$, $\ell_2$ & $\frac{i_1}{N-1}$ & $(1,1,0,0,i_1-$\par$1,i_2+1,i_3,i_4)$ & $\frac{2}{2r(i_1+1)+(1+r)(i_2+i_3)+2(i_4+1)}\cdot\frac{1}{2}\cdot\frac{i_1}{N-1}$ \\[2.5ex] 
  \hline
  \rule{0pt}{20pt}r & r & h in $\ell_2$ & $\frac{2}{2r(i_1+1)+(1+r)(i_2+i_3)+2(i_4+1)}$ & $\ell_2$ & r & m & l in $\ell_1$, $\ell_2$ & $\frac{i_3}{N-1}$ & $(1,1,0,0,i_1,$\par$i_2,i_3-1,i_4+1)$ & $\frac{2}{2r(i_1+1)+(1+r)(i_2+i_3)+2(i_4+1)}\cdot\frac{1}{2}\cdot\frac{i_3}{N-1}$ \\[2.5ex] 
  \hline
\end{tabular}
\captionof{table}{\label{table-star-star-state4}Eight types of state transition from state $(1,1,0,0,i_1,i_2,i_3,i_4)$ under model 1 on the coupled star graph.}
\end{table}
\newpage
Assume that the current state is $(1,0,1,1,i_1,i_2,i_3,i_4)$. There are eight types of events that can occur next.
%
%
Table~\ref{table-star-star-state5} shows the seven types of events that accompany a state transition and their probabilities. If any other event than these seven types occurs, the state remains unchanged. The probability of this case is 1 minus the sum of all entries in the last column of Table~\ref{table-star-star-state5}. 

\begin{table}[H]
\hspace*{-3.5em}
\begin{tabular}{ | w{c}{1.1em} | w{c}{1.1em} | M{2.5em} | c | w{c}{1.1em} | w{c}{1.1em} | w{c}{1.1em} | M{2.5em} | w{c}{1.3em} | M{7em} | M{11em} |} 
  \hline
  \multicolumn{4}{|c|}{parent} & \multirow{3}{*}{layer} & \multicolumn{4}{c|}{child} & \multirow{3}{*}{\shortstack{state after\\ transition}} & \multirow{3}{*}{transition probability}\\
  \cline{1-4}
  \cline{6-9}
    \multicolumn{2}{|c|}{type} & \multirow{2}{*}{h/l} & \multirow{2}{*}{probability} & & \multicolumn{2}{c|}{type} & \multirow{2}{*}{h/l} & \multirow{2}{*}{prob.} & & \\
   \cline{1-2}
   \cline{6-7}      
    $\ell_1$ & $\ell_2$ &  &  &  & $\ell_1$ & $\ell_2$ &  &  &  & \\
  \hline
  \rule{0pt}{20pt}m & r & h in $\ell_1$ & $\frac{1+r}{2r(i_1+1)+(1+r)(i_2+i_3+1)+2i_4}$ & $\ell_1$ & r & m & l in $\ell_1$, $\ell_2$ & $\frac{i_3}{N-1}$ & $(1,0,1,1,i_1+$\par$1,i_2,i_3-1,i_4)$ & $\frac{1+r}{2r(i_1+1)+(1+r)(i_2+i_3+1)+2i_4}\cdot\frac{1}{2}\cdot\frac{i_3}{N-1}$  \\[2.5ex] 
  \hline
  \rule{0pt}{20pt}m & r & h in $\ell_1$ & $\frac{1+r}{2r(i_1+1)+(1+r)(i_2+i_3+1)+2i_4}$ & $\ell_1$ & r & r & l in $\ell_1$, $\ell_2$ & $\frac{i_4}{N-1}$ & $(1,0,1,1,i_1,$\par$i_2+1,i_3,i_4-1)$ & $\frac{1+r}{2r(i_1+1)+(1+r)(i_2+i_3+1)+2i_4}\cdot\frac{1}{2}\cdot\frac{i_4}{N-1}$  \\[2.5ex] 
  \hline
  \rule{0pt}{20pt}r & m/r & l in $\ell_1$, $\ell_2$ & $\frac{(1+r)i_3+2i_4}{2r(i_1+1)+(1+r)(i_2+i_3+1)+2i_4}$ & $\ell_1$ & m & r & h in $\ell_1$ & $1$ & $(0,0,1,1,i_1,$\par$i_2,i_3,i_4)$ & $\frac{(1+r)i_3+2i_4}{2r(i_1+1)+(1+r)(i_2+i_3+1)+2i_4}\cdot\frac{1}{2}$ \\[2.5ex] 
  \hline
  \rule{0pt}{20pt}m & m & h in $\ell_2$ & $\frac{2r}{2r(i_1+1)+(1+r)(i_2+i_3+1)+2i_4}$ & $\ell_2$ & m & r & h in $\ell_1$ & $\frac{1}{N-1}$ & $(1,1,1,1,i_1,$\par$i_2,i_3,i_4)$ & $\frac{2r}{2r(i_1+1)+(1+r)(i_2+i_3+1)+2i_4}\cdot\frac{1}{2}\cdot\frac{1}{N-1}$ \\[2.5ex] 
  \hline
  \rule{0pt}{20pt}m & m & h in $\ell_2$ & $\frac{2r}{2r(i_1+1)+(1+r)(i_2+i_3+1)+2i_4}$ & $\ell_2$ & m & r & l in $\ell_1$, $\ell_2$ & $\frac{i_2}{N-1}$ & $(1,0,1,1,i_1+$\par$1,i_2-1,i_3,i_4)$ & $\frac{2r}{2r(i_1+1)+(1+r)(i_2+i_3+1)+2i_4}\cdot\frac{1}{2}\cdot\frac{i_2}{N-1}$ \\[2.5ex] 
  \hline
  \rule{0pt}{20pt}m & m & h in $\ell_2$ & $\frac{2r}{2r(i_1+1)+(1+r)(i_2+i_3+1)+2i_4}$ & $\ell_2$ & r & r & l in $\ell_1$, $\ell_2$ & $\frac{i_4}{N-1}$ & $(1,0,1,1,i_1,$\par$i_2,i_3+1,i_4-1)$ & $\frac{2r}{2r(i_1+1)+(1+r)(i_2+i_3+1)+2i_4}\cdot\frac{1}{2}\cdot\frac{i_4}{N-1}$ \\[2.5ex] 
  \hline
  \rule{0pt}{20pt}m/r & r & l in $\ell_2$ & $\frac{(1+r)(i_2+1)+2i_4}{2r(i_1+1)+(1+r)(i_2+i_3+1)+2i_4}$ & $\ell_2$ & m & m & h in $\ell_2$ & $1$ & $(1,0,1,0,i_1,$\par$i_2,i_3,i_4)$ & $\frac{(1+r)(i_2+1)+2i_4}{2r(i_1+1)+(1+r)(i_2+i_3+1)+2i_4}\cdot\frac{1}{2}$ \\[2.5ex] 
  \hline
\end{tabular}
\captionof{table}{\label{table-star-star-state5}Seven types of state transition from state $(1,0,1,1,i_1,i_2,i_3,i_4)$ under model 1 on the coupled star graph.}
\end{table}
\newpage
Assume that the current state is $(1,0,1,0,i_1,i_2,i_3,i_4)$. There are seven types of events that can occur next.
%
%
Table~\ref{table-star-star-state6} shows the six types of events that accompany a state transition and their probabilities. If any other event than these six types occurs, the state remains unchanged. The probability of this case is 1 minus the sum of all entries in the last column of Table~\ref{table-star-star-state6}. 

\begin{table}[H]
\hspace*{-2.5em}
\begin{tabular}{ | w{c}{1.1em} | w{c}{1.1em} | M{2.5em} | c | w{c}{1.1em} | w{c}{1.1em} | w{c}{1.1em} | M{2.5em} | w{c}{1.3em} | M{7em} | M{10em} |}  
  \hline
  \multicolumn{4}{|c|}{parent} & \multirow{3}{*}{layer} & \multicolumn{4}{c|}{child} & \multirow{3}{*}{\shortstack{state after\\ transition}} & \multirow{3}{*}{transition probability}\\
  \cline{1-4}
  \cline{6-9}
    \multicolumn{2}{|c|}{type} & \multirow{2}{*}{h/l} & \multirow{2}{*}{probability} & & \multicolumn{2}{c|}{type} & \multirow{2}{*}{h/l} & \multirow{2}{*}{prob.} & & \\
   \cline{1-2}
   \cline{6-7}      
    $\ell_1$ & $\ell_2$ &  &  &  & $\ell_1$ & $\ell_2$ &  &  &  & \\
  \hline
  \rule{0pt}{20pt}m & r & h in $\ell_1$ & $\frac{1+r}{2ri_1+(1+r)(i_2+i_3+2)+2i_4}$ & $\ell_1$ & r & m & l in $\ell_1$, $\ell_2$ & $\frac{i_3}{N-1}$ & $(1,0,1,0,i_1+$\par$1,i_2,i_3-1,i_4)$ & $\frac{1+r}{2ri_1+(1+r)(i_2+i_3+2)+2i_4}\cdot\frac{1}{2}\cdot\frac{i_3}{N-1}$  \\[2.5ex] 
  \hline
  \rule{0pt}{20pt}m & r & h in $\ell_1$ & $\frac{1+r}{2ri_1+(1+r)(i_2+i_3+2)+2i_4}$ & $\ell_1$ & r & r & l in $\ell_1$, $\ell_2$ & $\frac{i_4}{N-1}$ & $(1,0,1,0,i_1,$\par$i_2+1,i_3,i_4-1)$ & $\frac{1+r}{2ri_1+(1+r)(i_2+i_3+2)+2i_4}\cdot\frac{1}{2}\cdot\frac{i_4}{N-1}$  \\[2.5ex] 
  \hline
  \rule{0pt}{20pt}r & m/r & l in $\ell_1$, $\ell_2$ & $\frac{(1+r)i_3+2i_4}{2ri_1+(1+r)(i_2+i_3+2)+2i_4}$ & $\ell_1$ & m & r & h in $\ell_1$ & $1$ & $(0,0,1,0,i_1,$\par$i_2,i_3,i_4)$ & $\frac{(1+r)i_3+2i_4}{2ri_1+(1+r)(i_2+i_3+2)+2i_4}\cdot\frac{1}{2}$ \\[2.5ex] 
  \hline
  \rule{0pt}{20pt}m/r & m & l in $\ell_1$, $\ell_2$ & $\frac{2ri_1+(1+r)i_3}{2ri_1+(1+r)(i_2+i_3+2)+2i_4}$ & $\ell_2$ & m & r & h in $\ell_2$ & $1$ & $(1,0,1,1,i_1,$\par$i_2,i_3,i_4)$ & $\frac{2ri_1+(1+r)i_3}{2ri_1+(1+r)(i_2+i_3+2)+2i_4}\cdot\frac{1}{2}$ \\[2.5ex] 
  \hline
  \rule{0pt}{20pt}m & r & h in $\ell_2$ & $\frac{1+r}{2ri_1+(1+r)(i_2+i_3+2)+2i_4}$ & $\ell_2$ & m & m & l in $\ell_1$, $\ell_2$ & $\frac{i_1}{N-1}$ & $(1,0,1,0,i_1-$\par$1,i_2+1,i_3,i_4)$ & $\frac{1+r}{2ri_1+(1+r)(i_2+i_3+2)+2i_4}\cdot\frac{1}{2}\cdot\frac{i_1}{N-1}$ \\[2.5ex] 
  \hline
  \rule{0pt}{20pt}m & r & h in $\ell_2$ & $\frac{1+r}{2ri_1+(1+r)(i_2+i_3+2)+2i_4}$ & $\ell_2$ & r & m & l in $\ell_1$, $\ell_2$ & $\frac{i_3}{N-1}$ & $(1,0,1,0,i_1,$\par$i_2,i_3-1,i_4+1)$ & $\frac{1+r}{2ri_1+(1+r)(i_2+i_3+2)+2i_4}\cdot\frac{1}{2}\cdot\frac{i_3}{N-1}$ \\[2.5ex] 
  \hline
\end{tabular}
\captionof{table}{\label{table-star-star-state6}Six types of state transition from state $(1,0,1,0,i_1,i_2,i_3,i_4)$ under model 1 on the coupled star graph.}
\end{table}
\newpage
Assume that the current state is $(1,0,0,1,i_1,i_2,i_3,i_4)$. There are nine types of events that can occur next.
%
%
Table~\ref{table-star-star-state7} shows the eight types of events that accompany a state transition and their probabilities. If any other event than these eight types occurs, the state remains unchanged. The probability of this case is 1 minus the sum of all entries in the last column of Table~\ref{table-star-star-state7}. 

\begin{table}[H]
\hspace*{-2.5em}
\begin{tabular}{ | w{c}{1.1em} | w{c}{1.1em} | M{2.5em} | c | w{c}{1.1em} | w{c}{1.1em} | w{c}{1.1em} | M{2.5em} | w{c}{1.3em} | M{7em} | M{10em} |} 
  \hline
  \multicolumn{4}{|c|}{parent} & \multirow{3}{*}{layer} & \multicolumn{4}{c|}{child} & \multirow{3}{*}{\shortstack{state after\\ transition}} & \multirow{3}{*}{transition probability}\\
  \cline{1-4}
  \cline{6-9}
    \multicolumn{2}{|c|}{type} & \multirow{2}{*}{h/l} & \multirow{2}{*}{probability} & & \multicolumn{2}{c|}{type} & \multirow{2}{*}{h/l} & \multirow{2}{*}{prob.} & & \\
   \cline{1-2}
   \cline{6-7}      
    $\ell_1$ & $\ell_2$ &  &  &  & $\ell_1$ & $\ell_2$ &  &  &  & \\
  \hline
  \rule{0pt}{20pt}m & r & h in $\ell_1$ & $\frac{1+r}{2ri_1+(1+r)(i_2+i_3+2)+2i_4}$ & $\ell_1$ & r & m & h in $\ell_2$ & $\frac{1}{N-1}$ & $(1,0,1,1,i_1,$\par$i_2,i_3,i_4)$ & $\frac{1+r}{2ri_1+(1+r)(i_2+i_3+2)+2i_4}\cdot\frac{1}{2}\cdot\frac{1}{N-1}$  \\[2.5ex] 
  \hline
  \rule{0pt}{20pt}m & r & h in $\ell_1$ & $\frac{1+r}{2ri_1+(1+r)(i_2+i_3+2)+2i_4}$ & $\ell_1$ & r & m & l in $\ell_1$, $\ell_2$ & $\frac{i_3}{N-1}$ & $(1,0,0,1,i_1+$\par$1,i_2,i_3-1,i_4)$ & $\frac{1+r}{2ri_1+(1+r)(i_2+i_3+2)+2i_4}\cdot\frac{1}{2}\cdot\frac{i_3}{N-1}$  \\[2.5ex] 
  \hline
  \rule{0pt}{20pt}m & r & h in $\ell_1$ & $\frac{1+r}{2ri_1+(1+r)(i_2+i_3+2)+2i_4}$ & $\ell_1$ & r & r & l in $\ell_1$, $\ell_2$ & $\frac{i_4}{N-1}$ & $(1,0,0,1,i_1,$\par$i_2+1,i_3,i_4-1)$ & $\frac{1+r}{2ri_1+(1+r)(i_2+i_3+2)+2i_4}\cdot\frac{1}{2}\cdot\frac{i_4}{N-1}$  \\[2.5ex] 
  \hline
  \rule{0pt}{20pt}r & m/r & l in $\ell_1$ & $\frac{(1+r)(i_3+1)+2i_4}{2ri_1+(1+r)(i_2+i_3+2)+2i_4}$ & $\ell_1$ & m & r & h in $\ell_1$ & $1$ & $(0,0,0,1,i_1,$\par$i_2,i_3,i_4)$ & $\frac{(1+r)(i_3+1)+2i_4}{2ri_1+(1+r)(i_2+i_3+2)+2i_4}\cdot\frac{1}{2}$ \\[2.5ex] 
  \hline
  \rule{0pt}{20pt}r & m & h in $\ell_2$ & $\frac{1+r}{2ri_1+(1+r)(i_2+i_3+2)+2i_4}$ & $\ell_2$ & m & r & h in $\ell_2$ & $\frac{1}{N-1}$ & $(1,1,0,1,i_1,$\par$i_2,i_3,i_4)$ & $\frac{1+r}{2ri_1+(1+r)(i_2+i_3+2)+2i_4}\cdot\frac{1}{2}\cdot\frac{1}{N-1}$ \\[2.5ex] 
  \hline
  \rule{0pt}{20pt}r & m & h in $\ell_2$ & $\frac{1+r}{2ri_1+(1+r)(i_2+i_3+2)+2i_4}$ & $\ell_2$ & m & r & l in $\ell_1$, $\ell_2$ & $\frac{i_2}{N-1}$ & $(1,0,0,1,i_1+$\par$1,i_2-1,i_3,i_4)$ & $\frac{1+r}{2ri_1+(1+r)(i_2+i_3+2)+2i_4}\cdot\frac{1}{2}\cdot\frac{i_2}{N-1}$ \\[2.5ex] 
  \hline
  \rule{0pt}{20pt}r & m & h in $\ell_2$ & $\frac{1+r}{2ri_1+(1+r)(i_2+i_3+2)+2i_4}$ & $\ell_2$ & r & r & l in $\ell_1$, $\ell_2$ & $\frac{i_4}{N-1}$ & $(1,0,0,1,i_1,$\par$i_2,i_3+1,i_4-1)$ & $\frac{1+r}{2ri_1+(1+r)(i_2+i_3+2)+2i_4}\cdot\frac{1}{2}\cdot\frac{i_4}{N-1}$ \\[2.5ex] 
  \hline
  \rule{0pt}{20pt}m/r & r & l in $\ell_2$ & $\frac{(1+r)(i_2+1)+2i_4}{2ri_1+(1+r)(i_2+i_3+2)+2i_4}$ & $\ell_2$ & r & m & h in $\ell_2$ & $1$ & $(1,0,0,0,i_1,$\par$i_2,i_3,i_4)$ & $\frac{(1+r)(i_2+1)+2i_4}{2ri_1+(1+r)(i_2+i_3+2)+2i_4}\cdot\frac{1}{2}$ \\[2.5ex] 
  \hline
\end{tabular}
\captionof{table}{\label{table-star-star-state7}Eight types of state transition from state $(1,0,0,1,i_1,i_2,i_3,i_4)$ under model 1 on the coupled star graph.}
\end{table}
\newpage
Assume that the current state is $(1,0,0,0,i_1,i_2,i_3,i_4)$. There are eight types of events that can occur next.
%
%
Table~\ref{table-star-star-state8} shows the seven types of events that accompany a state transition and their probabilities. If any other event than these seven types occurs, the state remains unchanged. The probability of this case is 1 minus the sum of all entries in the last column of Table~\ref{table-star-star-state8}. 

\begin{table}[H]
\hspace*{-3.5em}
\begin{tabular}{ | w{c}{1.1em} | w{c}{1.1em} | M{2.5em} | c | w{c}{1.1em} | w{c}{1.1em} | w{c}{1.1em} | M{2.5em} | w{c}{1.3em} | M{7em} | M{11em} |} 
  \hline
  \multicolumn{4}{|c|}{parent} & \multirow{3}{*}{layer} & \multicolumn{4}{c|}{child} & \multirow{3}{*}{\shortstack{state after\\ transition}} & \multirow{3}{*}{transition probability}\\
  \cline{1-4}
  \cline{6-9}
    \multicolumn{2}{|c|}{type} & \multirow{2}{*}{h/l} & \multirow{2}{*}{probability} & & \multicolumn{2}{c|}{type} & \multirow{2}{*}{h/l} & \multirow{2}{*}{prob.} & & \\
   \cline{1-2}
   \cline{6-7}      
    $\ell_1$ & $\ell_2$ &  &  &  & $\ell_1$ & $\ell_2$ &  &  &  & \\
  \hline
  \rule{0pt}{20pt}m & r & h in $\ell_1$ & $\frac{1+r}{2ri_1+(1+r)(i_2+i_3+1)+2(i_4+1)}$ & $\ell_1$ & r & r & h in $\ell_2$ & $\frac{1}{N-1}$ & $(1,0,1,0,i_1,$\par$i_2,i_3,i_4)$ & $\frac{1+r}{2ri_1+(1+r)(i_2+i_3+1)+2(i_4+1)}\cdot\frac{1}{2}\cdot\frac{1}{N-1}$  \\[2.5ex] 
  \hline
  \rule{0pt}{20pt}m & r & h in $\ell_1$ & $\frac{1+r}{2ri_1+(1+r)(i_2+i_3+1)+2(i_4+1)}$ & $\ell_1$ & r & m & l in $\ell_1$, $\ell_2$ & $\frac{i_3}{N-1}$ & $(1,0,0,0,i_1+$\par$1,i_2,i_3-1,i_4)$ & $\frac{1+r}{2ri_1+(1+r)(i_2+i_3+1)+2(i_4+1)}\cdot\frac{1}{2}\cdot\frac{i_3}{N-1}$  \\[2.5ex] 
  \hline
  \rule{0pt}{20pt}m & r & h in $\ell_1$ & $\frac{1+r}{2ri_1+(1+r)(i_2+i_3+1)+2(i_4+1)}$ & $\ell_1$ & r & r & l in $\ell_1$, $\ell_2$ & $\frac{i_4}{N-1}$ & $(1,0,0,0,i_1,$\par$i_2+1,i_3,i_4-1)$ & $\frac{1+r}{2ri_1+(1+r)(i_2+i_3+1)+2(i_4+1)}\cdot\frac{1}{2}\cdot\frac{i_4}{N-1}$  \\[2.5ex] 
  \hline
  \rule{0pt}{20pt}r & m/r & l in $\ell_1$ & $\frac{(1+r)i_3+2(i_4+1)}{2ri_1+(1+r)(i_2+i_3+1)+2(i_4+1)}$ & $\ell_1$ & m & r & h in $\ell_1$ & $1$ & $(0,0,0,0,i_1,$\par$i_2,i_3,i_4)$ & $\frac{(1+r)i_3+2(i_4+1)}{2ri_1+(1+r)(i_2+i_3+1)+2(i_4+1)}\cdot\frac{1}{2}$ \\[2.5ex] 
  \hline
  \rule{0pt}{20pt}m/r & m & l in $\ell_1$, $\ell_2$ & $\frac{2ri_1+(1+r)i_3}{2ri_1+(1+r)(i_2+i_3+1)+2(i_4+1)}$ & $\ell_2$ & r & r & h in $\ell_2$ & $1$ & $(1,0,0,1,i_1,$\par$i_2,i_3,i_4)$ & $\frac{2ri_1+(1+r)i_3}{2ri_1+(1+r)(i_2+i_3+1)+2(i_4+1)}\cdot\frac{1}{2}$ \\[2.5ex] 
  \hline
  \rule{0pt}{20pt}r & r & h in $\ell_2$ & $\frac{2}{2ri_1+(1+r)(i_2+i_3+1)+2(i_4+1)}$ & $\ell_2$ & m & m & l in $\ell_1$, $\ell_2$ & $\frac{i_1}{N-1}$ & $(1,0,0,0,i_1-$\par$1,i_2+1,i_3,i_4)$ & $\frac{2}{2ri_1+(1+r)(i_2+i_3+1)+2(i_4+1)}\cdot\frac{1}{2}\cdot\frac{i_1}{N-1}$ \\[2.5ex] 
  \hline
  \rule{0pt}{20pt}r & r & h in $\ell_2$ & $\frac{2}{2ri_1+(1+r)(i_2+i_3+1)+2(i_4+1)}$ & $\ell_2$ & r & m & l in $\ell_1$, $\ell_2$ & $\frac{i_3}{N-1}$ & $(1,0,0,0,i_1,$\par$i_2,i_3-1,i_4+1)$ & $\frac{2}{2ri_1+(1+r)(i_2+i_3+1)+2(i_4+1)}\cdot\frac{1}{2}\cdot\frac{i_3}{N-1}$ \\[2.5ex] 
  \hline
\end{tabular}
\captionof{table}{\label{table-star-star-state8}Seven types of state transition from state $(1,0,0,0,i_1,i_2,i_3,i_4)$ under model 1 on the coupled star graph.}
\end{table}
\newpage
Assume that the current state is $(0,1,1,1,i_1,i_2,i_3,i_4)$. There are eight types of events that can occur next.
%
%
Table~\ref{table-star-star-state9} shows the seven types of events that accompany a state transition and their probabilities. If any other event than these seven types occurs, the state remains unchanged. The probability of this case is 1 minus the sum of all entries in the last column of Table~\ref{table-star-star-state9}. 

\begin{table}[H]
\hspace*{-3.5em}
\begin{tabular}{ | w{c}{1.1em} | w{c}{1.1em} | M{2.5em} | c | w{c}{1.1em} | w{c}{1.1em} | w{c}{1.1em} | M{2.5em} | w{c}{1.3em} | M{7em} | M{11em} |}  
  \hline
  \multicolumn{4}{|c|}{parent} & \multirow{3}{*}{layer} & \multicolumn{4}{c|}{child} & \multirow{3}{*}{\shortstack{state after\\ transition}} & \multirow{3}{*}{transition probability}\\
  \cline{1-4}
  \cline{6-9}
    \multicolumn{2}{|c|}{type} & \multirow{2}{*}{h/l} & \multirow{2}{*}{probability} & & \multicolumn{2}{c|}{type} & \multirow{2}{*}{h/l} & \multirow{2}{*}{prob.} & & \\
   \cline{1-2}
   \cline{6-7}      
    $\ell_1$ & $\ell_2$ &  &  &  & $\ell_1$ & $\ell_2$ &  &  &  & \\
  \hline
  \rule{0pt}{20pt}m & m/r & l in $\ell_1$ & $\frac{2r(i_1+1)+(1+r)i_2}{2r(i_1+1)+(1+r)(i_2+i_3+1)+2i_4}$ & $\ell_1$ & r & m & h in $\ell_1$ & $1$ & $(1,1,1,1,i_1,$\par$i_2,i_3,i_4)$ & $\frac{2r(i_1+1)+(1+r)i_2}{2r(i_1+1)+(1+r)(i_2+i_3+1)+2i_4}\cdot\frac{1}{2}$  \\[2.5ex] 
  \hline
  \rule{0pt}{20pt}r & m & h in $\ell_1$ & $\frac{1+r}{2r(i_1+1)+(1+r)(i_2+i_3+1)+2i_4}$ & $\ell_1$ & m & m & h in $\ell_2$ & $\frac{1}{N-1}$ & $(0,1,0,1,i_1,$\par$i_2,i_3,i_4)$ & $\frac{1+r}{2r(i_1+1)+(1+r)(i_2+i_3+1)+2i_4}\cdot\frac{1}{2}\cdot\frac{1}{N-1}$  \\[2.5ex] 
  \hline
  \rule{0pt}{20pt}r & m & h in $\ell_1$ & $\frac{1+r}{2r(i_1+1)+(1+r)(i_2+i_3+1)+2i_4}$ & $\ell_1$ & m & m & l in $\ell_1$, $\ell_2$ & $\frac{i_1}{N-1}$ & $(0,1,1,1,i_1-$\par$1,i_2,i_3+1,i_4)$ & $\frac{1+r}{2r(i_1+1)+(1+r)(i_2+i_3+1)+2i_4}\cdot\frac{1}{2}\cdot\frac{i_1}{N-1}$ \\[2.5ex] 
  \hline
  \rule{0pt}{20pt}r & m & h in $\ell_1$ & $\frac{1+r}{2r(i_1+1)+(1+r)(i_2+i_3+1)+2i_4}$ & $\ell_1$ & m & r & l in $\ell_1$, $\ell_2$ & $\frac{i_2}{N-1}$ & $(0,1,1,1,i_1,$\par$i_2-1,i_3,i_4+1)$ & $\frac{1+r}{2r(i_1+1)+(1+r)(i_2+i_3+1)+2i_4}\cdot\frac{1}{2}\cdot\frac{i_2}{N-1}$ \\[2.5ex] 
  \hline
  \rule{0pt}{20pt}m & m & h in $\ell_2$ & $\frac{2r}{2r(i_1+1)+(1+r)(i_2+i_3+1)+2i_4}$ & $\ell_2$ & m & r & l in $\ell_1$, $\ell_2$ & $\frac{i_2}{N-1}$ & $(0,1,1,1,i_1+$\par$1,i_2-1,i_3,i_4)$ & $\frac{2r}{2r(i_1+1)+(1+r)(i_2+i_3+1)+2i_4}\cdot\frac{1}{2}\cdot\frac{i_2}{N-1}$ \\[2.5ex] 
  \hline
  \rule{0pt}{20pt}m & m & h in $\ell_2$ & $\frac{2r}{2r(i_1+1)+(1+r)(i_2+i_3+1)+2i_4}$ & $\ell_2$ & r & r & l in $\ell_1$, $\ell_2$ & $\frac{i_4}{N-1}$ & $(0,1,1,1,i_1,$\par$i_2,i_3+1,i_4-1)$ & $\frac{2r}{2r(i_1+1)+(1+r)(i_2+i_3+1)+2i_4}\cdot\frac{1}{2}\cdot\frac{i_4}{N-1}$ \\[2.5ex] 
  \hline
  \rule{0pt}{20pt}m/r & r & l in $\ell_2$ & $\frac{(1+r)i_2+2i_4}{2r(i_1+1)+(1+r)(i_2+i_3+1)+2i_4}$ & $\ell_2$ & m & m & h in $\ell_2$ & $1$ & $(0,1,1,0,i_1,$\par$i_2,i_3,i_4)$ & $\frac{(1+r)i_2+2i_4}{2r(i_1+1)+(1+r)(i_2+i_3+1)+2i_4}\cdot\frac{1}{2}$ \\[2.5ex] 
  \hline
\end{tabular}
\captionof{table}{\label{table-star-star-state9}Seven types of state transition from state $(0,1,1,1,i_1,i_2,i_3,i_4)$ under model 1 on the coupled star graph.}
\end{table}
\newpage
Assume that the current state is $(0,1,1,0,i_1,i_2,i_3,i_4)$. There are nine types of events that can occur next.
%
%
Table~\ref{table-star-star-state10} shows the eight types of events that accompany a state transition and their probabilities. If any other event than these eight types occurs, the state remains unchanged. The probability of this case is 1 minus the sum of all entries in the last column of Table~\ref{table-star-star-state10}. 

\begin{table}[H]
\hspace*{-2.5em}
\begin{tabular}{ | w{c}{1.1em} | w{c}{1.1em} | M{2.5em} | c | w{c}{1.1em} | w{c}{1.1em} | w{c}{1.1em} | M{2.5em} | w{c}{1.3em} | M{7em} | M{10em} |} 
  \hline
  \multicolumn{4}{|c|}{parent} & \multirow{3}{*}{layer} & \multicolumn{4}{c|}{child} & \multirow{3}{*}{\shortstack{state after\\ transition}} & \multirow{3}{*}{transition probability}\\
  \cline{1-4}
  \cline{6-9}
    \multicolumn{2}{|c|}{type} & \multirow{2}{*}{h/l} & \multirow{2}{*}{probability} & & \multicolumn{2}{c|}{type} & \multirow{2}{*}{h/l} & \multirow{2}{*}{prob.} & & \\
   \cline{1-2}
   \cline{6-7}      
    $\ell_1$ & $\ell_2$ &  &  &  & $\ell_1$ & $\ell_2$ &  &  &  & \\
  \hline
  \rule{0pt}{20pt}m & m/r & l in $\ell_1$ & $\frac{2ri_1+(1+r)(i_2+1)}{2ri_1+(1+r)(i_2+i_3+2)+2i_4}$ & $\ell_1$ & r & m & h in $\ell_2$ & $1$ & $(1,1,1,0,i_1,$\par$i_2,i_3,i_4)$ & $\frac{2ri_1+(1+r)(i_2+1)}{2ri_1+(1+r)(i_2+i_3+2)+2i_4}\cdot\frac{1}{2}$  \\[2.5ex] 
  \hline
  \rule{0pt}{20pt}r & m & h in $\ell_1$ & $\frac{1+r}{2ri_1+(1+r)(i_2+i_3+2)+2i_4}$ & $\ell_1$ & m & r & h in $\ell_2$ & $\frac{1}{N-1}$ & $(0,1,0,0,i_1,$\par$i_2,i_3,i_4)$ & $\frac{1+r}{2ri_1+(1+r)(i_2+i_3+2)+2i_4}\cdot\frac{1}{2}\cdot\frac{1}{N-1}$  \\[2.5ex] 
  \hline
  \rule{0pt}{20pt}r & m & h in $\ell_1$ & $\frac{1+r}{2ri_1+(1+r)(i_2+i_3+2)+2i_4}$ & $\ell_1$ & m & m & l in $\ell_1$, $\ell_2$ & $\frac{i_1}{N-1}$ & $(0,1,1,0,i_1-$\par$1,i_2,i_3+1,i_4)$ & $\frac{1+r}{2ri_1+(1+r)(i_2+i_3+2)+2i_4}\cdot\frac{1}{2}\cdot\frac{i_1}{N-1}$  \\[2.5ex] 
  \hline
  \rule{0pt}{20pt}r & m & h in $\ell_1$ & $\frac{1+r}{2ri_1+(1+r)(i_2+i_3+2)+2i_4}$ & $\ell_1$ & m & r & l in $\ell_1$, $\ell_2$ & $\frac{i_2}{N-1}$ & $(0,1,1,0,i_1,$\par$i_2-1,i_3,i_4+1)$ & $\frac{1+r}{2ri_1+(1+r)(i_2+i_3+2)+2i_4}\cdot\frac{1}{2}\cdot\frac{i_2}{N-1}$ \\[2.5ex] 
  \hline
  \rule{0pt}{20pt}m/r & m & l in $\ell_2$ & $\frac{2ri_1+(1+r)(i_3+1)}{2ri_1+(1+r)(i_2+i_3+2)+2i_4}$ & $\ell_2$ & m & r & h in $\ell_2$ & $1$ & $(0,1,1,1,i_1,$\par$i_2,i_3,i_4)$ & $\frac{2ri_1+(1+r)(i_3+1)}{2ri_1+(1+r)(i_2+i_3+2)+2i_4}\cdot\frac{1}{2}$ \\[2.5ex] 
  \hline
  \rule{0pt}{20pt}m & r & h in $\ell_2$ & $\frac{1+r}{2ri_1+(1+r)(i_2+i_3+2)+2i_4}$ & $\ell_2$ & r & m & h in $\ell_1$ & $\frac{1}{N-1}$ & $(0,0,1,0,i_1,$\par$i_2,i_3,i_4)$ & $\frac{1+r}{2ri_1+(1+r)(i_2+i_3+2)+2i_4}\cdot\frac{1}{2}\cdot\frac{1}{N-1}$ \\[2.5ex] 
  \hline
  \rule{0pt}{20pt}m & r & h in $\ell_2$ & $\frac{1+r}{2ri_1+(1+r)(i_2+i_3+2)+2i_4}$ & $\ell_2$ & m & m & l in $\ell_1$, $\ell_2$ & $\frac{i_1}{N-1}$ & $(0,1,1,0,i_1-$\par$1,i_2+1,i_3,i_4)$ & $\frac{1+r}{2ri_1+(1+r)(i_2+i_3+2)+2i_4}\cdot\frac{1}{2}\cdot\frac{i_1}{N-1}$ \\[2.5ex] 
  \hline
  \rule{0pt}{20pt}m & r & h in $\ell_2$ & $\frac{1+r}{2ri_1+(1+r)(i_2+i_3+2)+2i_4}$ & $\ell_2$ & r & m & l in $\ell_1$, $\ell_2$ & $\frac{i_3}{N-1}$ & $(0,1,1,0,i_1,$\par$i_2,i_3-1,i_4+1)$ & $\frac{1+r}{2ri_1+(1+r)(i_2+i_3+2)+2i_4}\cdot\frac{1}{2}\cdot\frac{i_3}{N-1}$ \\[2.5ex] 
  \hline
\end{tabular}
\captionof{table}{\label{table-star-star-state10}Eight types of state transition from state $(0,1,1,0,i_1,i_2,i_3,i_4)$ under model 1 on the coupled star graph.}
\end{table}
 \newpage
Assume that the current state is $(0,1,0,1,i_1,i_2,i_3,i_4)$. There are seven types of events that can occur next.
%
%
Table~\ref{table-star-star-state11} shows the six types of events that accompany a state transition and their probabilities. If any other event than these six types occurs, the state remains unchanged. The probability of this case is 1 minus the sum of all entries in the last column of Table~\ref{table-star-star-state11}. 

\begin{table}[H]
\hspace*{-2.5em}
\begin{tabular}{ | w{c}{1.1em} | w{c}{1.1em} | M{2.5em} | c | w{c}{1.1em} | w{c}{1.1em} | w{c}{1.1em} | M{2.5em} | w{c}{1.3em} | M{7em} | M{10em} |}  
  \hline
  \multicolumn{4}{|c|}{parent} & \multirow{3}{*}{layer} & \multicolumn{4}{c|}{child} & \multirow{3}{*}{\shortstack{state after\\ transition}} & \multirow{3}{*}{transition probability}\\
  \cline{1-4}
  \cline{6-9}
    \multicolumn{2}{|c|}{type} & \multirow{2}{*}{h/l} & \multirow{2}{*}{probability} & & \multicolumn{2}{c|}{type} & \multirow{2}{*}{h/l} & \multirow{2}{*}{prob.} & & \\
   \cline{1-2}
   \cline{6-7}      
    $\ell_1$ & $\ell_2$ &  &  &  & $\ell_1$ & $\ell_2$ &  &  &  & \\
  \hline
  \rule{0pt}{20pt}m & m/r & l in $\ell_1$, $\ell_2$ & $\frac{2ri_1+(1+r)i_2}{2ri_1+(1+r)(i_2+i_3+2)+2i_4}$ & $\ell_1$ & r & m & h in $\ell_1$ & $1$ & $(1,1,0,1,i_1,$\par$i_2,i_3,i_4)$ & $\frac{2ri_1+(1+r)i_2}{2ri_1+(1+r)(i_2+i_3+2)+2i_4}\cdot\frac{1}{2}$  \\[2.5ex] 
  \hline
  \rule{0pt}{20pt}r & m & h in $\ell_1$ & $\frac{1+r}{2ri_1+(1+r)(i_2+i_3+2)+2i_4}$ & $\ell_1$ & m & m & l in $\ell_1$, $\ell_2$ & $\frac{i_1}{N-1}$ & $(0,1,0,1,i_1-$\par$1,i_2,i_3+1,i_4)$ & $\frac{1+r}{2ri_1+(1+r)(i_2+i_3+2)+2i_4}\cdot\frac{1}{2}\cdot\frac{i_1}{N-1}$  \\[2.5ex] 
  \hline
  \rule{0pt}{20pt}r & m & h in $\ell_1$ & $\frac{1+r}{2ri_1+(1+r)(i_2+i_3+2)+2i_4}$ & $\ell_1$ & m & r & l in $\ell_1$, $\ell_2$ & $\frac{i_2}{N-1}$ & $(0,1,0,1,i_1,$\par$i_2-1,i_3,i_4+1)$ & $\frac{1+r}{2ri_1+(1+r)(i_2+i_3+2)+2i_4}\cdot\frac{1}{2}\cdot\frac{i_2}{N-1}$  \\[2.5ex] 
  \hline
  \rule{0pt}{20pt}r & m & h in $\ell_2$ & $\frac{1+r}{2ri_1+(1+r)(i_2+i_3+2)+2i_4}$ & $\ell_2$ & m & r & l in $\ell_1$, $\ell_2$ & $\frac{i_2}{N-1}$ & $(0,1,0,1,i_1+$\par$1,i_2-1,i_3,i_4)$ & $\frac{1+r}{2ri_1+(1+r)(i_2+i_3+2)+2i_4}\cdot\frac{1}{2}\cdot\frac{i_2}{N-1}$ \\[2.5ex] 
  \hline
  \rule{0pt}{20pt}r & m & h in $\ell_2$ & $\frac{1+r}{2ri_1+(1+r)(i_2+i_3+2)+2i_4}$ & $\ell_2$ & r & r & l in $\ell_1$, $\ell_2$ & $\frac{i_4}{N-1}$ & $(0,1,0,1,i_1,$\par$i_2,i_3+1,i_4-1)$ & $\frac{1+r}{2ri_1+(1+r)(i_2+i_3+2)+2i_4}\cdot\frac{1}{2}\cdot\frac{i_4}{N-1}$ \\[2.5ex] 
  \hline
  \rule{0pt}{20pt}m/r & r & l in $\ell_1$, $\ell_2$ & $\frac{(1+r)i_2+2i_4}{2ri_1+(1+r)(i_2+i_3+2)+2i_4}$ & $\ell_2$ & r & m & h in $\ell_2$ & $1$ & $(0,1,0,0,i_1,$\par$i_2,i_3,i_4)$ & $\frac{(1+r)i_2+2i_4}{2ri_1+(1+r)(i_2+i_3+2)+2i_4}\cdot\frac{1}{2}$ \\[2.5ex] 
  \hline
\end{tabular}
\captionof{table}{\label{table-star-star-state11}Six types of state transition from state $(0,1,0,1,i_1,i_2,i_3,i_4)$ under model 1 on the coupled star graph.}
\end{table}
\newpage
Assume that the current state is $(0,1,0,0,i_1,i_2,i_3,i_4)$. There are eight types of events that can occur next.
%
%
Table~\ref{table-star-star-state12} shows the seven types of events that accompany a state transition and their probabilities. If any other event than these seven types occurs, the state remains unchanged. The probability of this case is 1 minus the sum of all entries in the last column of Table~\ref{table-star-star-state12}. 

\begin{table}[H]
\hspace*{-3.5em}
\begin{tabular}{ | w{c}{1.1em} | w{c}{1.1em} | M{2.5em} | c | w{c}{1.1em} | w{c}{1.1em} | w{c}{1.1em} | M{2.5em} | w{c}{1.3em} | M{7em} | M{11em} |}  
  \hline
  \multicolumn{4}{|c|}{parent} & \multirow{3}{*}{layer} & \multicolumn{4}{c|}{child} & \multirow{3}{*}{\shortstack{state after\\ transition}} & \multirow{3}{*}{transition probability}\\
  \cline{1-4}
  \cline{6-9}
    \multicolumn{2}{|c|}{type} & \multirow{2}{*}{h/l} & \multirow{2}{*}{probability} & & \multicolumn{2}{c|}{type} & \multirow{2}{*}{h/l} & \multirow{2}{*}{prob.} & & \\
   \cline{1-2}
   \cline{6-7}      
    $\ell_1$ & $\ell_2$ &  &  &  & $\ell_1$ & $\ell_2$ &  &  &  & \\
  \hline
  \rule{0pt}{20pt}m & m/r & l in $\ell_1$, $\ell_2$ & $\frac{2ri_1+(1+r)i_2}{2ri_1+(1+r)(i_2+i_3+1)+2(i_4+1)}$ & $\ell_1$ & r & m & h in $\ell_1$ & $1$ & $(1,1,0,0,i_1,$\par$i_2,i_3,i_4)$ & $\frac{2ri_1+(1+r)i_2}{2ri_1+(1+r)(i_2+i_3+1)+2(i_4+1)}\cdot\frac{1}{2}$  \\[2.5ex] 
  \hline
  \rule{0pt}{20pt}r & m & h in $\ell_1$ & $\frac{1+r}{2ri_1+(1+r)(i_2+i_3+1)+2(i_4+1)}$ & $\ell_1$ & m & m & l in $\ell_1$, $\ell_2$ & $\frac{i_1}{N-1}$ & $(0,1,0,0,i_1-$\par$1,i_2,i_3+1,i_4)$ & $\frac{1+r}{2ri_1+(1+r)(i_2+i_3+1)+2(i_4+1)}\cdot\frac{1}{2}\cdot\frac{i_1}{N-1}$  \\[2.5ex] 
  \hline
  \rule{0pt}{20pt}r & m & h in $\ell_1$ & $\frac{1+r}{2ri_1+(1+r)(i_2+i_3+1)+2(i_4+1)}$ & $\ell_1$ & m & r & l in $\ell_1$, $\ell_2$ & $\frac{i_2}{N-1}$ & $(0,1,0,0,i_1,$\par$i_2-1,i_3,i_4+1)$ & $\frac{1+r}{2ri_1+(1+r)(i_2+i_3+1)+2(i_4+1)}\cdot\frac{1}{2}\cdot\frac{i_2}{N-1}$  \\[2.5ex] 
  \hline
  \rule{0pt}{20pt}m/r & m & l in $\ell_2$ & $\frac{2ri_1+(1+r)(i_3+1)}{2ri_1+(1+r)(i_2+i_3+1)+2(i_4+1)}$ & $\ell_2$ & r & r & h in $\ell_2$ & $1$ & $(0,1,0,1,i_1,$\par$i_2,i_3,i_4)$ & $\frac{2ri_1+(1+r)(i_3+1)}{2ri_1+(1+r)(i_2+i_3+1)+2(i_4+1)}\cdot\frac{1}{2}$ \\[2.5ex] 
  \hline
  \rule{0pt}{20pt}r & r & h in $\ell_2$ & $\frac{2}{2ri_1+(1+r)(i_2+i_3+1)+2(i_4+1)}$ & $\ell_2$ & r & m & h in $\ell_1$ & $\frac{1}{N-1}$ & $(0,0,0,0,i_1,$\par$i_2,i_3,i_4)$ & $\frac{2}{2ri_1+(1+r)(i_2+i_3+1)+2(i_4+1)}\cdot\frac{1}{2}\cdot\frac{1}{N-1}$ \\[2.5ex] 
  \hline
  \rule{0pt}{20pt}r & r & h in $\ell_2$ & $\frac{2}{2ri_1+(1+r)(i_2+i_3+1)+2(i_4+1)}$ & $\ell_2$ & m & m & l in $\ell_1$, $\ell_2$ & $\frac{i_1}{N-1}$ & $(0,1,0,0,i_1-$\par$1,i_2+1,i_3,i_4)$ & $\frac{2}{2ri_1+(1+r)(i_2+i_3+1)+2(i_4+1)}\cdot\frac{1}{2}\cdot\frac{i_1}{N-1}$ \\[2.5ex] 
  \hline
  \rule{0pt}{20pt}r & r & h in $\ell_2$ & $\frac{2}{2ri_1+(1+r)(i_2+i_3+1)+2(i_4+1)}$ & $\ell_2$ & r & m & l in $\ell_1$, $\ell_2$ & $\frac{i_3}{N-1}$ & $(0,1,0,0,i_1,$\par$i_2,i_3-1,i_4+1)$ & $\frac{2}{2ri_1+(1+r)(i_2+i_3+1)+2(i_4+1)}\cdot\frac{1}{2}\cdot\frac{i_3}{N-1}$ \\[2.5ex] 
  \hline
\end{tabular}
\captionof{table}{\label{table-star-star-state12}Seven types of state transition from state $(0,1,0,0,i_1,i_2,i_3,i_4)$ under model 1 on the coupled star graph.}
\end{table}
\newpage
Assume that the current state is $(0,0,1,1,i_1,i_2,i_3,i_4)$. There are nine types of events that can occur next.
%
%
Table~\ref{table-star-star-state13} shows the eight types of events that accompany a state transition and their probabilities. If any other event than these eight types occurs, the state remains unchanged. The probability of this case is 1 minus the sum of all entries in the last column of Table~\ref{table-star-star-state13}. 

\begin{table}[H]
\hspace*{-4.5em}
\begin{tabular}{ | w{c}{1.1em} | w{c}{1.1em} | M{2.5em} | c | w{c}{1.1em} | w{c}{1.1em} | w{c}{1.1em} | M{2.5em} | w{c}{1.3em} | M{7em} | M{11.6em} |}  
  \hline
  \multicolumn{4}{|c|}{parent} & \multirow{3}{*}{layer} & \multicolumn{4}{c|}{child} & \multirow{3}{*}{\shortstack{state after\\ transition}} & \multirow{3}{*}{transition probability}\\
  \cline{1-4}
  \cline{6-9}
    \multicolumn{2}{|c|}{type} & \multirow{2}{*}{h/l} & \multirow{2}{*}{probability} & & \multicolumn{2}{c|}{type} & \multirow{2}{*}{h/l} & \multirow{2}{*}{prob.} & & \\
   \cline{1-2}
   \cline{6-7}      
    $\ell_1$ & $\ell_2$ &  &  &  & $\ell_1$ & $\ell_2$ &  &  &  & \\
  \hline
  \rule{0pt}{20pt}m & m/r & l in $\ell_1$ & $\frac{2r(i_1+1)+(1+r)i_2}{2r(i_1+1)+(1+r)(i_2+i_3)+2(i_4+1)}$ & $\ell_1$ & r & r & h in $\ell_1$ & $1$ & $(1,0,1,1,i_1,$\par$i_2,i_3,i_4)$ & $\frac{2r(i_1+1)+(1+r)i_2}{2r(i_1+1)+(1+r)(i_2+i_3)+2(i_4+1)}\cdot\frac{1}{2}$  \\[2.5ex] 
  \hline
  \rule{0pt}{20pt}r & r & h in $\ell_1$ & $\frac{2}{2r(i_1+1)+(1+r)(i_2+i_3)+2(i_4+1)}$ & $\ell_1$ & m & m & h in $\ell_2$ & $\frac{1}{N-1}$ & $(0,0,0,1,i_1,$\par$i_2,i_3,i_4)$ & $\frac{2}{2r(i_1+1)+(1+r)(i_2+i_3)+2(i_4+1)}\cdot\frac{1}{2}\cdot\frac{1}{N-1}$  \\[2.5ex] 
  \hline
  \rule{0pt}{20pt}r & r & h in $\ell_1$ & $\frac{2}{2r(i_1+1)+(1+r)(i_2+i_3)+2(i_4+1)}$ & $\ell_1$ & m & m & l in $\ell_1$, $\ell_2$ & $\frac{i_1}{N-1}$ & $(0,0,1,1,i_1-$\par$1,i_2,i_3+1,i_4)$ & $\frac{2}{2r(i_1+1)+(1+r)(i_2+i_3)+2(i_4+1)}\cdot\frac{1}{2}\cdot\frac{i_1}{N-1}$ \\[2.5ex] 
  \hline
  \rule{0pt}{20pt}r & r & h in $\ell_1$ & $\frac{2}{2r(i_1+1)+(1+r)(i_2+i_3)+2(i_4+1)}$ & $\ell_1$ & m & r & l in $\ell_1$, $\ell_2$ & $\frac{i_2}{N-1}$ & $(0,0,1,1,i_1,$\par$i_2-1,i_3,i_4+1)$ & $\frac{2}{2r(i_1+1)+(1+r)(i_2+i_3)+2(i_4+1)}\cdot\frac{1}{2}\cdot\frac{i_2}{N-1}$ \\[2.5ex] 
  \hline
  \rule{0pt}{20pt}m & m & h in $\ell_2$ & $\frac{2r}{2r(i_1+1)+(1+r)(i_2+i_3)+2(i_4+1)}$ & $\ell_2$ & r & r & h in $\ell_1$ & $\frac{1}{N-1}$ & $(0,1,1,1,i_1,$\par$i_2,i_3,i_4)$ & $\frac{2r}{2r(i_1+1)+(1+r)(i_2+i_3)+2(i_4+1)}\cdot\frac{1}{2}\cdot\frac{1}{N-1}$ \\[2.5ex] 
  \hline
  \rule{0pt}{20pt}m & m & h in $\ell_2$ & $\frac{2r}{2r(i_1+1)+(1+r)(i_2+i_3)+2(i_4+1)}$ & $\ell_2$ & m & r & l in $\ell_1$, $\ell_2$ & $\frac{i_2}{N-1}$ & $(0,0,1,1,i_1+$\par$1,i_2-1,i_3,i_4)$ & $\frac{2r}{2r(i_1+1)+(1+r)(i_2+i_3)+2(i_4+1)}\cdot\frac{1}{2}\cdot\frac{i_2}{N-1}$ \\[2.5ex] 
  \hline
  \rule{0pt}{20pt}m & m & h in $\ell_2$ & $\frac{2r}{2r(i_1+1)+(1+r)(i_2+i_3)+2(i_4+1)}$ & $\ell_2$ & r & m & l in $\ell_1$, $\ell_2$ & $\frac{i_3}{N-1}$ & $(0,0,1,1,i_1,$\par$i_2,i_3-1,i_4+1)$ & $\frac{2r}{2r(i_1+1)+(1+r)(i_2+i_3)+2(i_4+1)}\cdot\frac{1}{2}\cdot\frac{i_3}{N-1}$ \\[2.5ex] 
  \hline
  \rule{0pt}{20pt}m/r & r & l in $\ell_2$ & $\frac{(1+r)i_2+2(i_4+1)}{2r(i_1+1)+(1+r)(i_2+i_3)+2(i_4+1)}$ & $\ell_2$ & m & m & h in $\ell_2$ & $1$ & $(0,0,1,0,i_1,$\par$i_2,i_3,i_4)$ & $\frac{(1+r)i_2+2(i_4+1)}{2r(i_1+1)+(1+r)(i_2+i_3)+2(i_4+1)}\cdot\frac{1}{2}$ \\[2.5ex] 
  \hline
\end{tabular}
\captionof{table}{\label{table-star-star-state13}Eight types of state transition from state $(0,0,1,1,i_1,i_2,i_3,i_4)$ under model 1 on the coupled star graph.}
\end{table}
\newpage
Assume that the current state is $(0,0,1,0,i_1,i_2,i_3,i_4)$. There are eight types of events that can occur next.
%
%
Table~\ref{table-star-star-state14} shows the seven types of events that accompany a state transition and their probabilities. If any other event than these seven types occurs, the state remains unchanged. The probability of this case is 1 minus the sum of all entries in the last column of Table~\ref{table-star-star-state14}. 

\begin{table}[H]
\hspace*{-4em}
\begin{tabular}{ | w{c}{1.1em} | w{c}{1.1em} | M{2.5em} | c | w{c}{1.1em} | w{c}{1.1em} | w{c}{1.1em} | M{2.5em} | w{c}{1.3em} | M{7em} | M{11em} |}  
  \hline
  \multicolumn{4}{|c|}{parent} & \multirow{3}{*}{layer} & \multicolumn{4}{c|}{child} & \multirow{3}{*}{\shortstack{state after\\ transition}} & \multirow{3}{*}{transition probability}\\
  \cline{1-4}
  \cline{6-9}
    \multicolumn{2}{|c|}{type} & \multirow{2}{*}{h/l} & \multirow{2}{*}{probability} & & \multicolumn{2}{c|}{type} & \multirow{2}{*}{h/l} & \multirow{2}{*}{prob.} & & \\
   \cline{1-2}
   \cline{6-7}      
    $\ell_1$ & $\ell_2$ &  &  &  & $\ell_1$ & $\ell_2$ &  &  &  & \\
  \hline
  \rule{0pt}{20pt}m & m/r & l in $\ell_1$ & $\frac{2ri_1+(1+r)(i_2+1)}{2ri_1+(1+r)(i_2+i_3+1)+2(i_4+1)}$ & $\ell_1$ & r & r & h in $\ell_1$ & $1$ & $(1,0,1,0,i_1,$\par$i_2,i_3,i_4)$ & $\frac{2ri_1+(1+r)(i_2+1)}{2ri_1+(1+r)(i_2+i_3+1)+2(i_4+1)}\cdot\frac{1}{2}$  \\[2.5ex] 
  \hline
  \rule{0pt}{20pt}r & r & h in $\ell_1$ & $\frac{2}{2ri_1+(1+r)(i_2+i_3+1)+2(i_4+1)}$ & $\ell_1$ & m & r & h in $\ell_2$ & $\frac{1}{N-1}$ & $(0,0,0,0,i_1,$\par$i_2,i_3,i_4)$ & $\frac{2}{2ri_1+(1+r)(i_2+i_3+1)+2(i_4+1)}\cdot\frac{1}{2}\cdot\frac{1}{N-1}$  \\[2.5ex] 
  \hline
  \rule{0pt}{20pt}r & r & h in $\ell_1$ & $\frac{2}{2ri_1+(1+r)(i_2+i_3+1)+2(i_4+1)}$ & $\ell_1$ & m & m & l in $\ell_1$, $\ell_2$ & $\frac{i_1}{N-1}$ & $(0,0,1,0,i_1-$\par$1,i_2,i_3+1,i_4)$ & $\frac{2}{2ri_1+(1+r)(i_2+i_3+1)+2(i_4+1)}\cdot\frac{1}{2}\cdot\frac{i_1}{N-1}$  \\[2.5ex] 
  \hline
  \rule{0pt}{20pt}r & r & h in $\ell_1$ & $\frac{2}{2ri_1+(1+r)(i_2+i_3+1)+2(i_4+1)}$ & $\ell_1$ & m & r & l in $\ell_1$, $\ell_2$ & $\frac{i_2}{N-1}$ & $(0,0,1,0,i_1,$\par$i_2-1,i_3,i_4+1)$ & $\frac{2}{2ri_1+(1+r)(i_2+i_3+1)+2(i_4+1)}\cdot\frac{1}{2}\cdot\frac{i_2}{N-1}$ \\[2.5ex] 
  \hline
  \rule{0pt}{20pt}m/r & m & l in $\ell_1$, $\ell_2$ & $\frac{2ri_1+(1+r)i_3}{2ri_1+(1+r)(i_2+i_3+1)+2(i_4+1)}$ & $\ell_2$ & m & r & h in $\ell_2$ & $1$ & $(0,0,1,1,i_1,$\par$i_2,i_3,i_4)$ & $\frac{2ri_1+(1+r)i_3}{2ri_1+(1+r)(i_2+i_3+1)+2(i_4+1)}\cdot\frac{1}{2}$ \\[2.5ex] 
  \hline
  \rule{0pt}{20pt}m & r & h in $\ell_2$ & $\frac{1+r}{2ri_1+(1+r)(i_2+i_3+1)+2(i_4+1)}$ & $\ell_2$ & m & m & l in $\ell_1$, $\ell_2$ & $\frac{i_1}{N-1}$ & $(0,0,1,0,i_1-$\par$1,i_2+1,i_3,i_4)$ & $\frac{1+r}{2ri_1+(1+r)(i_2+i_3+1)+2(i_4+1)}\cdot\frac{1}{2}\cdot\frac{i_1}{N-1}$ \\[2.5ex] 
  \hline
  \rule{0pt}{20pt}m & r & h in $\ell_2$ & $\frac{1+r}{2ri_1+(1+r)(i_2+i_3+1)+2(i_4+1)}$ & $\ell_2$ & r & m & l in $\ell_1$, $\ell_2$ & $\frac{i_3}{N-1}$ & $(0,0,1,0,i_1,$\par$i_2,i_3-1,i_4+1)$ & $\frac{1+r}{2ri_1+(1+r)(i_2+i_3+1)+2(i_4+1)}\cdot\frac{1}{2}\cdot\frac{i_3}{N-1}$ \\[2.5ex] 
  \hline
\end{tabular}
\captionof{table}{\label{table-star-star-state14}Seven types of state transition from state $(0,0,1,0,i_1,i_2,i_3,i_4)$ under model 1 on the coupled star graph.}
\end{table}
\newpage
Assume that the current state is $(0,0,0,1,i_1,i_2,i_3,i_4)$. There are eight types of events that can occur next.
%
%
Table~\ref{table-star-star-state15} shows the seven types of events that accompany a state transition and their probabilities. If any other event than these seven types occurs, the state remains unchanged. The probability of this case is 1 minus the sum of all entries in the last column of Table~\ref{table-star-star-state15}. 

\begin{table}[H]
\hspace*{-4em}
\begin{tabular}{ | w{c}{1.1em} | w{c}{1.1em} | M{2.5em} | c | w{c}{1.1em} | w{c}{1.1em} | w{c}{1.1em} | M{2.5em} | w{c}{1.3em} | M{7em} | M{11em} |}  
  \hline
  \multicolumn{4}{|c|}{parent} & \multirow{3}{*}{layer} & \multicolumn{4}{c|}{child} & \multirow{3}{*}{\shortstack{state after\\ transition}} & \multirow{3}{*}{transition probability}\\
  \cline{1-4}
  \cline{6-9}
    \multicolumn{2}{|c|}{type} & \multirow{2}{*}{h/l} & \multirow{2}{*}{probability} & & \multicolumn{2}{c|}{type} & \multirow{2}{*}{h/l} & \multirow{2}{*}{prob.} & & \\
   \cline{1-2}
   \cline{6-7}      
    $\ell_1$ & $\ell_2$ &  &  &  & $\ell_1$ & $\ell_2$ &  &  &  & \\
  \hline
  \rule{0pt}{20pt}m & m/r & l in $\ell_1$, $\ell_2$ & $\frac{2ri_1+(1+r)i_2}{2ri_1+(1+r)(i_2+i_3+1)+2(i_4+1)}$ & $\ell_1$ & r & r & h in $\ell_1$ & $1$ & $(1,0,0,1,i_1,$\par$i_2,i_3,i_4)$ & $\frac{2ri_1+(1+r)i_2}{2ri_1+(1+r)(i_2+i_3+1)+2(i_4+1)}\cdot\frac{1}{2}$  \\[2.5ex] 
  \hline
  \rule{0pt}{20pt}r & r & h in $\ell_1$ & $\frac{2}{2ri_1+(1+r)(i_2+i_3+1)+2(i_4+1)}$ & $\ell_1$ & m & m & l in $\ell_1$, $\ell_2$ & $\frac{i_1}{N-1}$ & $(0,0,0,1,i_1-$\par$1,i_2,i_3+1,i_4)$ & $\frac{2}{2ri_1+(1+r)(i_2+i_3+1)+2(i_4+1)}\cdot\frac{1}{2}\cdot\frac{i_1}{N-1}$  \\[2.5ex] 
  \hline
  \rule{0pt}{20pt}r & r & h in $\ell_1$ & $\frac{2}{2ri_1+(1+r)(i_2+i_3+1)+2(i_4+1)}$ & $\ell_1$ & m & r & l in $\ell_1$, $\ell_2$ & $\frac{i_2}{N-1}$ & $(0,0,0,1,i_1,$\par$i_2-1,i_3,i_4+1)$ & $\frac{2}{2ri_1+(1+r)(i_2+i_3+1)+2(i_4+1)}\cdot\frac{1}{2}\cdot\frac{i_2}{N-1}$  \\[2.5ex] 
  \hline
  \rule{0pt}{20pt}r & m & h in $\ell_2$ & $\frac{1+r}{2ri_1+(1+r)(i_2+i_3+1)+2(i_4+1)}$ & $\ell_2$ & r & r & h in $\ell_1$ & $\frac{1}{N-1}$ & $(0,1,0,1,i_1,$\par$i_2,i_3,i_4)$ & $\frac{1+r}{2ri_1+(1+r)(i_2+i_3+1)+2(i_4+1)}\cdot\frac{1}{2}\cdot\frac{1}{N-1}$ \\[2.5ex] 
  \hline
  \rule{0pt}{20pt}r & m & h in $\ell_2$ & $\frac{1+r}{2ri_1+(1+r)(i_2+i_3+1)+2(i_4+1)}$ & $\ell_2$ & m & r & l in $\ell_1$, $\ell_2$ & $\frac{i_2}{N-1}$ & $(0,0,0,1,i_1+$\par$1,i_2-1,i_3,i_4)$ & $\frac{1+r}{2ri_1+(1+r)(i_2+i_3+1)+2(i_4+1)}\cdot\frac{1}{2}\cdot\frac{i_2}{N-1}$ \\[2.5ex] 
  \hline
  \rule{0pt}{20pt}r & m & h in $\ell_2$ & $\frac{1+r}{2ri_1+(1+r)(i_2+i_3+1)+2(i_4+1)}$ & $\ell_2$ & r & r & l in $\ell_1$, $\ell_2$ & $\frac{i_4}{N-1}$ & $(0,0,0,1,i_1,$\par$i_2,i_3+1,i_4-1)$ & $\frac{1+r}{2ri_1+(1+r)(i_2+i_3+1)+2(i_4+1)}\cdot\frac{1}{2}\cdot\frac{i_4}{N-1}$ \\[2.5ex] 
  \hline
  \rule{0pt}{20pt}m/r & r & l in $\ell_2$ & $\frac{(1+r)i_2+2(i_4+1)}{2ri_1+(1+r)(i_2+i_3+1)+2(i_4+1)}$ & $\ell_2$ & r & m & h in $\ell_2$ & $1$ & $(0,0,0,0,i_1,$\par$i_2,i_3,i_4)$ & $\frac{(1+r)i_2+2(i_4+1)}{2ri_1+(1+r)(i_2+i_3+1)+2(i_4+1)}\cdot\frac{1}{2}$ \\[2.5ex] 
  \hline
\end{tabular}
\captionof{table}{\label{table-star-star-state15}Seven types of state transition from state $(0,0,0,1,i_1,i_2,i_3,i_4)$ under model 1 on the coupled star graph.}
\end{table}
\newpage
Assume that the current state is $(0,0,0,0,i_1,i_2,i_3,i_4)$. There are seven types of events that can occur next.
%
%
Table~\ref{table-star-star-state16} shows the six types of events that accompany a state transition and their probabilities. If any other event than these six types occurs, the state remains unchanged. The probability of this case is 1 minus the sum of all entries in the last column of Table~\ref{table-star-star-state16}. 

\begin{table}[H]
\hspace*{-2.5em}
\begin{tabular}{ | w{c}{1.1em} | w{c}{1.1em} | M{2.5em} | c | w{c}{1.1em} | w{c}{1.1em} | w{c}{1.1em} | M{2.5em} | w{c}{1.3em} | M{7em} | M{10em} |} 
  \hline
  \multicolumn{4}{|c|}{parent} & \multirow{3}{*}{layer} & \multicolumn{4}{c|}{child} & \multirow{3}{*}{\shortstack{state after\\ transition}} & \multirow{3}{*}{transition probability}\\
  \cline{1-4}
  \cline{6-9}
    \multicolumn{2}{|c|}{type} & \multirow{2}{*}{h/l} & \multirow{2}{*}{probability} & & \multicolumn{2}{c|}{type} & \multirow{2}{*}{h/l} & \multirow{2}{*}{prob.} & & \\
   \cline{1-2}
   \cline{6-7}      
    $\ell_1$ & $\ell_2$ &  &  &  & $\ell_1$ & $\ell_2$ &  &  &  & \\
  \hline
  \rule{0pt}{20pt}m & m/r & l in $\ell_1$, $\ell_2$ & $\frac{2ri_1+(1+r)i_2}{2ri_1+(1+r)(i_2+i_3)+2(i_4+2)}$ & $\ell_1$ & r & r & h in $\ell_1$ & $1$ & $(1,0,0,0,i_1,$\par$i_2,i_3,i_4)$ & $\frac{2ri_1+(1+r)i_2}{2ri_1+(1+r)(i_2+i_3)+2(i_4+2)}\cdot\frac{1}{2}$  \\[2.5ex] 
  \hline
  \rule{0pt}{20pt}r & r & h in $\ell_1$ & $\frac{2}{2ri_1+(1+r)(i_2+i_3)+2(i_4+2)}$ & $\ell_1$ & m & m & l in $\ell_1$, $\ell_2$ & $\frac{i_1}{N-1}$ & $(0,0,0,0,i_1-$\par$1,i_2,i_3+1,i_4)$ & $\frac{2}{2ri_1+(1+r)(i_2+i_3)+2(i_4+2)}\cdot\frac{1}{2}\cdot\frac{i_1}{N-1}$  \\[2.5ex] 
  \hline
  \rule{0pt}{20pt}r & r & h in $\ell_1$ & $\frac{2}{2ri_1+(1+r)(i_2+i_3)+2(i_4+2)}$ & $\ell_1$ & m & r & l in $\ell_1$, $\ell_2$ & $\frac{i_2}{N-1}$ & $(0,0,0,0,i_1,$\par$i_2-1,i_3,i_4+1)$ & $\frac{2}{2ri_1+(1+r)(i_2+i_3)+2(i_4+2)}\cdot\frac{1}{2}\cdot\frac{i_2}{N-1}$  \\[2.5ex] 
  \hline
  \rule{0pt}{20pt}m/r & m & l in $\ell_1$, $\ell_2$ & $\frac{2ri_1+(1+r)i_3}{2ri_1+(1+r)(i_2+i_3)+2(i_4+2)}$ & $\ell_2$ & r & r & h in $\ell_2$ & $1$ & $(0,0,0,1,i_1,$\par$i_2,i_3,i_4)$ & $\frac{2ri_1+(1+r)i_3}{2ri_1+(1+r)(i_2+i_3)+2(i_4+2)}\cdot\frac{1}{2}$ \\[2.5ex] 
  \hline
  \rule{0pt}{20pt}r & r & h in $\ell_2$ & $\frac{2}{2ri_1+(1+r)(i_2+i_3)+2(i_4+2)}$ & $\ell_2$ & m & m & l in $\ell_1$, $\ell_2$ & $\frac{i_1}{N-1}$ & $(0,0,0,0,i_1-$\par$1,i_2+1,i_3,i_4)$ & $\frac{2}{2ri_1+(1+r)(i_2+i_3)+2(i_4+2)}\cdot\frac{1}{2}\cdot\frac{i_1}{N-1}$ \\[2.5ex] 
  \hline
  \rule{0pt}{20pt}r & r & h in $\ell_2$ & $\frac{2}{2ri_1+(1+r)(i_2+i_3)+2(i_4+2)}$ & $\ell_2$ & r & m & l in $\ell_1$, $\ell_2$ & $\frac{i_3}{N-1}$ & $(0,0,0,0,i_1,$\par$i_2,i_3-1,i_4+1)$ & $\frac{2}{2ri_1+(1+r)(i_2+i_3)+2(i_4+2)}\cdot\frac{1}{2}\cdot\frac{i_3}{N-1}$ \\[2.5ex] 
  \hline
\end{tabular}
\captionof{table}{\label{table-star-star-state16}Six types of state transition from state $(0,0,0,0,i_1,i_2,i_3,i_4)$ under model 1 on the coupled star graph.}
\end{table}
\newpage
\section{Derivation of the transition probability matrix for model 1 on two-layer networks composed of a complete graph layer and a complete bipartite network layer\label{complete-bi-m1-derivation}}

Denote the state of the evolutionary dynamics by $(i_1,i_2,i_3,i_4,j_1,j_2,j_3,j_4)$, i.e., an 8-tuple, where $i_1$ is the number of individuals in $V_1$ that have the mutant type in both layers; $i_2$ is the number of individuals in $V_1$ that have the mutant type in layer 1 and the resident type in layer 2; $i_3$ is the number of individuals in $V_1$ that have the resident type in layer 1 and the mutant type in layer 2; $i_4$ is the number of individuals in $V_1$ that have the resident type in both layers. We obtain $i_k\in \{0, 1, \ldots, N_1\}$, $\forall k\in \{1, 2, 3, 4\}$ and $i_1+i_2+i_3+i_4=N_1$. Similarly, $j_1$ is the number of individuals in $V_2$ that have the mutant type in both layers; $j_2$ is the number of individuals in $V_2$ that have the mutant type in layer 1 and the resident type in layer 2; $j_3$ is the number of individuals in $V_2$ that have the resident type in layer 1 and the mutant type in layer 2; $j_4$ is the number of individuals in $V_2$ that have the resident type in both layers. We obtain $j_k\in \{0, 1, \ldots, N_2\}$, $\forall k\in \{1, 2, 3, 4\}$ and $j_1+j_2+j_3+j_4=N_2$. There are $\binom{N_1+3}{3}\binom{N_2+3}{3}$ states in total.

Assume that the current state is $(i_1,i_2,i_3,i_4,j_1,j_2,j_3,j_4)$. There are 17 types of events that can occur next.

In the first type of event, an individual that has the mutant type in layer 1 is selected as the parent, which occurs with probability $[2r(i_1+j_1)+(1+r)(i_2+j_2)]/[2r(i_1+j_1)+(1+r)(i_2+i_3+j_2+j_3)+2(i_4+j_4)]$, and the layer 1 is selected with probability $1/2$. Then, we select a neighbor of the parent in layer 1 as the child, and the child that is a node in $V_1$ in layer 2 has the resident type in layer 1 and the mutant type in layer 2, which occurs with probability $i_3/(N-1)$. The state after this event is $(i_1+1,i_2,i_3-1,i_4,j_1,j_2,j_3,j_4)$. Therefore, we obtain
\begin{align}
&p_{(i_1,i_2,i_3,i_4,j_1,j_2,j_3,j_4) \to (i_1+1,i_2,i_3-1,i_4,j_1,j_2,j_3,j_4)}=\nonumber\\ 
&\quad\quad\quad\quad\quad\quad\frac{2r(i_1+j_1)+(1+r)(i_2+j_2)}{2r(i_1+j_1)+(1+r)(i_2+i_3+j_2+j_3)+2(i_4+j_4)}\cdot\frac{1}{2}\cdot\frac{i_3}{N-1}.
\label{eq:p1-complete-bi-state}
\end{align}
The first row of Table~\ref{table-complete-bipartite} (except the header rows) represents this state transition event and Eq.~\eqref{eq:p1-complete-bi-state}.
The remaining 15 types of events are listed in Table~\ref{table-complete-bipartite}. If any other event than the 16 types shown in Table~\ref{table-complete-bipartite} occurs, the state remains unchanged. The probability of this case is 1 minus the sum of all entries in the last column of Table~\ref{table-complete-bipartite}. 

\begin{table}[H]
\hspace*{-2.8em}
\footnotesize
\begin{tabular}{ | w{c}{1.1em} | w{c}{1.1em} | w{c}{1.5em} | w{c}{12em} | w{c}{1.1em} | w{c}{1.1em} | w{c}{1.1em} | w{c}{1.5em} | w{c}{1.3em} | M{7.7em} | M{16.5em} |} 
  \hline
  \multicolumn{4}{|c|}{parent} & \multirow{3}{*}{layer} & \multicolumn{4}{c|}{child} & \multirow{3}{*}{\shortstack{state after\\ transition}} & \multirow{3}{*}{transition probability}\\
  \cline{1-4}
  \cline{6-9}
    \multicolumn{2}{|c|}{type} & \multirow{2}{*}{$V_1$/$V_2$} & \multirow{2}{*}{probability} & & \multicolumn{2}{c|}{type} & \multirow{2}{*}{$V_1$/$V_2$} & \multirow{2}{*}{prob.} & & \\
   \cline{1-2}
   \cline{6-7}      
    $\ell_1$ & $\ell_2$ &  &  &  & $\ell_1$ & $\ell_2$ &  &  &  & \\
  \hline
  \rule{0pt}{13pt}m & m/r & $V_1$/$V_2$ & $\frac{2r(i_1+j_1)+(1+r)(i_2+j_2)}{2r(i_1+j_1)+(1+r)c+2(i_4+j_4)}$ & $\ell_1$ & r & m & $V_1$ & $\frac{i_3}{N-1}$ & $(i_1+1,i_2,i_3-1,i_4,j_1,j_2,j_3,j_4)$ & $\frac{2r(i_1+j_1)+(1+r)(i_2+j_2)}{2r(i_1+j_1)+(1+r)c+2(i_4+j_4)}\cdot\frac{1}{2}\cdot\frac{i_3}{N-1}$  \\[1.8ex] 
  \hline
  \rule{0pt}{13pt}m & m/r & $V_1$/$V_2$ & $\frac{2r(i_1+j_1)+(1+r)(i_2+j_2)}{2r(i_1+j_1)+(1+r)c+2(i_4+j_4)}$ & $\ell_1$ & r & r & $V_1$ & $\frac{i_4}{N-1}$ & $(i_1,i_2+1,i_3,i_4-1,j_1,j_2,j_3,j_4)$ & $\frac{2r(i_1+j_1)+(1+r)(i_2+j_2)}{2r(i_1+j_1)+(1+r)c+2(i_4+j_4)}\cdot\frac{1}{2}\cdot\frac{i_4}{N-1}$  \\[1.8ex] 
  \hline
  \rule{0pt}{13pt}m & m/r & $V_1$/$V_2$ & $\frac{2r(i_1+j_1)+(1+r)(i_2+j_2)}{2r(i_1+j_1)+(1+r)c+2(i_4+j_4)}$ & $\ell_1$ & r & m & $V_2$ & $\frac{j_3}{N-1}$ & $(i_1,i_2,i_3,i_4,j_1+1,j_2,j_3-1,j_4)$ & $\frac{2r(i_1+j_1)+(1+r)(i_2+j_2)}{2r(i_1+j_1)+(1+r)c+2(i_4+j_4)}\cdot\frac{1}{2}\cdot\frac{j_3}{N-1}$ \\[1.8ex] 
  \hline
  \rule{0pt}{13pt}m & m/r & $V_1$/$V_2$ & $\frac{2r(i_1+j_1)+(1+r)(i_2+j_2)}{2r(i_1+j_1)+(1+r)c+2(i_4+j_4)}$ & $\ell_1$ & r & r & $V_2$ & $\frac{j_4}{N-1}$ & $(i_1,i_2,i_3,i_4,j_1,$\par$j_2+1,j_3,j_4-1)$ & $\frac{2r(i_1+j_1)+(1+r)(i_2+j_2)}{2r(i_1+j_1)+(1+r)c+2(i_4+j_4)}\cdot\frac{1}{2}\cdot\frac{j_4}{N-1}$ \\[1.8ex] 
  \hline
  \rule{0pt}{13pt}r & m/r & $V_1$/$V_2$ & $\frac{(1+r)(i_3+j_3)+2(i_4+j_4)}{2r(i_1+j_1)+(1+r)c+2(i_4+j_4)}$ & $\ell_1$ & m & m & $V_1$ & $\frac{i_1}{N-1}$ & $(i_1-1,i_2,i_3+1,i_4,j_1,j_2,j_3,j_4)$ & $\frac{(1+r)(i_3+j_3)+2(i_4+j_4)}{2r(i_1+j_1)+(1+r)c+2(i_4+j_4)}\cdot\frac{1}{2}\cdot\frac{i_1}{N-1}$ \\[1.8ex] 
  \hline
  \rule{0pt}{13pt}r & m/r & $V_1$/$V_2$ & $\frac{(1+r)(i_3+j_3)+2(i_4+j_4)}{2r(i_1+j_1)+(1+r)c+2(i_4+j_4)}$ & $\ell_1$ & m & r & $V_1$ & $\frac{i_2}{N-1}$ & $(i_1,i_2-1,i_3,i_4+1,j_1,j_2,j_3,j_4)$ & $\frac{(1+r)(i_3+j_3)+2(i_4+j_4)}{2r(i_1+j_1)+(1+r)c+2(i_4+j_4)}\cdot\frac{1}{2}\cdot\frac{i_2}{N-1}$ \\[1.8ex] 
  \hline
  \rule{0pt}{13pt}r & m/r & $V_1$/$V_2$ & $\frac{(1+r)(i_3+j_3)+2(i_4+j_4)}{2r(i_1+j_1)+(1+r)c+2(i_4+j_4)}$ & $\ell_1$ & m & m & $V_2$ & $\frac{j_1}{N-1}$ & $(i_1,i_2,i_3,i_4,j_1-1,j_2,j_3+1,j_4)$ & $\frac{(1+r)(i_3+j_3)+2(i_4+j_4)}{2r(i_1+j_1)+(1+r)c+2(i_4+j_4)}\cdot\frac{1}{2}\cdot\frac{j_1}{N-1}$ \\[1.8ex] 
  \hline
  \rule{0pt}{13pt}r & m/r & $V_1$/$V_2$ & $\frac{(1+r)(i_3+j_3)+2(i_4+j_4)}{2r(i_1+j_1)+(1+r)c+2(i_4+j_4)}$ & $\ell_1$ & m & r & $V_2$ & $\frac{j_2}{N-1}$ & $(i_1,i_2,i_3,i_4,j_1,$\par$j_2-1,j_3,j_4+1)$ & $\frac{(1+r)(i_3+j_3)+2(i_4+j_4)}{2r(i_1+j_1)+(1+r)c+2(i_4+j_4)}\cdot\frac{1}{2}\cdot\frac{j_2}{N-1}$ \\[1.8ex]  
  \hline
  \rule{0pt}{13pt}m/r & m & $V_1$ & $\frac{2ri_1+(1+r)i_3}{2r(i_1+j_1)+(1+r)c+2(i_4+j_4)}$ & $\ell_2$ & m & r & $V_2$ & $\frac{j_2}{N_2}$ & $(i_1,i_2,i_3,i_4,j_1+1,j_2-1,j_3,j_4)$ & $\frac{2ri_1+(1+r)i_3}{2r(i_1+j_1)+(1+r)c+2(i_4+j_4)}\cdot\frac{1}{2}\cdot\frac{j_2}{N_2}$ \\[1.8ex]  
  \hline
  \rule{0pt}{13pt}m/r & m & $V_1$ & $\frac{2ri_1+(1+r)i_3}{2r(i_1+j_1)+(1+r)c+2(i_4+j_4)}$ & $\ell_2$ & r & r & $V_2$ & $\frac{j_4}{N_2}$ & $(i_1,i_2,i_3,i_4,j_1,j_2,$\par$j_3+1,j_4-1)$ & $\frac{2ri_1+(1+r)i_3}{2r(i_1+j_1)+(1+r)c+2(i_4+j_4)}\cdot\frac{1}{2}\cdot\frac{j_4}{N_2}$ \\[1.8ex]  
  \hline
  \rule{0pt}{13pt}m/r & m & $V_2$ & $\frac{2rj_1+(1+r)j_3}{2r(i_1+j_1)+(1+r)c+2(i_4+j_4)}$ & $\ell_2$ & m & r & $V_1$ & $\frac{i_2}{N_1}$ & $(i_1+1,i_2-1,$\par$i_3,i_4,j_1,j_2,j_3,j_4)$ & $\frac{2rj_1+(1+r)j_3}{2r(i_1+j_1)+(1+r)c+2(i_4+j_4)}\cdot\frac{1}{2}\cdot\frac{i_2}{N_1}$ \\[1.8ex]  
  \hline
  \rule{0pt}{13pt}m/r & m & $V_2$ & $\frac{2rj_1+(1+r)j_3}{2r(i_1+j_1)+(1+r)c+2(i_4+j_4)}$ & $\ell_2$ & r & r & $V_1$ & $\frac{i_4}{N_1}$ & $(i_1,i_2,i_3+1,i_4-1,j_1,j_2,j_3,j_4)$ & $\frac{2rj_1+(1+r)j_3}{2r(i_1+j_1)+(1+r)c+2(i_4+j_4)}\cdot\frac{1}{2}\cdot\frac{i_4}{N_1}$ \\[1.8ex]  
  \hline
  \rule{0pt}{13pt}m/r & r & $V_1$ & $\frac{(1+r)i_2+2i_4}{2r(i_1+j_1)+(1+r)c+2(i_4+j_4)}$ & $\ell_2$ & m & m & $V_2$ & $\frac{j_1}{N_2}$ & $(i_1,i_2,i_3,i_4,j_1-1,j_2+1,j_3,j_4)$ & $\frac{(1+r)i_2+2i_4}{2r(i_1+j_1)+(1+r)c+2(i_4+j_4)}\cdot\frac{1}{2}\cdot\frac{j_1}{N_2}$ \\[1.8ex]  
  \hline
  \rule{0pt}{13pt}m/r & r & $V_1$ & $\frac{(1+r)i_2+2i_4}{2r(i_1+j_1)+(1+r)c+2(i_4+j_4)}$ & $\ell_2$ & r & m & $V_2$ & $\frac{j_3}{N_2}$ & $(i_1,i_2,i_3,i_4,j_1,j_2,$\par$j_3-1,j_4+1)$ & $\frac{(1+r)i_2+2i_4}{2r(i_1+j_1)+(1+r)c+2(i_4+j_4)}\cdot\frac{1}{2}\cdot\frac{j_3}{N_2}$ \\[1.8ex]  
  \hline
  \rule{0pt}{13pt}m/r & r & $V_2$ & $\frac{(1+r)j_2+2j_4}{2r(i_1+j_1)+(1+r)c+2(i_4+j_4)}$ & $\ell_2$ & m & m & $V_1$ & $\frac{i_1}{N_1}$ & $(i_1-1,i_2+1,$\par$i_3,i_4,j_1,j_2,j_3,j_4)$ & $\frac{(1+r)j_2+2j_4}{2r(i_1+j_1)+(1+r)c+2(i_4+j_4)}\cdot\frac{1}{2}\cdot\frac{i_1}{N_1}$ \\[1.8ex]  
  \hline
  \rule{0pt}{13pt}m/r & r & $V_2$ & $\frac{(1+r)j_2+2j_4}{2r(i_1+j_1)+(1+r)c+2(i_4+j_4)}$ & $\ell_2$ & r & m & $V_1$ & $\frac{i_3}{N_1}$ & $(i_1,i_2,i_3-1,i_4+1,j_1,j_2,j_3,j_4)$ & $\frac{(1+r)j_2+2j_4}{2r(i_1+j_1)+(1+r)c+2(i_4+j_4)}\cdot\frac{1}{2}\cdot\frac{i_3}{N_1}$ \\[1.8ex]  
  \hline
\end{tabular}
\captionof{table}{\label{table-complete-bipartite}Sixteen types of state transition from state $(i_1,i_2,i_3,i_4,j_1,j_2,j_3,j_4)$ under model 1 on the two-layer network composed of a complete graph layer and a complete bipartite graph layer. We set $c=i_2+i_3+j_2+j_3$ to simplify the notation.}
\end{table}

\newpage
\section{Derivation of the transition probability matrix for model 1 on two-layer networks composed of the complete graph and a two-community network\label{complete-community-m1-derivation}}

We use the same 8-tuple as that used in section~\ref{complete-bi-m1-derivation}, i.e., $(i_1,i_2,i_3,i_4,j_1,j_2,j_3,j_4)$, to describe the state of the evolutionary dynamics. Assume that the current state is $(i_1,i_2,i_3,i_4,j_1,j_2,j_3,j_4)$. There are 17 types of events that can occur next.

In the first type of event, an individual that has the mutant type in layer 1 is selected as the parent, which occurs with probability $[2r(i_1+j_1)+(1+r)(i_2+j_2)]/[2r(i_1+j_1)+(1+r)(i_2+i_3+j_2+j_3)+2(i_4+j_4)]$, and the layer 1 is selected with probability $1/2$. Then, we select a neighbor of the parent in layer 1 as the child, and the child that is a node in $V_1$ in layer 2 has the resident type in layer 1 and the mutant type in layer 2, which occurs with probability $i_3/(N-1)$. The state after this event is $(i_1+1,i_2,i_3-1,i_4,j_1,j_2,j_3,j_4)$. Therefore, we obtain
\begin{align}
&p_{(i_1,i_2,i_3,i_4,j_1,j_2,j_3,j_4) \to (i_1+1,i_2,i_3-1,i_4,j_1,j_2,j_3,j_4)}=\nonumber\\ 
&\quad\quad\quad\quad\quad\quad\frac{2r(i_1+j_1)+(1+r)(i_2+j_2)}{2r(i_1+j_1)+(1+r)(i_2+i_3+j_2+j_3)+2(i_4+j_4)}\cdot\frac{1}{2}\cdot\frac{i_3}{N-1}.
\label{eq:p1-complete-community-state}
\end{align}
The first row of Table~\ref{table-complete-community} (except the header rows) represents this state transition event and Eq.~\eqref{eq:p1-complete-community-state}.
The remaining 15 types of the table represent the other types of state transition. 
In Table~\ref{table-complete-community}, some types of events have two rows because both cases result in the same state after the state transition.
If any other event than the 16 types shown in Table~\ref{table-complete-community} occurs, the state remains unchanged. The probability of this case is 1 minus the sum of all entries in the last column of Table~\ref{table-complete-community}.

\begin{table}[H]
\hspace*{-4.3em}
\footnotesize
\begin{tabular}{ | w{c}{1.1em} | w{c}{1.1em} | w{c}{1.5em} | w{c}{12em} | w{c}{1.1em} | w{c}{1.1em} | w{c}{1.1em} | w{c}{1.5em} | w{c}{4em} | M{7.3em} | M{17.5em} |} 
  \hline
  \multicolumn{4}{|c|}{parent} & \multirow{3}{*}{layer} & \multicolumn{4}{c|}{child} & \multirow{3}{*}{\shortstack{state after\\ transition}} & \multirow{3}{*}{transition probability}\\
  \cline{1-4}
  \cline{6-9}
    \multicolumn{2}{|c|}{type} & \multirow{2}{*}{$V_1$/$V_2$} & \multirow{2}{*}{probability} & & \multicolumn{2}{c|}{type} & \multirow{2}{*}{$V_1$/$V_2$} & \multirow{2}{*}{prob.} & & \\
   \cline{1-2}
   \cline{6-7}      
    $\ell_1$ & $\ell_2$ &  &  &  & $\ell_1$ & $\ell_2$ &  &  &  & \\
  \hline
  \rule{0pt}{14pt}m & m/r & $V_1$/$V_2$ & $\frac{2r(i_1+j_1)+(1+r)(i_2+j_2)}{2r(i_1+j_1)+(1+r)c+2(i_4+j_4)}$ & $\ell_1$ & r & m & $V_1$ & $\frac{i_3}{N-1}$ & $(i_1+1,i_2,i_3-1,i_4,j_1,j_2,j_3,j_4)$ & $\frac{2r(i_1+j_1)+(1+r)(i_2+j_2)}{2r(i_1+j_1)+(1+r)c+2(i_4+j_4)}\cdot\frac{1}{2}\cdot\frac{i_3}{N-1}$  \\[2ex] 
  \hline
  \rule{0pt}{14pt}m & m/r & $V_1$/$V_2$ & $\frac{2r(i_1+j_1)+(1+r)(i_2+j_2)}{2r(i_1+j_1)+(1+r)c+2(i_4+j_4)}$ & $\ell_1$ & r & r & $V_1$ & $\frac{i_4}{N-1}$ & $(i_1,i_2+1,i_3,i_4-1,j_1,j_2,j_3,j_4)$ & $\frac{2r(i_1+j_1)+(1+r)(i_2+j_2)}{2r(i_1+j_1)+(1+r)c+2(i_4+j_4)}\cdot\frac{1}{2}\cdot\frac{i_4}{N-1}$  \\[2ex] 
  \hline
  \rule{0pt}{14pt}m & m/r & $V_1$/$V_2$ & $\frac{2r(i_1+j_1)+(1+r)(i_2+j_2)}{2r(i_1+j_1)+(1+r)c+2(i_4+j_4)}$ & $\ell_1$ & r & m & $V_2$ & $\frac{j_3}{N-1}$ & $(i_1,i_2,i_3,i_4,j_1+1,j_2,j_3-1,j_4)$ & $\frac{2r(i_1+j_1)+(1+r)(i_2+j_2)}{2r(i_1+j_1)+(1+r)c+2(i_4+j_4)}\cdot\frac{1}{2}\cdot\frac{j_3}{N-1}$ \\[2ex] 
  \hline
  \rule{0pt}{14pt}m & m/r & $V_1$/$V_2$ & $\frac{2r(i_1+j_1)+(1+r)(i_2+j_2)}{2r(i_1+j_1)+(1+r)c+2(i_4+j_4)}$ & $\ell_1$ & r & r & $V_2$ & $\frac{j_4}{N-1}$ & $(i_1,i_2,i_3,i_4,j_1,$\par$j_2+1,j_3,j_4-1)$ & $\frac{2r(i_1+j_1)+(1+r)(i_2+j_2)}{2r(i_1+j_1)+(1+r)c+2(i_4+j_4)}\cdot\frac{1}{2}\cdot\frac{j_4}{N-1}$ \\[2ex] 
  \hline
  \rule{0pt}{14pt}r & m/r & $V_1$/$V_2$ & $\frac{(1+r)(i_3+j_3)+2(i_4+j_4)}{2r(i_1+j_1)+(1+r)c+2(i_4+j_4)}$ & $\ell_1$ & m & m & $V_1$ & $\frac{i_1}{N-1}$ & $(i_1-1,i_2,i_3+1,i_4,j_1,j_2,j_3,j_4)$ & $\frac{(1+r)(i_3+j_3)+2(i_4+j_4)}{2r(i_1+j_1)+(1+r)c+2(i_4+j_4)}\cdot\frac{1}{2}\cdot\frac{i_1}{N-1}$ \\[2ex] 
  \hline
  \rule{0pt}{14pt}r & m/r & $V_1$/$V_2$ & $\frac{(1+r)(i_3+j_3)+2(i_4+j_4)}{2r(i_1+j_1)+(1+r)c+2(i_4+j_4)}$ & $\ell_1$ & m & r & $V_1$ & $\frac{i_2}{N-1}$ & $(i_1,i_2-1,i_3,i_4+1,j_1,j_2,j_3,j_4)$ & $\frac{(1+r)(i_3+j_3)+2(i_4+j_4)}{2r(i_1+j_1)+(1+r)c+2(i_4+j_4)}\cdot\frac{1}{2}\cdot\frac{i_2}{N-1}$ \\[2ex] 
  \hline
  \rule{0pt}{14pt}r & m/r & $V_1$/$V_2$ & $\frac{(1+r)(i_3+j_3)+2(i_4+j_4)}{2r(i_1+j_1)+(1+r)c+2(i_4+j_4)}$ & $\ell_1$ & m & m & $V_2$ & $\frac{j_1}{N-1}$ & $(i_1,i_2,i_3,i_4,j_1-1,j_2,j_3+1,j_4)$ & $\frac{(1+r)(i_3+j_3)+2(i_4+j_4)}{2r(i_1+j_1)+(1+r)c+2(i_4+j_4)}\cdot\frac{1}{2}\cdot\frac{j_1}{N-1}$ \\[2ex] 
  \hline
  \rule{0pt}{14pt}r & m/r & $V_1$/$V_2$ & $\frac{(1+r)(i_3+j_3)+2(i_4+j_4)}{2r(i_1+j_1)+(1+r)c+2(i_4+j_4)}$ & $\ell_1$ & m & r & $V_2$ & $\frac{j_2}{N-1}$ & $(i_1,i_2,i_3,i_4,j_1,$\par$j_2-1,j_3,j_4+1)$ & $\frac{(1+r)(i_3+j_3)+2(i_4+j_4)}{2r(i_1+j_1)+(1+r)c+2(i_4+j_4)}\cdot\frac{1}{2}\cdot\frac{j_2}{N-1}$ \\[2ex]  
  \hline
  \rule{0pt}{14pt}m/r & m & $V_1$ & $\frac{2ri_1+(1+r)i_3}{2r(i_1+j_1)+(1+r)c+2(i_4+j_4)}$ & \multirow{2}[5]{*}{$\ell_2$} & m & r & $V_2$ & $\frac{\overline{\epsilon} j_2}{N_1-1+\overline{\epsilon} N_2}$ & \multirow{2}[5]{*}{\shortstack{$(i_1,i_2,i_3,i_4,j_1+$\\$1,j_2-1,j_3,j_4)$}} & \multirow[b]{2}{*}{\shortstack{$\frac{[2ri_1+(1+r)i_3]\overline{\epsilon} j_2}{[2r(i_1+j_1)+(1+r)c+2(i_4+j_4)]2(N_1-1+\overline{\epsilon} N_2)}$\\$+$\\$\frac{[2rj_1+(1+r)j_3]j_2}{[2r(i_1+j_1)+(1+r)c+2(i_4+j_4)]2(N_2-1+\overline{\epsilon} N_1)}$}} \\[2ex]
  \cline{1-4}
  \cline{6-9}
  \rule{0pt}{14pt}m/r & m & $V_2$ & $\frac{2rj_1+(1+r)j_3}{2r(i_1+j_1)+(1+r)c+2(i_4+j_4)}$ &  & m & r & $V_2$ & $\frac{j_2}{N_2-1+\overline{\epsilon} N_1}$ &  &  \\[2ex]
  \hline
  \rule{0pt}{14pt}m/r & m & $V_1$ & $\frac{2ri_1+(1+r)i_3}{2r(i_1+j_1)+(1+r)c+2(i_4+j_4)}$ & \multirow{2}[5]{*}{$\ell_2$} & r & r & $V_2$ & $\frac{\overline{\epsilon} j_4}{N_1-1+\overline{\epsilon} N_2}$ & \multirow{2}[5]{*}{\shortstack{$(i_1,i_2,i_3,i_4,j_1,$\\$j_2,j_3+1,j_4-1)$}} & \multirow[b]{2}{*}{\shortstack{$\frac{[2ri_1+(1+r)i_3]\overline{\epsilon} j_4}{[2r(i_1+j_1)+(1+r)c+2(i_4+j_4)]2(N_1-1+\overline{\epsilon} N_2)}$\\$+$\\$\frac{[2rj_1+(1+r)j_3]j_4}{[2r(i_1+j_1)+(1+r)c+2(i_4+j_4)]2(N_2-1+\overline{\epsilon} N_1)}$}} \\[2ex]
  \cline{1-4}
  \cline{6-9}
  \rule{0pt}{14pt}m/r & m & $V_2$ & $\frac{2rj_1+(1+r)j_3}{2r(i_1+j_1)+(1+r)c+2(i_4+j_4)}$ &  & r & r & $V_2$ & $\frac{j_4}{N_2-1+\overline{\epsilon} N_1}$ &  &   \\[2ex]  
  \hline
  \rule{0pt}{14pt}m/r & m & $V_1$ & $\frac{2ri_1+(1+r)i_3}{2r(i_1+j_1)+(1+r)c+2(i_4+j_4)}$ & \multirow{2}[5]{*}{$\ell_2$} & m & r & $V_1$ & $\frac{i_2}{N_1-1+\overline{\epsilon} N_2}$ & \multirow{2}[5]{*}{\shortstack{$(i_1+1,i_2-1,i_3,$\\$i_4,j_1,j_2,j_3,j_4)$}} & \multirow[b]{2}{*}{\shortstack{$\frac{[2ri_1+(1+r)i_3]i_2}{[2r(i_1+j_1)+(1+r)c+2(i_4+j_4)]2(N_1-1+\overline{\epsilon} N_2)}$\\$+$\\$\frac{[2rj_1+(1+r)j_3]\overline{\epsilon} i_2}{[2r(i_1+j_1)+(1+r)c+2(i_4+j_4)]2(N_2-1+\overline{\epsilon} N_1)}$}} \\[2ex]
  \cline{1-4}
  \cline{6-9}
  \rule{0pt}{14pt}m/r & m & $V_2$ & $\frac{2rj_1+(1+r)j_3}{2r(i_1+j_1)+(1+r)c+2(i_4+j_4)}$ &  & m & r & $V_1$ & $\frac{\overline{\epsilon} i_2}{N_2-1+\overline{\epsilon} N_1}$ &  &   \\[2ex]  
  \hline
  \rule{0pt}{14pt}m/r & m & $V_1$ & $\frac{2ri_1+(1+r)i_3}{2r(i_1+j_1)+(1+r)c+2(i_4+j_4)}$ & \multirow{2}[5]{*}{$\ell_2$} & r & r & $V_1$ & $\frac{i_4}{N_1-1+\overline{\epsilon} N_2}$ & \multirow{2}[5]{*}{\shortstack{$(i_1,i_2,i_3+1,i_4$\\$-1,j_1,j_2,j_3,j_4)$}} & \multirow[b]{2}{*}{\shortstack{$\frac{[2ri_1+(1+r)i_3]i_4}{[2r(i_1+j_1)+(1+r)c+2(i_4+j_4)]2(N_1-1+\overline{\epsilon} N_2)}$\\$+$\\$\frac{[2rj_1+(1+r)j_3]\overline{\epsilon} i_4}{[2r(i_1+j_1)+(1+r)c+2(i_4+j_4)]2(N_2-1+\overline{\epsilon} N_1)}$}} \\[2ex]
  \cline{1-4}
  \cline{6-9}
  \rule{0pt}{14pt}m/r & m & $V_2$ & $\frac{2rj_1+(1+r)j_3}{2r(i_1+j_1)+(1+r)c+2(i_4+j_4)}$ &  & r & r & $V_1$ & $\frac{\overline{\epsilon} i_4}{N_2-1+\overline{\epsilon} N_1}$ &  &   \\[2ex]  
  \hline
   \rule{0pt}{14pt}m/r & r & $V_1$ & $\frac{(1+r)i_2+2i_4}{2r(i_1+j_1)+(1+r)c+2(i_4+j_4)}$ & \multirow{2}[5]{*}{$\ell_2$} & m & m & $V_2$ & $\frac{\overline{\epsilon} j_1}{N_1-1+\overline{\epsilon} N_2}$ & \multirow{2}[5]{*}{\shortstack{$(i_1,i_2,i_3,i_4,j_1$\\$-1,j_2+1,j_3,j_4)$}} & \multirow[b]{2}{*}{\shortstack{$\frac{[(1+r)i_2+2i_4]\overline{\epsilon} j_1}{[2r(i_1+j_1)+(1+r)c+2(i_4+j_4)]2(N_1-1+\overline{\epsilon} N_2)}$\\$+$\\$\frac{[(1+r)j_2+2j_4]j_1}{[2r(i_1+j_1)+(1+r)c+2(i_4+j_4)]2(N_2-1+\overline{\epsilon} N_1)}$}} \\[2ex]
  \cline{1-4}
  \cline{6-9}
  \rule{0pt}{14pt}m/r & r & $V_2$ & $\frac{(1+r)j_2+2j_4}{2r(i_1+j_1)+(1+r)c+2(i_4+j_4)}$ &  & m & m & $V_2$ & $\frac{j_1}{N_2-1+\overline{\epsilon} N_1}$ &  &   \\[2ex] 
  \hline
\end{tabular}
\captionof{table}{\label{table-complete-community}Sixteen types of state transition from state $(i_1,i_2,i_3,i_4,j_1,j_2,j_3,j_4)$ under model 1 on the two-layer network composed of a complete graph layer and a two-community network layer. We set $c=i_2+i_3+j_2+j_3$ to simplify the notation.}
\end{table}

\begin{table}[H]
\hspace*{-4.3em}
\footnotesize
\begin{tabular}{ | w{c}{1.1em} | w{c}{1.1em} | w{c}{1.5em} | w{c}{12em} | w{c}{1.1em} | w{c}{1.1em} | w{c}{1.1em} | w{c}{1.5em} | w{c}{4em} | M{7.3em} | M{17.5em} |} 
  \hline
  \multicolumn{4}{|c|}{parent} & \multirow{3}{*}{layer} & \multicolumn{4}{c|}{child} & \multirow{3}{*}{\shortstack{state after\\ transition}} & \multirow{3}{*}{transition probability}\\
  \cline{1-4}
  \cline{6-9}
    \multicolumn{2}{|c|}{type} & \multirow{2}{*}{$V_1$/$V_2$} & \multirow{2}{*}{probability} & & \multicolumn{2}{c|}{type} & \multirow{2}{*}{$V_1$/$V_2$} & \multirow{2}{*}{prob.} & & \\
   \cline{1-2}
   \cline{6-7}      
    $\ell_1$ & $\ell_2$ &  &  &  & $\ell_1$ & $\ell_2$ &  &  &  & \\
  \hline
  \rule{0pt}{14pt}m/r & r & $V_1$ & $\frac{(1+r)i_2+2i_4}{2r(i_1+j_1)+(1+r)c+2(i_4+j_4)}$ & \multirow{2}[5]{*}{$\ell_2$} & r & m & $V_2$ & $\frac{\overline{\epsilon} j_3}{N_1-1+\overline{\epsilon} N_2}$ & \multirow{2}[5]{*}{\shortstack{$(i_1,i_2,i_3,i_4,j_1$\\$,j_2,j_3-1,j_4+1)$}} & \multirow[b]{2}{*}{\shortstack{$\frac{[(1+r)i_2+2i_4]\overline{\epsilon} j_3}{[2r(i_1+j_1)+(1+r)c+2(i_4+j_4)]2(N_1-1+\overline{\epsilon} N_2)}$\\$+$\\$\frac{[(1+r)j_2+2j_4]j_3}{[2r(i_1+j_1)+(1+r)c+2(i_4+j_4)]2(N_2-1+\overline{\epsilon} N_1)}$}} \\[2ex]
  \cline{1-4}
  \cline{6-9}
  \rule{0pt}{14pt}m/r & r & $V_2$ & $\frac{(1+r)j_2+2j_4}{2r(i_1+j_1)+(1+r)c+2(i_4+j_4)}$ &  & r & m & $V_2$ & $\frac{j_3}{N_2-1+\overline{\epsilon} N_1}$ &  &   \\[2ex]  
  \hline
  \rule{0pt}{14pt}m/r & r & $V_1$ & $\frac{(1+r)i_2+2i_4}{2r(i_1+j_1)+(1+r)c+2(i_4+j_4)}$ & \multirow{2}[5]{*}{$\ell_2$} & m & m & $V_1$ & $\frac{i_1}{N_1-1+\overline{\epsilon} N_2}$ & \multirow{2}[5]{*}{\shortstack{$(i_1-1,i_2+1,i_3,$\\$i_4,j_1,j_2,j_3,j_4)$}} & \multirow[b]{2}{*}{\shortstack{$\frac{[(1+r)i_2+2i_4]i_1}{[2r(i_1+j_1)+(1+r)c+2(i_4+j_4)]2(N_1-1+\overline{\epsilon} N_2)}$\\$+$\\$\frac{[(1+r)j_2+2j_4]\overline{\epsilon} i_1}{[2r(i_1+j_1)+(1+r)c+2(i_4+j_4)]2(N_2-1+\overline{\epsilon} N_1)}$}} \\[2ex]
  \cline{1-4}
  \cline{6-9}
  \rule{0pt}{14pt}m/r & r & $V_2$ & $\frac{(1+r)j_2+2j_4}{2r(i_1+j_1)+(1+r)c+2(i_4+j_4)}$ &  & m & m & $V_1$ & $\frac{\overline{\epsilon} i_1}{N_2-1+\overline{\epsilon} N_1}$ &  &   \\[2ex]  
  \hline
  \rule{0pt}{14pt}m/r & r & $V_1$ & $\frac{(1+r)i_2+2i_4}{2r(i_1+j_1)+(1+r)c+2(i_4+j_4)}$ & \multirow{2}[5]{*}{$\ell_2$} & r & m & $V_1$ & $\frac{i_3}{N_1-1+\overline{\epsilon} N_2}$ & \multirow{2}[5]{*}{\shortstack{$(i_1,i_2,i_3-1,i_4$\\$+1,j_1,j_2,j_3,j_4)$}} & \multirow[b]{2}{*}{\shortstack{$\frac{[(1+r)i_2+2i_4]i_3}{[2r(i_1+j_1)+(1+r)c+2(i_4+j_4)]2(N_1-1+\overline{\epsilon} N_2)}$\\$+$\\$\frac{[(1+r)j_2+2j_4]\overline{\epsilon} i_3}{[2r(i_1+j_1)+(1+r)c+2(i_4+j_4)]2(N_2-1+\overline{\epsilon} N_1)}$}} \\[2ex]
  \cline{1-4}
  \cline{6-9}
  \rule{0pt}{14pt}m/r & r & $V_2$ & $\frac{(1+r)j_2+2j_4}{2r(i_1+j_1)+(1+r)c+2(i_4+j_4)}$ &  & r & m & $V_1$ & $\frac{\overline{\epsilon} i_3}{N_2-1+\overline{\epsilon} N_1}$ &  &   \\[2ex]  
  \hline
\end{tabular}
\end{table}

\newpage
\section{Derivation of the transition probability matrix for the variant of model 1 with the dB updating rule on the coupled complete graph \label{complete-db-m1-derivation}}

Assume that the current state is $\bm{i}=(i_1, i_2, i_3, i_4)$, where $i_1$, $i_2$, $i_3$, and $i_4$ are as defined in section~5.2. There are nine types of events that can occur next.

In the first type of event, an individual that has the resident type in layer 1 and the mutant type in layer 2 is selected to die, which occurs with probability $i_3/N$, and the layer 1 is selected with probability $1/2$. Then, we select a neighbor of the dying individual in layer 1, the selected neighbor has the mutant type in layer 1, and it occupies the empty site, which occurs with probability $[2ri_1+(1+r)i_2]/[2ri_1+(1+r)(i_2+i_3-1)+2i_4]$. The state after this event is $(i_1+1, i_2, i_3-1, i_4)$. Therefore, we obtain
\begin{equation}
p_{(i_1,i_2,i_3,i_4) \to (i_1+1, i_2, i_3-1, i_4)} = \frac{i_3}{N}\cdot\frac{1}{2}\cdot\frac{2ri_1+(1+r)i_2}{2ri_1+(1+r)(i_2+i_3-1)+2i_4}.
\label{eq:p1-complete-state-db}
\end{equation}
The first row of Table~\ref{table-complete-db} (except the header rows) represents this state transition event and Eq.~\eqref{eq:p1-complete-state-db}.
The remaining seven rows of the table represent the other types of state transition. If any other event than the eight types shown in Table~\ref{table-complete-db} occurs, the state remains unchanged. The probability of this case is 1 minus the sum of all entries in the last column of Table~\ref{table-complete-db}. 

\begin{table}[H]
\hspace*{-2.7em}
\begin{tabular}{ | w{c}{1.5em} | w{c}{1.5em} | w{c}{1.5em} | w{c}{1.5em} | c | c | c | c | M{12.5em} |} 
  \hline
  \multicolumn{3}{|c|}{dying individual} & \multirow{3}{*}{\shortstack{sele-\\cted\\ layer}} & \multicolumn{3}{c|}{reproducing individual} & \multirow{3}{*}{state after transition} & \multirow{3}{*}{transition probability}\\
  \cline{1-3}
  \cline{5-7}
    \multicolumn{2}{|c|}{type} & \multirow{2}{*}{prob.} & & \multicolumn{2}{c|}{type} & \multirow{2}{*}{probability} & & \\
   \cline{1-2}
   \cline{5-6}      
    $\ell_1$ & $\ell_2$ &  &  & $\ell_1$ & $\ell_2$ &  &  & \\
  \hline
  \rule{0pt}{14pt}r & m & $\frac{i_3}{N}$ & $\ell_1$ & m & m/r & $\frac{2ri_1+(1+r)i_2}{2ri_1+(1+r)(i_2+i_3-1)+2i_4}$ & $(i_1+1,i_2,i_3-1,i_4)$ & $\frac{i_3}{N}\cdot\frac{1}{2}\cdot\frac{2ri_1+(1+r)i_2}{2ri_1+(1+r)(i_2+i_3-1)+2i_4}$  \\[1.2ex] 
  \hline
  \rule{0pt}{14pt}r & r & $\frac{i_4}{N}$ & $\ell_1$ & m & m/r & $\frac{2ri_1+(1+r)i_2}{2ri_1+(1+r)(i_2+i_3)+2(i_4-1)}$ & $(i_1,i_2+1,i_3,i_4-1)$ & $\frac{i_4}{N}\cdot\frac{1}{2}\cdot\frac{2ri_1+(1+r)i_2}{2ri_1+(1+r)(i_2+i_3)+2(i_4-1)}$  \\[1.2ex]
  \hline
  \rule{0pt}{14pt}m & m & $\frac{i_1}{N}$ & $\ell_1$ & r & m/r & $\frac{(1+r)i_3+2i_4}{2r(i_1-1)+(1+r)(i_2+i_3)+2i_4}$ & $(i_1-1,i_2,i_3+1,i_4)$ & $\frac{i_1}{N}\cdot\frac{1}{2}\cdot\frac{(1+r)i_3+2i_4}{2r(i_1-1)+(1+r)(i_2+i_3)+2i_4}$  \\[1.2ex] 
  \hline
  \rule{0pt}{14pt}m & r & $\frac{i_2}{N}$ & $\ell_1$ & r & m/r & $\frac{(1+r)i_3+2i_4}{2ri_1+(1+r)(i_2+i_3-1)+2i_4}$ & $(i_1,i_2-1,i_3,i_4+1)$ & $\frac{i_2}{N}\cdot\frac{1}{2}\cdot\frac{(1+r)i_3+2i_4}{2ri_1+(1+r)(i_2+i_3-1)+2i_4}$  \\[1.2ex] 
  \hline
  \rule{0pt}{14pt}m & r & $\frac{i_2}{N}$ & $\ell_2$ & m/r & m & $\frac{2ri_1+(1+r)i_3}{2ri_1+(1+r)(i_2+i_3-1)+2i_4}$ & $(i_1+1,i_2-1,i_3,i_4)$ & $\frac{i_2}{N}\cdot\frac{1}{2}\cdot\frac{2ri_1+(1+r)i_3}{2ri_1+(1+r)(i_2+i_3-1)+2i_4}$  \\[1.2ex] 
  \hline
  \rule{0pt}{14pt}r & r & $\frac{i_4}{N}$ & $\ell_2$ & m/r & m & $\frac{2ri_1+(1+r)i_3}{2ri_1+(1+r)(i_2+i_3)+2(i_4-1)}$ & $(i_1,i_2,i_3+1,i_4-1)$ & $\frac{i_4}{N}\cdot\frac{1}{2}\cdot\frac{2ri_1+(1+r)i_3}{2ri_1+(1+r)(i_2+i_3)+2(i_4-1)}$  \\[1.2ex] 
  \hline
  \rule{0pt}{14pt}m & m & $\frac{i_1}{N}$ & $\ell_2$ & m/r & r & $\frac{(1+r)i_2+2i_4}{2r(i_1-1)+(1+r)(i_2+i_3)+2i_4}$ & $(i_1-1,i_2+1,i_3,i_4)$ & $\frac{i_1}{N}\cdot\frac{1}{2}\cdot\frac{(1+r)i_2+2i_4}{2r(i_1-1)+(1+r)(i_2+i_3)+2i_4}$  \\[1.2ex] 
  \hline
  \rule{0pt}{14pt}r & m & $\frac{i_3}{N}$ & $\ell_2$ & m/r & r & $\frac{(1+r)i_2+2i_4}{2ri_1+(1+r)(i_2+i_3-1)+2i_4}$ & $(i_1,i_2,i_3-1,i_4+1)$ & $\frac{i_3}{N}\cdot\frac{1}{2}\cdot\frac{(1+r)i_2+2i_4}{2ri_1+(1+r)(i_2+i_3-1)+2i_4}$  \\[1.2ex] 
  \hline
\end{tabular}
\captionof{table}{\label{table-complete-db}Eight types of state transition from state $(i_1,i_2,i_3,i_4)$ under model 1 and the dB updating rule on the coupled complete graphs.}
\end{table}

\newpage
\section{Derivation of the transition probability matrix for model 2 on the coupled complete graph\label{complete-m2-derivation}}

Assume that the current state is $\bm{i}=(i_1, i_2, i_3, i_4)$. There are 21 types of events that can occur next.

In the first type of event, the state changes to $(i_1+1,i_2,i_3,i_4-1)$.
This state change can occur in one of the following three manners.

In the first case, an individual that has the mutant type in layer 1 is selected as the parent, which occurs with probability $[2ri_1+(1+r)i_2]/[2ri_1+(1+r)(i_2+i_3)+2i_4]$. Then, we select a neighbor of the parent in layer 1, denoted by $v$, as the child, and $v$ has the resident type in both layers, which occurs with probability $i_4/(N-1)$. At the same time, an individual that has the mutant type in layer 2 is selected as the parent, which occurs with probability $[2ri_1+(1+r)i_3]/[2ri_1+(1+r)(i_2+i_3)+2i_4]$. Then, we select $v$ in layer 2, which occurs with probability $1/(N-1)$. 

In the second case, an individual that has the mutant type in layer 1 is selected as the parent. Then, we select a neighbor of the parent in layer 1 as the child, and the child has the resident type in layer 1 and the mutant type in layer 2, which occurs with probability $i_3/(N-1)$. At the same time, an individual that has the mutant type in layer 2 is selected as the parent. Then, we select a neighbor of the parent in layer 2 as the child, and the child has the resident type in both layers, which occurs with probability $i_4/(N-1)$. 

In the third case, an individual that has the mutant type in layer 1 is selected as the parent. Then, we select a neighbor of the parent in layer 1 as the child, and the child has the resident type in both layers, which occurs with probability $i_4/(N-1)$. At the same time, an individual that has the mutant type in layer 2 is selected as the parent. Then, we select a neighbor of the parent in layer 2 as the child, and the child has the mutant type in layer 1 and the resident type in layer 2, which occurs with probability $i_2/(N-1)$. 

In all the three cases, the state after this event is $(i_1+1,i_2,i_3,i_4-1)$. Therefore, we obtain
\begin{align}
p_{(i_1,i_2,i_3,i_4)\to (i_1+1,i_2,i_3,i_4-1)} = & \frac{2ri_1+(1+r)i_2}{2ri_1+(1+r)(i_2+i_3)+2i_4}\cdot\frac{i_4}{N-1}\cdot\frac{2ri_1+(1+r)i_3}{2ri_1+(1+r)(i_2+i_3)+2i_4}\cdot\frac{1}{N-1}\nonumber\\
+&\frac{2ri_1+(1+r)i_2}{2ri_1+(1+r)(i_2+i_3)+2i_4}\cdot\frac{i_3}{N-1}\cdot\frac{2ri_1+(1+r)i_3}{2ri_1+(1+r)(i_2+i_3)+2i_4}\cdot\frac{i_4}{N-1}\nonumber\\
+&\frac{2ri_1+(1+r)i_2}{2ri_1+(1+r)(i_2+i_3)+2i_4}\cdot\frac{i_4}{N-1}\cdot\frac{2ri_1+(1+r)i_3}{2ri_1+(1+r)(i_2+i_3)+2i_4}\cdot\frac{i_2}{N-1}\nonumber\\
=&\frac{[2ri_1+(1+r)i_2][2ri_1+(1+r)i_3]i_4(i_2+i_3+1)}{[2ri_1+(1+r)(i_2+i_3)+2i_4]^2(N-1)^2}.
\label{eq:p1-complete-m2-state}
\end{align}

The first row of Tables~\ref{table-complete-m2-layer1} and~\ref{table-complete-m2-layer2}  (except the header rows) represent this state transition event in layer 1 and layer 2, respectively, and Eq.~\eqref{eq:p1-complete-m2-state}. The remaining nineteen rows of each table represent the other types of state transition. In Tables~\ref{table-complete-m2-layer1} and~\ref{table-complete-m2-layer2}, some types of events have two or three rows because these cases result in the same state after the state transition.
We split each type of state transition into two tables because of limited space. If any other event than the twenty types shown in Tables~\ref{table-complete-m2-layer1} and~\ref{table-complete-m2-layer2} occurs, the state remains unchanged. The probability of this case is 1 minus the sum of all entries in the last column of Table~\ref{table-complete-m2-layer2}. 

\begin{table}[H]
\begin{tabular}{ | c | c | c | c | c | c | M{10em} |} 
  \hline
  \multicolumn{3}{|c|}{parent} & \multicolumn{3}{c|}{child} & \multirow{3}{*}{state after transition} \\
  \cline{1-3}
  \cline{4-6}
    \multicolumn{2}{|c|}{type} & \multirow{2}{*}{probability} & \multicolumn{2}{c|}{type} & \multirow{2}{*}{prob.} & \\
   \cline{1-2}
   \cline{4-5}      
    $\ell_1$ & $\ell_2$ &  & $\ell_1$ & $\ell_2$ &  & \\
  \hline
  \rule{0pt}{11pt}\multirow{3}[3]{*}{m} & \multirow{3}[3]{*}{m/r} & \multirow{3}[3]{*}{$\frac{2ri_1+(1+r)i_2}{2ri_1+(1+r)(i_2+i_3)+2i_4}$} & r & r & $\frac{i_4}{N-1}$ & \multirow{3}[3]{*}{$(i_1+1,i_2,i_3,i_4-1)$} \\[0.5ex] 
  \cline{4-6}
  \rule{0pt}{11pt}& & & r & m & $\frac{i_3}{N-1}$ & \\[0.5ex]
  \cline{4-6}
  \rule{0pt}{11pt}& & & r & r & $\frac{i_4}{N-1}$ & \\[0.5ex]
  \cline{4-6}
  \hline
  \rule{0pt}{13pt}m & m/r & $\frac{2ri_1+(1+r)i_2}{2ri_1+(1+r)(i_2+i_3)+2i_4}$ & r & r & $\frac{i_4}{N-1}$ & $(i_1,i_2+1,i_3+1,i_4-2)$\\[1ex]
  \hline
  \rule{0pt}{13pt}m & m/r & $\frac{2ri_1+(1+r)i_2}{2ri_1+(1+r)(i_2+i_3)+2i_4}$ & r & m & $\frac{i_3}{N-1}$ & $(i_1+2,i_2-1,i_3-1,i_4)$\\[1ex]
  \hline
  \rule{0pt}{11pt}\multirow{3}[3]{*}{m} & \multirow{3}[3]{*}{m/r} & \multirow{3}[3]{*}{$\frac{2ri_1+(1+r)i_2}{2ri_1+(1+r)(i_2+i_3)+2i_4}$} & r & m & $\frac{i_3}{N-1}$ & \multirow{3}[3]{*}{$(i_1,i_2+1,i_3-1,i_4)$} \\[0.5ex] 
  \cline{4-6}
  \rule{0pt}{11pt}& & & r & m & $\frac{i_3}{N-1}$ & \\[0.5ex]
  \cline{4-6}
  \rule{0pt}{11pt}& & & r & r & $\frac{i_4}{N-1}$ & \\[0.5ex]
  \cline{4-6}
  \hline
  \rule{0pt}{13pt}m & m/r & $\frac{2ri_1+(1+r)i_2}{2ri_1+(1+r)(i_2+i_3)+2i_4}$ & r & m & $\frac{i_3}{N-1}$ & $(i_1+1,i_2,i_3-2,i_4+1)$\\[1ex]
  \hline
  \rule{0pt}{13pt}m & m/r & $\frac{2ri_1+(1+r)i_2}{2ri_1+(1+r)(i_2+i_3)+2i_4}$ & r & r & $\frac{i_4}{N-1}$ & $(i_1-1,i_2+2,i_3,i_4-1)$\\[1ex]
  \hline
  \rule{0pt}{11pt}\multirow{2}[2]{*}{m} & \multirow{2}[2]{*}{m/r} & \multirow{2}[2]{*}{$\frac{2ri_1+(1+r)i_2}{2ri_1+(1+r)(i_2+i_3)+2i_4}$} & r & m & $\frac{i_3}{N-1}$ & \multirow{2}[2]{*}{$(i_1+1,i_2,i_3-1,i_4)$} \\[0.5ex] 
  \cline{4-6}
  \rule{0pt}{11pt}& & & r & m & $\frac{i_3}{N-1}$ & \\[0.5ex]
  \cline{4-6}
  \hline
  \rule{0pt}{11pt}\multirow{2}[2]{*}{m} & \multirow{2}[2]{*}{m/r} & \multirow{2}[2]{*}{$\frac{2ri_1+(1+r)i_2}{2ri_1+(1+r)(i_2+i_3)+2i_4}$} & r & r & $\frac{i_4}{N-1}$ & \multirow{2}[2]{*}{$(i_1,i_2+1,i_3,i_4-1)$} \\[0.5ex] 
  \cline{4-6}
  \rule{0pt}{11pt}& & & r & r & $\frac{i_4}{N-1}$ & \\[0.5ex]
  \cline{4-6}
  \hline
  \rule{0pt}{11pt}\multirow{3}[3]{*}{r} & \multirow{3}[3]{*}{m/r} & \multirow{3}[3]{*}{$\frac{(1+r)i_3+2i_4}{2ri_1+(1+r)(i_2+i_3)+2i_4}$} & m & m & $\frac{i_1}{N-1}$ & \multirow{3}[3]{*}{$(i_1,i_2-1,i_3+1,i_4)$} \\[0.5ex] 
  \cline{4-6}
  \rule{0pt}{11pt}& & & m & r & $\frac{i_2}{N-1}$ & \\[0.5ex]
  \cline{4-6}
  \rule{0pt}{11pt}& & & m & r & $\frac{i_2}{N-1}$ & \\[0.5ex]
  \cline{4-6}
  \hline
  \rule{0pt}{13pt}r & m/r & $\frac{(1+r)i_3+2i_4}{2ri_1+(1+r)(i_2+i_3)+2i_4}$ & m & m & $\frac{i_1}{N-1}$ & $(i_1-1,i_2,i_3+2,i_4-1)$\\[1ex]
  \hline
  \rule{0pt}{13pt}r & m/r & $\frac{(1+r)i_3+2i_4}{2ri_1+(1+r)(i_2+i_3)+2i_4}$ & m & r & $\frac{i_2}{N-1}$ & $(i_1+1,i_2-2,i_3,i_4+1)$\\[1ex]
  \hline
  \rule{0pt}{11pt}\multirow{3}[3]{*}{r} & \multirow{3}[3]{*}{m/r} & \multirow{3}[3]{*}{$\frac{(1+r)i_3+2i_4}{2ri_1+(1+r)(i_2+i_3)+2i_4}$} & m & m & $\frac{i_1}{N-1}$ & \multirow{3}[3]{*}{$(i_1-1,i_2,i_3,i_4+1)$} \\[0.5ex] 
  \cline{4-6}
  \rule{0pt}{11pt}& & & m & m & $\frac{i_1}{N-1}$ & \\[0.5ex]
  \cline{4-6}
  \rule{0pt}{11pt}& & & m & r & $\frac{i_2}{N-1}$ & \\[0.5ex]
  \cline{4-6}
  \hline
  \rule{0pt}{13pt}r & m/r & $\frac{(1+r)i_3+2i_4}{2ri_1+(1+r)(i_2+i_3)+2i_4}$ & m & m & $\frac{i_1}{N-1}$ & $(i_1-2,i_2+1,i_3+1,i_4)$\\[1ex]
  \hline
  \rule{0pt}{13pt}r & m/r & $\frac{(1+r)i_3+2i_4}{2ri_1+(1+r)(i_2+i_3)+2i_4}$ & m & r & $\frac{i_2}{N-1}$ & $(i_1,i_2-1,i_3-1,i_4+2)$\\[1ex]
  \hline
  \rule{0pt}{11pt}\multirow{2}[2]{*}{r} & \multirow{2}[2]{*}{m/r} & \multirow{2}[2]{*}{$\frac{(1+r)i_3+2i_4}{2ri_1+(1+r)(i_2+i_3)+2i_4}$} & m & m & $\frac{i_1}{N-1}$ & \multirow{2}[2]{*}{$(i_1-1,i_2,i_3+1,i_4)$} \\[0.5ex] 
  \cline{4-6}
  \rule{0pt}{11pt}& & & m & m & $\frac{i_1}{N-1}$ & \\[0.5ex]
  \cline{4-6}
  \hline
  \rule{0pt}{11pt}\multirow{2}[2]{*}{r} & \multirow{2}[2]{*}{m/r} & \multirow{2}[2]{*}{$\frac{(1+r)i_3+2i_4}{2ri_1+(1+r)(i_2+i_3)+2i_4}$} & m & r & $\frac{i_2}{N-1}$ & \multirow{2}[2]{*}{$(i_1,i_2-1,i_3,i_4+1)$} \\[0.5ex] 
  \cline{4-6}
  \rule{0pt}{11pt}& & & m & r & $\frac{i_2}{N-1}$ & \\[0.5ex]
  \cline{4-6}
  \hline
  \rule{0pt}{13pt}m & m/r & $\frac{2ri_1+(1+r)i_2}{2ri_1+(1+r)(i_2+i_3)+2i_4}$ & m & m/r & $\frac{i_1+i_2-1}{N-1}$ & \multirow{2}[3]{*}{$(i_1+1,i_2-1,i_3,i_4)$} \\[1ex] 
  \cline{1-3}
  \cline{4-6}
  \rule{0pt}{13pt}r & m/r & $\frac{(1+r)i_3+2i_4}{2ri_1+(1+r)(i_2+i_3)+2i_4}$ & r & m/r & $\frac{i_3+i_4-1}{N-1}$ & \\[1ex]
  \cline{1-3}
  \cline{4-6}
  \hline
\end{tabular}
\captionof{table}{\label{table-complete-m2-layer1}Twenty types of state transition in $\ell_1$ from state $(i_1,i_2,i_3,i_4)$ for model 2 on the coupled complete graph.}
\end{table}

\begin{table}[H]
\begin{tabular}{ | c | c | c | c | c | c | M{10em} |} 
  \hline
  \multicolumn{3}{|c|}{parent} & \multicolumn{3}{c|}{child} & \multirow{3}{*}{state after transition} \\
  \cline{1-3}
  \cline{4-6}
    \multicolumn{2}{|c|}{type} & \multirow{2}{*}{probability} & \multicolumn{2}{c|}{type} & \multirow{2}{*}{prob.} & \\
   \cline{1-2}
   \cline{4-5}      
    $\ell_1$ & $\ell_2$ &  & $\ell_1$ & $\ell_2$ &  & \\
  \hline
  \rule{0pt}{13pt}m & m/r & $\frac{2ri_1+(1+r)i_2}{2ri_1+(1+r)(i_2+i_3)+2i_4}$ & m & m/r & $\frac{i_1+i_2-1}{N-1}$ & \multirow{2}[3]{*}{$(i_1,i_2,i_3+1,i_4-1)$} \\[1ex] 
  \cline{1-3}
  \cline{4-6}
  \rule{0pt}{13pt}r & m/r & $\frac{(1+r)i_3+2i_4}{2ri_1+(1+r)(i_2+i_3)+2i_4}$ & r & m/r & $\frac{i_3+i_4-1}{N-1}$ & \\[1ex]
  \cline{1-3}
  \cline{4-6}
  \hline
  \rule{0pt}{13pt}m & m/r & $\frac{2ri_1+(1+r)i_2}{2ri_1+(1+r)(i_2+i_3)+2i_4}$ & m & m/r & $\frac{i_1+i_2-1}{N-1}$ & \multirow{2}[3]{*}{$(i_1-1,i_2+1,i_3,i_4)$} \\[1ex] 
  \cline{1-3}
  \cline{4-6}
  \rule{0pt}{13pt}r & m/r & $\frac{(1+r)i_3+2i_4}{2ri_1+(1+r)(i_2+i_3)+2i_4}$ & r & m/r & $\frac{i_3+i_4-1}{N-1}$ & \\[1ex]
  \cline{1-3}
  \cline{4-6}
  \hline
  \rule{0pt}{13pt}m & m/r & $\frac{2ri_1+(1+r)i_2}{2ri_1+(1+r)(i_2+i_3)+2i_4}$ & m & m/r & $\frac{i_1+i_2-1}{N-1}$ & \multirow{2}[3]{*}{$(i_1,i_2,i_3-1,i_4+1)$} \\[1ex] 
  \cline{1-3}
  \cline{4-6}
  \rule{0pt}{13pt}r & m/r & $\frac{(1+r)i_3+2i_4}{2ri_1+(1+r)(i_2+i_3)+2i_4}$ & r & m/r & $\frac{i_3+i_4-1}{N-1}$ & \\[1ex]
  \cline{1-3}
  \cline{4-6}
  \hline
\end{tabular}
\end{table}

\begin{table}[H]
\begin{tabular}{ | c | c | c | c | c | c | c | M{15em} |} 
  \hline
  \multicolumn{3}{|c|}{parent} & \multicolumn{4}{c|}{child} & \multirow{3}{*}{transition probability} \\
  \cline{1-3}
  \cline{4-7}
    \multicolumn{2}{|c|}{type} & \multirow{2}{*}{probability} & \multicolumn{2}{c|}{type} & \multirow{2}{*}{same as $\ell_1$} & \multirow{2}{*}{prob.} & \\
   \cline{1-2}
   \cline{4-5}      
    $\ell_1$ & $\ell_2$ &  & $\ell_1$ & $\ell_2$ &  &  & \\
  \hline
  \rule{0pt}{11pt}\multirow{3}[3]{*}{m/r} & \multirow{3}[3]{*}{m} & \multirow{3}[3]{*}{$\frac{2ri_1+(1+r)i_3}{2ri_1+(1+r)(i_2+i_3)+2i_4}$} & r & r & Y & $\frac{1}{N-1}$ & \multirow{3}[3]{*}{$\frac{[2ri_1+(1+r)i_2][2ri_1+(1+r)i_3]i_4(i_2+i_3+1)}{[2ri_1+(1+r)(i_2+i_3)+2i_4]^2(N-1)^2}$} \\[0.5ex] 
  \cline{4-7}
  \rule{0pt}{11pt}& & & r & r & N & $\frac{i_4}{N-1}$ & \\[0.5ex]
  \cline{4-7}
  \rule{0pt}{11pt}& & & m & r & N & $\frac{i_2}{N-1}$ & \\[0.5ex]
  \cline{4-7}
  \hline
  \rule{0pt}{13pt}m/r & m & $\frac{2ri_1+(1+r)i_3}{2ri_1+(1+r)(i_2+i_3)+2i_4}$ & r & r & N & $\frac{i_4-1}{N-1}$ & $\frac{[2ri_1+(1+r)i_2][2ri_1+(1+r)i_3]i_4(i_4-1)}{[2ri_1+(1+r)(i_2+i_3)+2i_4]^2(N-1)^2}$\\[1ex]
  \hline
  \rule{0pt}{13pt}m/r & m & $\frac{2ri_1+(1+r)i_3}{2ri_1+(1+r)(i_2+i_3)+2i_4}$ & m & r & N & $\frac{i_2}{N-1}$ & $\frac{[2ri_1+(1+r)i_2][2ri_1+(1+r)i_3]i_2i_3}{[2ri_1+(1+r)(i_2+i_3)+2i_4]^2(N-1)^2}$\\[1ex]
  \hline
  \rule{0pt}{11pt}\multirow{3}[3]{*}{m/r} & \multirow{3}[3]{*}{r} & \multirow{3}[3]{*}{$\frac{(1+r)i_2+2i_4}{2ri_1+(1+r)(i_2+i_3)+2i_4}$} & m & m & N & $\frac{i_1}{N-1}$ & \multirow{3}[3]{*}{$\frac{[2ri_1+(1+r)i_2][(1+r)i_2+2i_4]i_3(i_1+i_4+1)}{[2ri_1+(1+r)(i_2+i_3)+2i_4]^2(N-1)^2}$} \\[0.5ex] 
  \cline{4-7}
  \rule{0pt}{11pt}& & & r & m & Y & $\frac{1}{N-1}$ & \\[0.5ex]
  \cline{4-7}
  \rule{0pt}{11pt}& & & r & m & N & $\frac{i_3}{N-1}$ & \\[0.5ex]
  \cline{4-7}
  \hline
  \rule{0pt}{13pt}m/r & r & $\frac{(1+r)i_2+2i_4}{2ri_1+(1+r)(i_2+i_3)+2i_4}$ & r & m & N & $\frac{i_3-1}{N-1}$ & $\frac{[2ri_1+(1+r)i_2][(1+r)i_2+2i_4]i_3(i_3-1)}{[2ri_1+(1+r)(i_2+i_3)+2i_4]^2(N-1)^2}$\\[1ex]
  \hline
  \rule{0pt}{13pt}m/r & r & $\frac{(1+r)i_2+2i_4}{2ri_1+(1+r)(i_2+i_3)+2i_4}$ & m & m & N & $\frac{i_1}{N-1}$ & $\frac{[2ri_1+(1+r)i_2][(1+r)i_2+2i_4]i_1i_4}{[2ri_1+(1+r)(i_2+i_3)+2i_4]^2(N-1)^2}$\\[1ex]
  \hline
  \rule{0pt}{16pt}m/r & m & $\frac{2ri_1+(1+r)i_3}{2ri_1+(1+r)(i_2+i_3)+2i_4}$ & m/r & m & Y/N & $\frac{i_1+i_3-1}{N-1}$ & \multirow{2}{*}{\shortstack{$\frac{[2ri_1+(1+r)i_2]i_3}{[2ri_1+(1+r)(i_2+i_3)+2i_4]^2(N-1)^2}\times$\\$\{[2ri_1+(1+r)i_3](i_1+i_3-1)+$\\$[(1+r)i_2+2i_4](i_2+i_4-1)\}$}} \\[1.8ex] 
  \cline{1-3}
  \cline{4-7}
  \rule{0pt}{16pt}m/r & r & $\frac{(1+r)i_2+2i_4}{2ri_1+(1+r)(i_2+i_3)+2i_4}$ & m/r & r & N & $\frac{i_2+i_4-1}{N-1}$ & \\[1.8ex]
  \cline{1-3}
  \cline{4-7}
  \hline
  \rule{0pt}{16pt}m/r & m & $\frac{2ri_1+(1+r)i_3}{2ri_1+(1+r)(i_2+i_3)+2i_4}$ & m/r & m & N & $\frac{i_1+i_3-1}{N-1}$ & \multirow{2}{*}{\shortstack{$\frac{[2ri_1+(1+r)i_2]i_4}{[2ri_1+(1+r)(i_2+i_3)+2i_4]^2(N-1)^2}\times$\\$\{[2ri_1+(1+r)i_3](i_1+i_3-1)+$\\$[(1+r)i_2+2i_4](i_2+i_4-1)\}$}} \\[1.8ex] 
  \cline{1-3}
  \cline{4-7}
  \rule{0pt}{16pt}m/r & r & $\frac{(1+r)i_2+2i_4}{2ri_1+(1+r)(i_2+i_3)+2i_4}$ & m/r & r & Y/N & $\frac{i_2+i_4-1}{N-1}$ & \\[1.8ex]
  \cline{1-3}
  \cline{4-7}
  \hline
  \rule{0pt}{11pt}\multirow{3}[3]{*}{m/r} & \multirow{3}[3]{*}{m} & \multirow{3}[3]{*}{$\frac{2ri_1+(1+r)i_3}{2ri_1+(1+r)(i_2+i_3)+2i_4}$} & m & r & N & $\frac{i_2}{N-1}$ & \multirow{3}[3]{*}{$\frac{[(1+r)i_3+2i_4][2ri_1+(1+r)i_3]i_2(i_1+i_4+1)}{[2ri_1+(1+r)(i_2+i_3)+2i_4]^2(N-1)^2}$} \\[0.5ex] 
  \cline{4-7}
  \rule{0pt}{11pt}& & & m & r & Y & $\frac{1}{N-1}$ & \\[0.5ex] 
  \cline{4-7}
  \rule{0pt}{11pt}& & & r & r & N & $\frac{i_4}{N-1}$ & \\[0.5ex] 
  \cline{4-7}
  \hline
  \rule{0pt}{13pt}m/r & m & $\frac{2ri_1+(1+r)i_3}{2ri_1+(1+r)(i_2+i_3)+2i_4}$ & r & r & N & $\frac{i_4}{N-1}$ & $\frac{[(1+r)i_3+2i_4][2ri_1+(1+r)i_3]i_1i_4}{[2ri_1+(1+r)(i_2+i_3)+2i_4]^2(N-1)^2}$\\[1ex]
  \hline
  \rule{0pt}{13pt}m/r & m & $\frac{2ri_1+(1+r)i_3}{2ri_1+(1+r)(i_2+i_3)+2i_4}$ & m & r & N & $\frac{i_2-1}{N-1}$ & $\frac{[(1+r)i_3+2i_4][2ri_1+(1+r)i_3]i_2(i_2-1)}{[2ri_1+(1+r)(i_2+i_3)+2i_4]^2(N-1)^2}$\\[1ex]
  \hline
  \rule{0pt}{11pt}\multirow{3}[3]{*}{m/r} & \multirow{3}[3]{*}{r} & \multirow{3}[3]{*}{$\frac{(1+r)i_2+2i_4}{2ri_1+(1+r)(i_2+i_3)+2i_4}$} & m & m & Y & $\frac{1}{N-1}$ & \multirow{3}[3]{*}{$\frac{[(1+r)i_3+2i_4][(1+r)i_2+2i_4]i_1(i_2+i_3+1)}{[2ri_1+(1+r)(i_2+i_3)+2i_4]^2(N-1)^2}$} \\[0.5ex] 
  \cline{4-7}
  \rule{0pt}{11pt}& & & r & m & N & $\frac{i_3}{N-1}$ & \\[0.5ex] 
  \cline{4-7}
  \rule{0pt}{11pt}& & & m & m & N & $\frac{i_1}{N-1}$ & \\[0.5ex] 
  \cline{4-7}
  \hline
  \rule{0pt}{13pt}m/r & r & $\frac{(1+r)i_2+2i_4}{2ri_1+(1+r)(i_2+i_3)+2i_4}$ & m & m & N & $\frac{i_1-1}{N-1}$ & $\frac{[(1+r)i_3+2i_4][(1+r)i_2+2i_4]i_1(i_1-1)}{[2ri_1+(1+r)(i_2+i_3)+2i_4]^2(N-1)^2}$\\[1ex]
  \hline
  \rule{0pt}{13pt}m/r & r & $\frac{(1+r)i_2+2i_4}{2ri_1+(1+r)(i_2+i_3)+2i_4}$ & r & m & N & $\frac{i_3}{N-1}$ & $\frac{[(1+r)i_3+2i_4][(1+r)i_2+2i_4]i_2i_3}{[2ri_1+(1+r)(i_2+i_3)+2i_4]^2(N-1)^2}$\\[1ex]
  \hline
  \rule{0pt}{16pt}m/r & m & $\frac{2ri_1+(1+r)i_3}{2ri_1+(1+r)(i_2+i_3)+2i_4}$ & m/r & m & Y/N & $\frac{i_1+i_3-1}{N-1}$ & \multirow{2}{*}{\shortstack{$\frac{[(1+r)i_3+2i_4]i_1}{[2ri_1+(1+r)(i_2+i_3)+2i_4]^2(N-1)^2}\times$\\$\{[2ri_1+(1+r)i_3](i_1+i_3-1)+$\\$[(1+r)i_2+2i_4](i_2+i_4-1)\}$}} \\[1.8ex] 
  \cline{1-3}
  \cline{4-7}
  \rule{0pt}{16pt}m/r & r & $\frac{(1+r)i_2+2i_4}{2ri_1+(1+r)(i_2+i_3)+2i_4}$ & m/r & r & N & $\frac{i_2+i_4-1}{N-1}$ & \\[1.8ex] 
  \cline{1-3}
  \cline{4-7}
  \hline
  \end{tabular}
\captionof{table}{\label{table-complete-m2-layer2}Twenty types of state transition in $\ell_2$ from state $(i_1,i_2,i_3,i_4)$ for model 2 on the coupled complete graph.}
\end{table}

\begin{table}[H]
\begin{tabular}{ | c | c | c | c | c | c | c | M{15em} |} 
  \hline
  \multicolumn{3}{|c|}{parent} & \multicolumn{4}{c|}{child} & \multirow{3}{*}{transition probability} \\
  \cline{1-3}
  \cline{4-7}
    \multicolumn{2}{|c|}{type} & \multirow{2}{*}{probability} & \multicolumn{2}{c|}{type} & \multirow{2}{*}{same as $\ell_1$} & \multirow{2}{*}{prob.} & \\
   \cline{1-2}
   \cline{4-5}      
    $\ell_1$ & $\ell_2$ &  & $\ell_1$ & $\ell_2$ &  &  & \\
\hline
  \rule{0pt}{16pt}m/r & m & $\frac{2ri_1+(1+r)i_3}{2ri_1+(1+r)(i_2+i_3)+2i_4}$ & m/r & m & N & $\frac{i_1+i_3-1}{N-1}$ & \multirow{2}{*}{\shortstack{$\frac{[(1+r)i_3+2i_4]i_2}{[2ri_1+(1+r)(i_2+i_3)+2i_4]^2(N-1)^2}\times$\\$\{[2ri_1+(1+r)i_3](i_1+i_3-1)+$\\$[(1+r)i_2+2i_4](i_2+i_4-1)\}$}} \\[1.8ex]  
  \cline{1-3}
  \cline{4-7}
  \rule{0pt}{16pt}m/r & r & $\frac{(1+r)i_2+2i_4}{2ri_1+(1+r)(i_2+i_3)+2i_4}$ & m/r & r & Y/N & $\frac{i_2+i_4-1}{N-1}$ & \\[1.8ex] 
  \cline{1-3}
  \cline{4-7}
  \hline
  \rule{0pt}{16pt}\multirow{2}[5]{*}{m/r} & \multirow{2}[5]{*}{m} & \multirow{2}[5]{*}{$\frac{2ri_1+(1+r)i_3}{2ri_1+(1+r)(i_2+i_3)+2i_4}$} & m & r & Y/N & $\frac{i_2}{N-1}$ & \multirow{2}{*}{\shortstack{$\frac{[2ri_1+(1+r)i_3]i_2}{[2ri_1+(1+r)(i_2+i_3)+2i_4]^2(N-1)^2}\times$\\$\{[2ri_1+(1+r)i_2](i_1+i_2-1)+$\\$[(1+r)i_3+2i_4](i_3+i_4-1)\}$}} \\[1.8ex]  
  \cline{4-7}
  \rule{0pt}{16pt}& & & m & r & N & $\frac{i_2}{N-1}$ & \\[1.8ex] 
  \cline{4-7}    
  \hline
\rule{0pt}{16pt}\multirow{2}[5]{*}{m/r} & \multirow{2}[5]{*}{m} & \multirow{2}[5]{*}{$\frac{2ri_1+(1+r)i_3}{2ri_1+(1+r)(i_2+i_3)+2i_4}$} & r & r & N & $\frac{i_4}{N-1}$ & \multirow{2}{*}{\shortstack{$\frac{[2ri_1+(1+r)i_3]i_4}{[2ri_1+(1+r)(i_2+i_3)+2i_4]^2(N-1)^2}\times$\\$\{[2ri_1+(1+r)i_2](i_1+i_2-1)+$\\$[(1+r)i_3+2i_4](i_3+i_4-1)\}$}} \\ [1.8ex] 
  \cline{4-7}
  \rule{0pt}{16pt}& & & r & r & Y/N & $\frac{i_4}{N-1}$ & \\[1.8ex] 
  \cline{4-7}
  \hline
  \rule{0pt}{16pt}\multirow{2}[5]{*}{m/r} & \multirow{2}[5]{*}{r} & \multirow{2}[5]{*}{$\frac{(1+r)i_2+2i_4}{2ri_1+(1+r)(i_2+i_3)+2i_4}$} & m & m & Y/N & $\frac{i_1}{N-1}$ & \multirow{2}{*}{\shortstack{$\frac{[2ri_1+(1+r)i_3]i_4}{[2ri_1+(1+r)(i_2+i_3)+2i_4]^2(N-1)^2}\times$\\$\{[2ri_1+(1+r)i_2](i_1+i_2-1)+$\\$[(1+r)i_3+2i_4](i_3+i_4-1)\}$}} \\[1.8ex]  
  \cline{4-7}
  \rule{0pt}{16pt}& & & m & m & N & $\frac{i_1}{N-1}$ & \\[1.8ex] 
  \cline{4-7}
  \hline
  \rule{0pt}{16pt}\multirow{2}[5]{*}{m/r} & \multirow{2}[5]{*}{r} & \multirow{2}[5]{*}{$\frac{(1+r)i_2+2i_4}{2ri_1+(1+r)(i_2+i_3)+2i_4}$} & r & m & N & $\frac{i_3}{N-1}$ & \multirow{2}{*}{\shortstack{$\frac{[(1+r)i_2+2i_4]i_3}{[2ri_1+(1+r)(i_2+i_3)+2i_4]^2(N-1)^2}\times$\\$\{[2ri_1+(1+r)i_2](i_1+i_2-1)+$\\$[(1+r)i_3+2i_4](i_3+i_4-1)\}$}} \\[1.8ex]  
  \cline{4-7}
  \rule{0pt}{16pt}& & & r & m & Y/N & $\frac{i_3}{N-1}$ & \\[1.8ex] 
  \cline{4-7}
  \hline
\end{tabular}
\end{table}

\end{document}